\documentclass[11pt]{article}
\usepackage{amsmath,amsfonts,amsthm,amssymb,color}
\usepackage{thmtools}
\usepackage{thm-restate}

\usepackage{fancybox}
\usepackage{multirow}
\usepackage{nicefrac}
\usepackage{hyperref}

\usepackage{ifthen}

\newboolean{@full}
\setboolean{@full}{true}

\newcommand{\iffull}{\ifthenelse{\boolean{@full}}}

\newtheorem{theorem}{Theorem}[section]
\newtheorem{lemma}[theorem]{Lemma}

\newtheorem{corollary}[theorem]{Corollary}

\newtheorem{proposition}[theorem]{Proposition}
\newtheorem{claim}[theorem]{Claim}
\newtheorem{fact}[theorem]{Fact}

\theoremstyle{definition}
\newtheorem{definition}[theorem]{Definition}

\newenvironment{fminipage}%
  {\begin{Sbox}\begin{minipage}}%
  {\end{minipage}\end{Sbox}\fbox{\TheSbox}}

\iffull{
	\usepackage[margin=1in]{geometry}

	\newenvironment{algbox}[0]{\vskip 0.2in
	\noindent 
	\begin{fminipage}{6.3in}
	}{
	\end{fminipage}
	\vskip 0.2in
	}

	\expandafter\def\expandafter\normalsize\expandafter{%
	    \normalsize
	    \setlength\belowdisplayskip{5pt}
	    \setlength\belowdisplayshortskip{5pt}
	}

	\usepackage{titlesec}
	\titlespacing*{\section}{0pt}{0.7\baselineskip}{0.4\baselineskip}
	\titlespacing*{\subsection}{0pt}{0.7\baselineskip}{0.4\baselineskip}
	
}{
	\newenvironment{algbox}[0]{
	\noindent 
	\begin{fminipage}{6.3in}
	}{
	\end{fminipage}
	}

	\usepackage[margin=1in]{geometry}
	\usepackage{paralist}
	\renewenvironment{itemize}[1]{\begin{compactitem}#1}{\end{compactitem}}
	\renewenvironment{enumerate}[1]{\begin{compactenum}#1}{\end{compactenum}}
	
	\setlength{\abovecaptionskip}{5pt}
	\setlength{\belowcaptionskip}{-15pt}
	
	\usepackage{titlesec}
	\titlespacing*{\section}{0pt}{0.2\baselineskip}{0.2\baselineskip}
	\titlespacing*{\subsection}{0pt}{0.2\baselineskip}{0.2\baselineskip}
	\titlespacing*{\theorem}{0pt}{0\baselineskip}{0\baselineskip}
	\titlespacing*{\lemma}{0pt}{0\baselineskip}{0\baselineskip}

	\expandafter\def\expandafter\normalsize\expandafter{%
	    \normalsize
	    \setlength\abovedisplayskip{5pt}
	    \setlength\belowdisplayskip{5pt}
	    \setlength\abovedisplayshortskip{5pt}
	    \setlength\belowdisplayshortskip{5pt}
	}

}

\def\pleq{\preccurlyeq}
\def\pgeq{\succcurlyeq}

\def\bvec#1{{\mbox{\boldmath $#1$}}}

\def\prob#1#2{\mbox{Pr}_{#1}\left[ #2 \right]}

\def\expec#1#2{{\mathbb{E}}_{#1}\left[ #2 \right]}

\def\defeq{\stackrel{\mathrm{def}}{=}}
\def\setof#1{\left\{#1  \right\}}
\def\sizeof#1{\left|#1  \right|}

\def\floor#1{\left\lfloor #1 \right\rfloor}
\def\ceil#1{\left\lceil #1 \right\rceil}

\def\union{\cup}

\def\abs#1{\left|#1  \right|}

\def\norm#1{\left\| #1 \right\|}

\newcommand\ttau{\boldsymbol{\tau}}

\newcommand\eepsilon{\boldsymbol{\epsilon}}

\def\aa{\pmb{\mathit{a}}}
\newcommand\bb{\boldsymbol{\mathit{b}}}

\newcommand\dd{\boldsymbol{\mathit{d}}}

\newcommand\ff{\boldsymbol{\mathit{f}}}
\renewcommand\gg{\boldsymbol{\mathit{g}}}

\newcommand\rr{\boldsymbol{\mathit{r}}}

\newcommand\uu{\boldsymbol{\mathit{u}}}
\newcommand\vv{\boldsymbol{\mathit{v}}}
\newcommand\ww{\boldsymbol{\mathit{w}}}

\newcommand\xx{\boldsymbol{\mathit{x}}}

\renewcommand\AA{\boldsymbol{\mathit{A}}}
\newcommand\BB{\boldsymbol{\mathit{B}}}
\newcommand\BBtil{\boldsymbol{\tilde{\mathit{B}}}}
\newcommand\CC{\boldsymbol{\mathit{C}}}
\newcommand\CChat{\boldsymbol{\widehat{\mathit{C}}}}
\newcommand\DD{\boldsymbol{\mathit{D}}}

\newcommand\FF{\boldsymbol{\mathit{F}}}
\newcommand\GG{\boldsymbol{\mathit{G}}}
\newcommand\II{\boldsymbol{\mathit{I}}}

\newcommand\MM{\boldsymbol{\mathit{M}}}

\newcommand\Mtil{\boldsymbol{\widetilde{\mathit{M}}}}

\newcommand\QQ{\boldsymbol{\mathit{Q}}}
\newcommand\LL{\boldsymbol{\mathit{L}}}

\newcommand\RR{\boldsymbol{\mathit{R}}}
\renewcommand\SS{\boldsymbol{\mathit{S}}}
\newcommand\TT{\boldsymbol{\mathit{T}}}
\newcommand\UU{\boldsymbol{\mathit{U}}}
\newcommand\WW{\boldsymbol{\mathit{W}}}
\newcommand\VV{\boldsymbol{\mathit{V}}}
\newcommand\XX{\boldsymbol{\mathit{X}}}
\newcommand\YY{\boldsymbol{\mathit{Y}}}

\newcommand\ZZ{\boldsymbol{\mathit{Z}}}

\newcommand\MMhat{\boldsymbol{\widehat{\mathit{M}}}}
\newcommand\MMtil{\boldsymbol{\widetilde{\mathit{M}}}}

\newcommand\UUhat{\boldsymbol{\widehat{\mathit{U}}}}

\newcommand\xxtil{\boldsymbol{\tilde{\mathit{x}}}}

\newcommand\Ghat{{\widehat{{G}}}}
\newcommand\Vhat{{\widehat{{V}}}}
\newcommand\Ehat{{\widehat{{E}}}}
\newcommand\ddhat{{\hat{\dd}}}
\newcommand\dhat{{\hat{{d}}}}
\newcommand\nhat{{\hat{{n}}}}
\def\Gtil{\widetilde{G}}

\newcommand{\nfrac}{\nicefrac}

\usepackage{dsfont}

\newcommand{\schur}[2]{Sc \left(#1,  #2\right) }

\newcommand{\bdd}{{bDD}}

\newcommand{\abdd}{{$\alpha$-\bdd}}

\newcommand{\blk}[2]{\ensuremath{{#1}_{[#2]}}}

\newcommand{\id}{\ensuremath{\mathbb{I}}}
\newcommand{\dg}{\ensuremath{*}}
\newcommand{\complex}{\mathds{C}}
\DeclareMathOperator*{\diagop}{Diag}

\allowdisplaybreaks
\usepackage{thm-restate}

\begin{document}

\title{
  Sparsified Cholesky and Multigrid Solvers\\
for Connection Laplacians\thanks{
This paper incorporates and improves upon results previously announced by a subset
  of the authors in \cite{LeePengSpielman}.
}
}
\author{
Rasmus Kyng\thanks{Supported by 
   NSF grant CCF-1111257.}\\
Yale University\\
rasmus.kyng@yale.edu \\
\and
Yin Tat Lee
\thanks{Supported in part by NSF awards 0843915 and 1111109.  Part of this work was done while visiting the Simons Institute for the Theory of Computing, UC Berkeley.}
\\
M.I.T.\\
yintat@mit.edu
\and
Richard Peng\\
Georgia Tech\\
rpeng@cc.gatech.edu
\and
Sushant Sachdeva\thanks{Supported by a Simons Investigator Award to Daniel A. Spielman.}\\
Yale University\\
sachdeva@cs.yale.edu
\and
Daniel A. Spielman\thanks{Supported by AFOSR Award FA9550-12-1-0175,
   NSF grant CCF-1111257, a Simons Investigator Award to Daniel A. Spielman, and a MacArthur Fellowship.}
\\ 
Yale University\\
spielman@cs.yale.edu
}

\maketitle

\begin{abstract}
We introduce
  the \textit{sparsified Cholesky} and \textit{sparsified multigrid} algorithms
  for solving systems of linear equations.
These algorithms accelerate Gaussian elimination by sparsifying the nonzero
  matrix entries created by the elimination process.

We use these new algorithms to derive the first nearly linear time algorithms
  for solving systems of equations in connection Laplacians---a generalization of
  Laplacian matrices that arise in many problems in image and signal processing.

We also prove that every connection Laplacian has a linear sized approximate inverse.
This is an LU factorization with a linear number of nonzero entries that is a strong
  approximation of the original matrix.
Using such a factorization one can solve systems of equations in a connection Laplacian
  in linear time.
Such a factorization was unknown even for ordinary graph Laplacians.

\end{abstract}
\thispagestyle{empty}

\newpage

\setcounter{page}{1}

\section{Introduction}

We introduce
  the \textit{sparsified Cholesky} and \textit{sparsified multigrid} algorithms
  for solving systems of linear equations.
Two advantages of these algorithms over other recently introduced 
  nearly-linear time algorithms
  for solving systems of equations in Laplacian matrices
  \cite{Vaidya,SpielmanTengLinsolve,KMP1,KMP2,KOSZ,CohenKMPPRX} are:
\begin{enumerate}
\item [1.] They give nearly-linear time algorithms for solving systems of equations in a much broader class of matrices---the 
  connection Laplacians and Hermitian block diagonally dominant matrices.
Connection Laplacians \cite{singer2012vector,connection} are a generalization of graph Laplacians that arise
  in many applications, including celebrated work on cryo-electron microscopy
  \cite{singer2011three,shkolnisky2012viewing,zhao2014rotationally}, phase retrieval \cite{alexeev2014phase,marchesini2014alternating},
  and many image processing problems (e.g. \cite{stable2015,arie2012global}).
Previous algorithms for solving systems of equations in graph Laplacians cannot be extended to
  solve equations in connection Laplacians because the previous algorithms
  relied on some form of low stretch spanning trees, a concept
  that has no analog for the more general connection Laplacians.
\item [2.] They provide linear-sized approximate inverses of connection Laplacian matrices.
That is, for every $n$-dimensional connection Laplacian $\MM$, the sparsified Cholesky factorization algorithm
  produces an block-upper-triangular matrix $\UU$ and block diagonal $\DD$ with $O (n)$ nonzero entries
  so that $\UU^{T}\DD  \UU$ is a constant-factor approximation of $\MM$.
  Such matrices $\UU$ and $\DD$ allow one to solve systems of equations in $\MM$ to accuracy
  $\epsilon$ in time $O (m \log \epsilon^{-1})$, where $m$ is the number of nonzero entries of $\MM$.
 Even for ordinary Laplacian matrices, the existence of such approximate inverses is entirely new.
 The one caveat of this result is that we do not yet know how to
    compute those approximate inverses in nearly linear time.
\end{enumerate}

The sparsified Cholesky and sparsified multigrid algorithms work by sparsifying
  the matrices produced during Gaussian elimination.
Recall that Cholesky factorization is the version of Gaussian
elimination 
  for symmetric matrices, and that the high running time of Gaussian elimination comes from
  \textit{fill}---new nonzero matrix entries that are created by row operations.
If Gaussian elimination never produced rows with a super-constant number of entries,
  then it would run in linear time.
The sparsified Cholesky algorithm accelerates Gaussian elimination by sparsifying
  the rows that are produced by elimination, thereby guaranteeing
  that the elimination will be fast.
Sparsified Cholesky is inspired by one of the major advances in algorithms for solving
  linear equations in  Laplacian
  matrices---the Incomplete Cholesky factorization (ICC) \cite{ICC}.
However, ICC merely drops entries produced by elimination,
  whereas sparsification also carefully increases ones that remain.
The difference is crucial, and is why ICC does not provide a nearly-linear
  time solver.

To control the error introduced by sparsification, we have to be careful not to do it
  too often.
This means that our algorithm actually chooses a large set of rows and columns to
  eliminate at once, and then sparsifies the result.
This is the basis of our first algorithms, which establish the existence of linear time
  solvers and linear sized approximate inverses, after precomputation.
This precomputation is analogous to the procedure of computing a matrix inverse:
  the approximate inverse takes much less time to apply than to compute.

To produce entire algorithms that run in nearly linear time (both to compute and apply the approximate inverse)
  requires a little more work.
To avoid the work of computing the matrix obtained by eliminating the large set of rows and columns,
  we design a fast algorithm for approximating it quickly.
This resulting algorithm produces a solver routine that does a little more work 
  at each level, and so resembles a multigrid V-cycle \cite{trottenberg2000multigrid}.
We call the resulting algorithm the \textit{sparsified multigrid}.
We note that Krishnan, Fattal, and Szeliski
  \cite{krishnan2013efficient} present experimental results from 
  the use of a sparsification heuristic 
  in a multigrid algorithm for solving problems in computer vision.

Our new algorithms 
  are most closely related to the  Laplacian solver recently
  introduced by Peng and Spielman \cite{PengS14}:
  unlike the other Laplacian solvers, they rely only on sparsification and do not directly
  rely on 
  ``support theory'' preconditioners or any form of low stretch spanning trees.
This is why our algorithms can solve a much broader family of linear systems.
To sparsify without using graph theoretic algorithms, we employ a recently developed 
  algorithm of Cohen \textit{et. al.}~\cite{cohen2014uniform} that allows us to
  sparsify a matrix by solving systems of equations in a subsampled matrix.

\subsection{Connection Laplacians and Block DD Matrices}

In this section, we define block diagonally dominant
  (\bdd) matrices---the most general family of matrices such that the
  associated systems of linear equations 
  can be solved by our algorithms.
We begin by defining our motivating case: the connection Laplacians

Connection Laplacians may be thought of as a generalization of graph Laplacians
where every vertex is  associated with  a vector, instead of a real number,
  and every edge is associated with a unitary matrix.
Like graph Laplacians, they describe a natural quadratic form.
Let $\blk{\MM}{i,j} \in \complex^{r \times r}$ be the unitary
  matrix associated with edge $(i,j)$, and let $w_{i,j}$ be the
  (nonnegative) weight of edge $(i,j)$.
We require that $\blk{\MM}{i,j} = \blk{\MM}{j,i}^{*}$, where
  $*$ denotes the conjugate transpose.
The quadratic form associated with this connection Laplacian
  is a function of vectors $\vv^{(i)} \in \complex^{r}$,
  one for each vertex $i$, that equals
\[
\textstyle  \sum_{(i,j) \in E} \ww_{i,j} \norm{\vv^{(i)} - \blk{\MM}{i,j} \vv^{(j)}}^{2}.
\]
The matrix corresponding to this quadratic form is a block matrix
  with blocks $\blk{\MM}{i,j}$.
Most applications of the connection Laplacian require one to either solve
  systems of linear equations in this matrix,
  or to compute approximations of its smallest eigenvalues
  and eigenvectors.
By applying the inverse power method (or inverse Lanczos), we can
  perform these eigenvector calculations by solving a logarithmic
  number of linear systems in the matrix
  (see \cite[Section 7]{SpielmanTengLinsolve}).

The matrices obtained from the connection Laplacian are a special case
  of block diagonally dominant (\bdd) matrices, which we now define.
Throughout this paper, we consider block matrices having 
  entries in $\complex^{r \times r},$
  where $r > 0$ is a fixed integer. 
We say that
  $\MM \in \left({\complex^{r \times r}}\right)^{m \times n}$ has $m$
  block-rows, and $n$ block-columns. 
For $i \in [m], j \in [n],$ we let
  $\blk{\MM}{i,j} \in \complex^{r \times r}$ denote the $i,j$-block in
  $\MM,$ and $\blk{\MM}{j}$ denote the $j^\textrm{th}$ block-column,
  i.e.,
  $\blk{\MM}{j} = [(\blk{\MM}{1,j})^{\dg}, (\blk{\MM}{2,j})^{\dg},
  \ldots ,(\blk{\MM}{m,j})^{\dg}]^{\dg}.$
For sets $F \subseteq [m], C \subseteq [n],$ we let
  $\blk{\MM}{F,C} \in \left({\complex^{r \times r}}\right)^{|F| \times
  |C|}$ denote the block-submatrix with blocks $\blk{\MM}{i,j}$ for
  $i \in F, j \in C.$ 
$\MM$ is block-diagonal, if $\blk{\MM}{i,j} = 0$ for $i \neq j.$
\textit{We emphasize that all computations are done over $\complex$, not
  over a matrix group.}

To define bSDD matrices, let
  $\norm{\AA}$ denote the operator norm\footnote{%
Recall that the operator norm is the largest singular value of $\AA$ and
 square root of the largest eigenvalue of $\AA^{\dg} \AA$.}
  of a matrix $\AA$.
For Hermitian matrices $\AA, \BB,$ write $\AA \succeq \BB$ iff
$\AA-\BB$ is positive semidefinite.
\begin{definition}
A Hermitian block-matrix $\MM \in \left({\complex^{r
      \times r}}\right)^{n \times n}$ is block diagonally dominant (or {\bdd}) if
\[\textstyle \textrm{for all } i \in [n], \qquad \blk{\MM}{i,i} \succeq \id_{r}
\cdot \sum_{j: j \neq i} \norm{\blk{\MM}{i,j}},\]
where $\id_{r} \in \complex^{r \times r}$ denotes the identity matrix.
\end{definition}
\noindent Equivalently, a Hermitian $\MM$ is a {\bdd} matrix iff it can be
written as $\DD - \AA$ where $\DD$ is block-diagonal, and $
\blk{\DD}{i,i} \succeq \id_{r}\sum_{j} \norm{\blk{\AA}{i,j}}$ (see Lemma~\ref{lem:extra-diagonal}).

Throughout the paper we treat $r$ as a constant.
The hidden dependence on $r,$ of the running times of the algorithms
we present, is polynomial.
The major results of this paper are the following.

\begin{restatable}[Sparsified Multigrid]{theorem}{recursive}
\label{thm:recursive}
There is an algorithm that, when given
  a {\bdd} matrix $\MM$ with $n$ block rows 
  and $m$ nonzero blocks,
produces a solver for $\MM$ in $O(m\log n+n\log^{2+o(1)}n)$ work and
  $O(n^{o(1)})$ depth so that the solver finds $\epsilon$-approximate
  solutions to systems of equations in $\MM$ in $O((m+n\log^{1+o(1)}n)\log(1/\epsilon))$
work and $O(\log^{2}n\log\log n\log(1/\epsilon))$ depth.
\end{restatable}
\noindent Previously, the existence of nearly-linear time solvers was even unknown for the 1-dimensional case ($r=1$)
  when the off-diagonal entries were allowed to be complex numbers.

We can use the above algorithm to find approximations of the smallest eigenvalues and eigenvectors of
  such matrices at an additional logarithmic cost.
\begin{restatable}[Sparsified Cholesky]{theorem}{thmUDU}
\label{thm:UDU}
For every {\bdd} matrix $\MM$ with $n$ block-rows
there exists a diagonal matrix $\DD$ 
  and an upper triangular matrix $\UU$ with $O (n)$ nonzero blocks so that
\[
  \UU^{T} \DD \UU \approx_{3/4} \MM .
\]
Moreover, linear equations in $\UU$, $\UU^{T}$, and $\DD$ can be solved
   with linear work in depth $O (\log^{2} n)$, and these matrices can be computed
  in polynomial time.
\end{restatable}
\noindent These matrices allow one to solve systems of linear equations in $\MM$ to $\epsilon$-accuracy
  in parallel time $O (\log^{2} n \log^{-1} \epsilon )$ and work $O (m \log^{-1} \epsilon )$.
Results of this form were previously unknown
  even for graph Laplacians.

In the next two sections we explain the ideas used to prove these theorems.
Proofs may be found in the appendices that follow.


\section{Background}

We require some standard facts about the $\pleq$ order on matrices.
\begin{fact}\label{fact:orderInverse}
For $\AA$ and $\BB$ positive definite,
$\AA \pgeq \BB$ if and only if $\BB^{-1} \pgeq \AA^{-1}$.
\end{fact}
\begin{fact}\label{fact:orderCAC}
If $\AA \pgeq \BB$ and $\CC$ is any matrix of compatible dimension,
  then $\CC \AA \CC^{T} \pgeq \CC \BB \CC^{T}$.
\end{fact}
We say that $\AA$ is an $\epsilon$-approximation of $\BB$, written
  $\AA  \approx_{\epsilon} \BB $,
if $  e^{\epsilon} \BB \pgeq \AA \pgeq e^{-\epsilon} \BB$.
This relation is symmetric.
We say that $\xxtil$ is an $\epsilon$-approximate solution to the
  system $\AA \xx = \bb$ if
  $\norm{\xxtil - \AA^{-1} \bb}_{\AA} \leq \epsilon \norm{\xx}_{\AA}$,
 where $\norm{\xx}_{\AA} =  (\xx^{T} \AA \xx)^{1/2}$.
If $\epsilon<1/2$, $\AA \approx_{\epsilon} \BB$ and 
  $\BB \xxtil = \bb$, then
  $\xxtil$ is a $2 \epsilon$ approximate solution to
  $\AA \xx = \bb$.
\begin{fact}\label{frac:orderComposition}
If $\AA \approx_{\epsilon} \BB$
  and $\BB \approx_{\delta } \CC$, then
  $\AA \approx_{\epsilon + \delta} \CC$.
\end{fact}

We now address one technicality of dealing with $\bdd$ matrices:
  it is not immediate whether or not a $\bdd$ matrix is singular.
Moreover, if it is singular, the structure of its null space is not immediately clear either.
Throughout the rest of this paper, we will consider $\bdd$ matrices to which a small multiple
  of the identity have been added.
These matrices will be nonsingular.
To reduce the problem of solving equations in a general $\bdd$ matrix $\MM$
  to that of solving equations in a nonsingular matrix, we require an estimate
  of the smallest nonzero eigenvalue of $\MM$.
%
%
\begin{claim}\label{clm:LZ3L_new}
Suppose that all nonzero eigenvalues of $\MM$ are at least $\mu$ and 
$\ZZ \approx_{\epsilon} (\MM + \epsilon \mu \id)^{-1}$ for some $0<\epsilon<1/2$.
Given any $\bb$ such that $\MM \xx = \bb$ for some $\xx$, we have
\[
\norm{\xx-\ZZ \bb}^2_{\MM} \leq 6 \epsilon \norm{x}^2_{\MM}.
\]
\end{claim}

Hence, we can solve systems in $\MM$ by approximately solving systems in $\MM + \epsilon \mu \II$.
Any lower bound $\mu$ on the smallest nonzero eigenvalue of $\MM$ will suffice.
It only impacts the running times of our algorithms in the numerical precision
  with which one must carry out the computations.

The above claim allows us to solve systems in $\MM$ that have a solution.
If $\MM$ is singular and we want to apply its pseudoinverse
  (that is, to find the closest solution in the range of $\MM$),
  we can do so
  by pre and post multiplying by $\MM$ to project onto its range.
The resulting algorithm, which is implicit in the following claim,
  requires applying a solver for $\MM$ three times.
It also takes $O (\log \kappa)$ times as long to run,
  where $\kappa$ is the finite condition number\footnote{%
The finite condition number is the ratio of the largest singular value
  to the smallest \textit{nonzero} singular value.}
 of $\MM$.
\begin{claim}\label{clm:LZ3L}
Let $\kappa$ be an upper bound on the finite condition number of $\MM$. 
Given an error parameter $0 < \epsilon < 1$,
let $\delta = \nicefrac{\epsilon}{56\kappa^3}$.
If all nonzero eigenvalues of $\MM$ are at least $\mu$, and if
$\ZZ \approx_{\delta} (\MM + \epsilon \mu \II)^{-1}$,
then
\[
 \MM \ZZ^{3} \MM\approx_{4 \epsilon} \MM^{+}.
\]
\end{claim}


\newcommand{\subsolver}{\textsc{Subsolver}}

\section{Overview of the algorithms}

We index the block rows and columns of a block matrix by a set of
  vertices (or indices) $V$.
When we perform an elimination, we eliminate a set $F \subset V$,
  and let $C = V - F$.
Here, $F$ stands for ``fine'' and $C$ stands for ``coarse''.
In contrast with standard multigrid methods, we will have $\sizeof{F} \leq \sizeof{C}$.

To describe block-Cholesky factorization, we write the matrix $\MM$
  with the rows and columns in $F$ first:
\[
\MM
=
\left[
\begin{array}{cc}
\blk{\MM}{F,F} & \blk{\MM}{F,C}\\
\blk{\MM}{C,F} & \blk{\MM}{C,C}
\end{array}
\right].
\]
Cholesky factorization writes the inverse of this matrix as
\begin{equation}\label{eqn:blockInverse}
\MM^{-1}
=
\left[
\begin{array}{cc}
\II & -\blk{\MM}{F,F}^{-1} \blk{\MM}{F,C}\\
0 & \II
\end{array}
\right]
\left[
\begin{array}{cc}
\blk{\MM}{F,F}^{-1} & 0 \\
0 & \schur{\MM}{F}^{-1}
\end{array}
\right]
\left[
\begin{array}{cc}
\II & 0\\
-\blk{\MM}{C,F} \blk{\MM}{F,F}^{-1} & \II
\end{array}
\right],
\end{equation}
where
\[
\schur{\MM}{F} \defeq \blk{\MM}{C,C} - \blk{\MM}{C,F} \blk{\MM}{F,F}^{-1} \blk{\MM}{F,C}
\]
is the Schur complement of $\MM$ with respect to $F$.

Our algorithms rely on two fundamental facts about $\bdd$ matrices:
  that the Schur complement of a $\bdd$ matrix is a $\bdd$ matrix (Lemma \ref{lem:schurClosed})
  and that one can sparsify $\bdd$ matrices.
The following theorem is implicit in \cite{CarliSilvaHS}.
\begin{theorem}\label{thm:blockBSS}
For every $\epsilon \leq 1$, every {\bdd} matrix $\MM$ with $n$ 
  $r$-dimensional block rows and columns
  can be $\epsilon$-approximated by a {\bdd} matrix having at most
  $10 n r/ \epsilon^{2}$ nonzero blocks.
\end{theorem}

We use the identity \eqref{eqn:blockInverse} to reduce the problem of solving
  a system of equations in $\MM$ to that of solving equations in its Schur complement.
The easiest part of this is multiplication by $\blk{\MM}{C,F}$:
  the time is proportional to the number of nonzero entries in the submatrix,
  and we can sparsify $\MM$ to guarantee that this is small.

The costlier part of the reduction is the application of the   
  inverse of $\blk{\MM}{F,F}$ three times.
This would be fast if $\blk{\MM}{F,F}$ were block-diagonal, which corresponds
  to $F$ being an independent set.
We cannot find a sufficiently large independent set $F$,
  but we can find a large set that is almost independent.
This results in a matrix $\blk{\MM}{F,F}$ that is well approximated
  by its diagonal, and thus linear equations in this matrix
  can be quickly solved to high accuracy by a few Jacobi iterations
  (see Theorem \ref{thm:jacobi} and Section \ref{sec:jacobi}).

We can prove the existence of linear sized approximate inverses by
  explicitly computing $\schur{\MM}{F}$, sparsifying it,
  and then recursively applying the algorithm just described.
To make this algorithm efficient, we must compute a sparse approximation
  to $\schur{\MM}{F}$ without constructing $\schur{\MM}{F}$.
This is the problem of \textit{spectral vertex sparsification},
  and we provide a fast algorithm for this task in Sections
  \ref{sec:overviewSquaring} and
  \ref{sec:vertexSparsify}.


\subsection{Schur Complement Chains}
We encode this recursive algorithm by a \textit{Schur complement chain} (SCC).
An SCC defines a linear operator that can be used to approximately
  solve equations in an initial matrix $\MM^{(0)}$.
If the $\bdd$ matrix $\MM$ is sparse, then it is the same as $\MM^{(0)}$;
  if not, then $\MM^{(0)}$ is a sparse approximation to $\MM$.
Let $F_{1}$ be the first set of vertices eliminated,
 $\MM^{(1)}$ a sparse approximation to the Schur complement
  of $\MM^{(0)}$ with respect to $F_{1},$ and   $\ZZ^{(1)}$ an operators that approximates the inverse of
  $\blk{\MM}{F_1,F_1}^{(0)}.$
\begin{definition}[Schur Complement Chain]
\label{def:chain}
An $\eepsilon$-Schur complement chain ($\eepsilon$-SCC) for a matrix
  $\MM^{{0}}$ indexed by vertex set $V$ is a sequence of operators and subsets,
\[
\small
\left( (\MM^{(1)}, \ZZ^{(1)}) , \ldots, (\MM^{(d)},
  \ZZ^{(d)}); F_1,  \ldots, F_{d} \right) ,
\]
  so that for $C_{0} = V$ and 
  $C_{i+1} = C_{i} \setminus F_{i+1} $, $\sizeof{C_{d}} \leq 1000$ and
  for $1 \leq i \leq d$,
\[
\small
\textstyle\MM^{(i)} \approx_{\epsilon_{i}} \schur{\MM^{(i-1)}}{F_i}
\quad \text{ and } \quad 
0 \preceq (\ZZ^{(i)})^{-1} -
  \blk{\MM}{F_{i},F_{i}}^{(i-1)} \preceq \epsilon_{i}\cdot \schur{\MM^{(i-1)}}{C_{i}}.
\]
\end{definition}


The algorithm \textsc{ApplyChain}, described in Section \ref{sec:chains},
  applies an SCC to solve equations in $\MM^{(0)}$ in the natural way
  and satisfies the following guarantee.
\begin{restatable}[]{lemma}{applychain}
\label{lem:apply_chain}
Consider an $\eepsilon$-SCC for $\MM^{(0)}$ where $\MM^{(i)}$ and
$\ZZ^{(i)}$ can be applied to a vector in work
$W_{\MM^{(i)}}, W_{\DD^{(i)}}$ and depth
$D_{\MM^{(i)}}, D_{\ZZ^{(i)}}$ respectively.

The algorithm
$\textsc{ApplyChain}\left(\small
 (\MM^{(1)}, \ZZ^{(1)}) , \ldots,
  (\MM^{(d)}, \ZZ^{(d)}); F_1, \ldots, F_{d} ; \bb \right)$
corresponds to a linear operator $\WW$ acting on $\bb$ such that
\begin{enumerate}
\item
$\small \WW^{-1} \approx_{\sum_{i = 1}^{d} 2\epsilon_i} \MM^{(0)},$ and
\item for any vector $\bb$, $\textsc{ApplyChain}\left( \small (\MM^{(1)}, \ZZ^{(1)}) , \ldots, (\MM^{(d)},
  \ZZ^{(d)}); F_1,
  \ldots, F_{d} ; \bb \right)$ runs in \\
$O\left(\small \sum_{i = 1}^{d} \left( D_{\MM^{(i)}} + D_{\ZZ^{(i)}} \right)\right)$ depth
and $O\left(\small \sum_{i = 1}^{d}  \left(  W_{\MM^{(i)}} + W_{\ZZ^{(i)}} \right)\right)$ work.
\end{enumerate}
\end{restatable}
\noindent For an $\eepsilon$-SCC chain for $\MM^{(0)},$ we define the depth and
work of the chain to be the work and depth required by
$\textsc{ApplyChain}$ to apply the SCC.

\subsection{Choosing $F_{i}$: {$\alpha$-\bdd} matrices}

We must choose the set of vertices $F_{i}$ so that
  we can approximate the inverse of
  $\blk{\MM}{F_{i},F_{i}}^{(i)}$ by an operator
  $\ZZ^{(i)}$ that is efficiently computable.
We do this by requiring that the matrix
  $\blk{\MM}{F_{i},F_{i}}^{(i)}$
  be $\alpha$-block diagonally dominant (\abdd), a term
  that we now define.
\begin{definition}
\label{def:strongDD}
A Hermitian block-matrix $\MM$ is {\abdd} 
  if
\begin{equation}\label{eqn:abdd}
\textstyle \forall i, \quad
\blk{\MM}{i,i} \pgeq (1+\alpha) \id_{r} \cdot \text{$\sum_{j:j \neq i} \norm{\blk{\MM}{i,j}}$}.
\end{equation}
\end{definition}
\noindent We remark that a $0$-{\bdd} matrix is
simply a {\bdd} matrix. In particular, for $r=1,$ Laplacian matrices are $0$-\bdd.

By picking a subset of rows at random and discarding those
  that violate condition \eqref{eqn:abdd}, the algorithm \textsc{bDDSubset} (described in
Section~\ref{sec:absdd}) finds a linear sized subset $F$ of the
block-rows of a {\bdd} matrix $\MM$ so that $\blk{\MM}{F,F}$ is
{\abdd}.
\begin{restatable}[]{lemma}{abddSubsetAlgo}
\label{lem:abddSubset}
  Given a {\bdd} matrix $\MM$ with $n$ block-rows, and an
  $\alpha \geq 0$, \textsc{bDDSubset} computes a subset $F$
  of size least $n/(8 (1+\alpha))$ such that $\blk{\MM}{F,F}$ is {\abdd}.
It runs in
 runs in $O(m)$ expected work and $O(\log{n})$ expected depth,
 where $m$
  is the number of nonzero blocks in $\MM$.
\end{restatable}

We can express an {\abdd} matrix as a sum of a block diagonal matrix
  and a {\bdd} matrix so that it is well-approximated by the diagonal.

\begin{restatable}[]{lemma}{split}
\label{lem:split}
Every {\abdd} matrix $\MM$ can be written in the form $\XX + \LL$
  where $\XX$ is block-diagonal, $\LL$ is {\bdd}, and
  $\XX  \pgeq \frac{\alpha}{2} \LL$.
\end{restatable}
As $\LL$ is positive semidefinite,
$
\left( 1 +\nfrac{2}{\alpha}  \right) \XX \pgeq \MM \pgeq \XX,
$
which means that  $\XX$ is a good approximation of $\MM$ when $\alpha$ is reasonably big.
As block-diagonal matrices like $\XX$ are easy to invert,
systems in these well-conditioned matrices
can be solved rapidly using preconditioned iterative methods.
In Section~\ref{sec:jacobi}, we show that a variant of Jacobi iteration
 provides an operator that 
  satisfies the requirements of Definition~\ref{def:chain}.


\begin{restatable}[]{theorem}{jacobi}
\label{thm:jacobi}
Let $\MM$ be a {\bdd} matrix with index set $V$, and let $F
\subseteq V$
such that $\blk{\MM}{F,F}$ 
 is {\abdd} for some $\alpha \ge 4,$ and has $m_{FF}$ nonzero blocks.
The algorithm
  $\textsc{Jacobi}(\epsilon, \blk{\MM}{F,F},\bb) $ 
  acts as a linear operator $\ZZ$ on $\bb$ that satisfies
\[
0 \preceq (\ZZ)^{-1} -
  \blk{\MM}{F,F} \preceq \epsilon \cdot \schur{\MM }{V\setminus F}.
 \]
The algorithm takes 
$O(m_{FF}\log(\frac{1}{\epsilon})) $ work and $O(\log{n} \log(\frac{1}{\epsilon}))$ depth.
%
\end{restatable}

We now explain how Theorems \ref{thm:jacobi} and \ref{thm:blockBSS}
  allow us to construct Schur complement chains that can be applied
  in nearly linear time.
We optimize the construction in the next section.

Theorem \ref{thm:blockBSS} tells us that there is a $\bdd$ matrix
   $\MM^{(0)}$
  with $O (n / \epsilon^{2})$ nonzero blocks that
 that $\epsilon$-approximates $\MM$,
  and that for every $i$ there is a $\bdd$ matrix
  $\MM^{(i+1)}$ with $O (\sizeof{C_{i}} / \epsilon^{2})$ nonzero
  blocks that is an $\epsilon$-approximation of
  $\schur{\MM^{(i)}}{F_{i}}$.
We will pick $\epsilon$ later.
Lemma \ref{lem:abddSubset} provides a set $F_{i}$ containing
  a constant fraction of the block-rows of $\MM^{(i)}$ 
  so that $\blk{\MM}{F_{i}, F_{i}}^{(i)}$
  is $4$-{\bdd}.
Theorem \ref{thm:jacobi} then provides an operator
  that solves systems in $\blk{\MM}{F_{i},F_{i}}^{(i)}$ 
  to $\epsilon$ accuracy in time $O (\log \epsilon^{-1})$
  times the number of nonzero entries in $\blk{\MM}{F_{i}, F_{i}}^{(i)}$.
This is at most the number of nonzero entries in $\MM^{(i)}$,
  and thus at most $O (\sizeof{C_{i}} / \epsilon^{2})$.
As each $F_{i}$ contains at least a constant fraction of the rows of $\MM^{(i)}$,
  the depth of the recursion, $d$, is $O (\log n)$.
Thus, we can obtain constant accuracy by setting $\epsilon = \Theta (1/\log n)$.
The time required to apply the resulting SCC would thus be
  $O (n \log^{2} n \log \log n)$.

We can reduce the running time by setting $\epsilon_{1}$
  to a constant and allowing it to shrink as $i$ increases.
For example, setting 
  $\epsilon_{i} = 1/2 (i+1)^{2}$
  results in a linear time algorithm
  that produces a constant-factor approximation of the inverse of $\MM$.
We refine this idea in the next section.


\subsection{Linear Sized Approximate Inverses}\label{sec:overviewExists}

In this section we sketch the proof of Theorem \ref{thm:UDU}, which tells us that
  every $\bdd$ matrix has a linear-sized approximate inverse.
The rest of the details appear in Section \ref{sec:existence}.

The linear-sized approximate inverse of a $\bdd$ matrix $\MM$ with $n$
  block rows and columns
  is provided by a $3/4$-approximate of the form
  $\UU^{T} \DD \UU$ where $\DD$ is block diagonal
  and $\UU$ is block upper-triangular and has
  $O (n)$ nonzero blocks.
As systems of linear equations in block-triangular matrices
  like $\UU$ and $\UU^{T}$ can be solved in time proportional
  to their number of nonzero blocks, this provides
  a linear time algorithm for computing a $3/4$ approximation
  of the inverse of $\MM$.
By iteratively refining the solutions provided by the approximate inverse,
  this allows one to find $\epsilon$-accurate
  solutions to systems of equations in $\MM$ in time $O (m \log \epsilon^{-1})$.

The matrix $\UU$ that we construct has meta-block structure that allows
  solves in $\UU$ and $\UU^{T}$ to be performed with linear work and
  depth $O (\log^{2} n)$.
This results in a parallel algorithm for solving equations in $\MM$
  in work $O (m \log \epsilon^{-1})$ and depth $O (\log^{2} n \log \epsilon^{-1})$.
We remark that with some additional work (analogous to Section 7 of \cite{LeePengSpielman}),
  one could reduce this depth to $O (\log n \log \log n \log \epsilon^{-1})$.

The key to constructing $\UU$ is realizing that the algorithm \textsc{Jacobi}
  corresponds to multiplying by the matrix 
  $\ZZ^{(k)}$ defined in equation \eqref{eqn:defZ}.
Moreover, this matrix is a polynomial in 
  $\XX $
  and $\LL$ of degree $\log (3/\epsilon)$,
  where $\blk{\MM}{F_{i}, F_{i}}^{(i)} = \XX + \LL$
  where $\LL$ is  {\bdd}, and  $\XX$ block diagonal such that $\XX
  \pgeq 2\LL.$
To force $\ZZ^{k}$  to be a sparse matrix, we require
that $\LL$ be sparse.

If we use the algorithm \textsc{bSDDSubset} to choose $F_{i}$,
  then $\LL$ need not be sparse.
However, this problem is easily remedied by 
  forbidding algorithm \textsc{bSDDSubset}
  from choosing any vertex of more than twice average degree.
Thus, we can ensure
   that all $\ZZ^{(i)}$ are sparse.
\begin{lemma}
\label{lem:subsetLowDeg}
For every {\bdd} matrix $\MM$ and every $\alpha \geq 0$, 
  there is a subset $F$ of size at least $\frac{n}{16(1+\alpha)}$
  such that $\blk{M}{F, F}$ is {\abdd} and
   the number of nonzeros blocks in each block-row of $F$ at most twice
   the average number of nonzero blocks in a block-row of $\MM$.
\end{lemma}
\begin{proof}
Discard every block-row of $\MM$ that has more than twice the average
  number of nonzeros blocks per row-block.
Then remove the corresponding row blocks.
The remaining matrix has dimension at least $n/2$.
We can now use Lemma~\ref{lem:abddSubset} to find an {\abdd} submatrix.
\end{proof}

We obtain the $\UU^{T} \DD \UU$  factorization by applying the
  inverse of the factorization \eqref{eqn:blockInverse}:
\begin{equation*}
\MM
=
\left[
\begin{array}{cc}
\II & 0\\
\blk{\MM}{F,F}^{-1} \blk{\MM}{F,C} & \II
\end{array}
\right]
\left[
\begin{array}{cc}
\blk{\MM}{F,F} & 0 \\
0 & \schur{\MM}{F}
\end{array}
\right]
\left[
\begin{array}{cc}
\II & \blk{\MM}{C,F} \blk{\MM}{F,F}^{-1}\\
0 & \II
\end{array}
\right].
\end{equation*}
In the left and right triangular matrices we replace
  $\blk{\MM}{F,F}^{-1}$ with the polynomial we obtain from
  \textsc{Jacobi}.
In the middle matrix, it suffices to approximate $\blk{\MM}{F,F}$
  by a block diagonal matrix, and  $\schur{\MM}{F}$ by a factorization
  of its sparse approximation given by Theorem~\ref{thm:blockBSS}.
The details, along with a careful setting of $\epsilon_{i}$, are carried
  out in Section \ref{sec:existence}.


\newcommand{\schurApx}{\textsc{ApproxSchur}}
\newcommand{\schurSquare}{\textsc{SchurSquare}}
\newcommand{\lastStep}{\textsc{LastStep}}

\subsection{Spectral Vertex Sparsification Algorithm}
\label{sec:overviewSquaring}
In this section, we outline a procedure {\schurApx} that efficiently
approximates the Schur complement of a {\bdd} matrix $\MM$ w.r.t. a
set of indices $F$ s.t. $\blk{\MM}{F,F}$ is {\abdd}. The following
lemma summarizes the guarantees of {\schurApx}.
\begin{restatable}[]{lemma}{approxSchur}
\label{lem:approxSchur}
Let $\MM$ be a {\bdd} matrix with index set $V$, and $m$ nonzero
blocks. Let $F \subseteq V$ be such that $\blk{\MM}{F,F}$ is {\abdd}
for some $\alpha \ge 4.$ The algorithm $\schurApx(\MM,F,\epsilon)$,
returns a matrix $\MMtil_{SC}$ s.t.
\begin{enumerate}
\item $\MMtil_{SC}$ has $O(m ( \epsilon^{-1}
  \log \log \epsilon^{-1} )^{O(\log \log \epsilon^{-1})} )$ nonzero blocks, and
\item  $\MMtil_{SC} \approx_{\epsilon} \schur{\MM}{F}$,
\end{enumerate}
in
$O(m ( \epsilon^{-1} \log \log \epsilon^{-1} )^{O(\log \log
  \epsilon^{-1})} )$
work and $O(\log n (\log \log \epsilon^{-1}) )$ depth.
\end{restatable}
\noindent We sketch a proof of the above lemma in this section. A complete proof
and pseudocode for {\schurApx} are given in
Section~\ref{sec:vertexSparsify}.

First consider a very simple special case: where $F$ is a singleton,
$F = \{i\}$. Let $C = V \setminus F$ be the remaining indices.
\[
\schur{\MM}{i} = \blk{\MM}{C,C} - \blk{\MM}{C,i} \blk{\MM}{i,i}^{-1} \blk{\MM}{i,C}
\]
Thus, if $\blk{\MM}{C,i}$ has $k$ nonzero blocks, $\schur{\MM}{i}$
could have $k^2$ additional nonzero blocks compared to
$\blk{\MM}{C,C}$, potentially making it dense. 
If $\MM$ were a graph
Laplacian, then $\blk{\MM}{C,i} \blk{\MM}{i,i}^{-1} \blk{\MM}{i,C}$
would represent the adjacency matrix of a weighted clique.  In
Section~\ref{sec:weightedExp}, we construct weighted expanders that
allow us to $\epsilon$-approximate $\schur{\MM}{i}$ in this case using
$m + O(k\epsilon^{-4})$ edges. In
Section~\ref{sec:product-clique-sparsify}, we show how to use such
weighted expanders to sparsify $\schur{\MM}{i}$ when $\MM$ is a {\bdd}
matrix.

This reduction can also be performed in parallel.
If $F$ is such that $\blk{\MM}{F,F}$ is block diagonal, we can
approximate $\schur{\MM}{F}$ by expressing $\blk{\MM}{C,i}
\blk{\MM}{i,i}^{-1} \blk{\MM}{i,C}$ as $\sum_{i
 \in F} \blk{\MM}{C,i} \blk{\MM}{i,i}^{-1} \blk{\MM}{i,C},$ and using
$|F|$ weighted expanders.
However, $\blk{\MM}{F,F}$ may not be diagonal.
Instead, we give a procedure {\schurSquare} that generates $\MM'$
that is better approximated by its diagonal.

Invoking {\schurSquare} a few times leads a sequence
of matrices $\blk{\MM}{F,F}^{(0)}, \blk{\MM}{F,F}^{(1)}, \ldots, \blk{\MM}{F,F}^{(i)}$.
We will show that $\blk{\MM}{F,F}^{(i)}$ is $\epsilon$-approximated by its
diagonal and we call the procedure
{\lastStep} to approximate $\schur{\MM ^{(i)}}{F}$.
An additional caveat is that replacing $\blk{\MM}{F,F}^{(i)}$ by its
diagonal at this step gives errors that are difficult to bound.
We discuss the correct approximation below




{\schurSquare} is based on a \emph{squaring} identity for matrix
inverse developed in \cite{PengS14}. 
Given a splitting of $\blk{\MM}{F,F}$ into
$\DD - \AA,$ where $\DD$ is block-diagonal, and $\AA$ has its
diagonal blocks as zero, it relies on the fact
that the matrix
\begin{equation}
\label{eq:schurSquare:overview}
\MM_{2} = \frac{1}{2}
\left[
\begin{array}{cc}
\DD - \AA \DD^{-1} \AA &
\blk{\MM}{F,C} + \AA \DD^{-1}\blk{\MM}{F,C}\\
\blk{\MM}{C,F} + \blk{\MM}{C,F}\DD^{-1} \AA  
&  2 \blk{\MM}{C,C} - \blk{\MM}{C,F} \DD^{-1} \blk{\MM}{F,C}
\end{array}
\right]
\end{equation}
satisfies $\schur{\MM}{F} = \schur{\MM_{2}}{F}$.
Furthermore, we can show that if $\MM$ is $\alpha$-{\bdd},
$\DD - \AA \DD^{-1} \AA $ is $\alpha^{2}$-\bdd,
which indicates that the block on $[F, F]$ rapidly approaches being diagonal.

As $\MM_{2}$ may be dense, we construct a
sparse approximation to it.
Since $\DD$ is diagonal, we can construct sparse approximations to
the blocks $\blk{(\MM_{2})}{F,F}$ and $\blk{(\MM_{2})}{C,C}$
in a manner analogous to the case of diagonal $\blk{\MM}{F,F}$.
Similarly, we use bipartite expanders to construct sparse
approximations to $\blk{(\MM_{2})}{C,F}$ and $\blk{(\MM_{2})}{F,C}$.



Our sequence of calls to {\schurSquare} terminates with
$\blk{\MM}{F,F}^{(i)}$ being roughly $\epsilon^{-1}$-{\bdd}. We then
return $\lastStep(\MM^{(i)}, F,\epsilon^{-1}, \epsilon)$.
As mentioned above, we cannot just replace $\blk{\MM}{F,F}^{(i)}$
by its diagonal.
Instead, {\lastStep} performs one step of \emph{squaring}
similar to {\schurSquare} with a key difference: Rather than expressing
$\blk{\MM}{F,F}^{(i)}$ as $\DD - \AA,$ it expresses it as $\XX + \LL,$
where $\XX$ is block-diagonal, and $\LL$ is just barely {\bdd}.
With this splitting, it constructs $\MM^{\left(last\right)}$ after
performing one iteration similar to
Eq.~\eqref{eq:schurSquare:overview}.
After this step, it replaces the $\blk{\MM}{F,F}^{\left(last\right)}$ block
with the block-diagonal matrix $\XX$.
Again, we directly produce sparse approximations to
$\MM^{\left(last\right)}$ and its Schur complement
via weighted (bipartite) expanders.
A precise description and proofs are given in Section~\ref{sec:lastStep}.

\subsection{Sparsifying {\bdd} matrices}\label{sec:overviewSparse}

The main technical hurdle left to address is how we
  sparsify {\bdd} matrices.
We to do this both to approximate $\MM$ by $\MM^{(0)}$, if $\MM$ is not already sparse,
  and to ensure that all the matrices $\MM^{(i)}$ remain sparse.
While the spectral vertex sparsification algorithm described in the previous section
  allows us to compute an approximation to a Schur complement 
  $\schur{\MM^{(i)}}{F_{i}}$, it is sparse only when $\MM^{(i)}$ is already sparse.
As we iteratively apply this procedure, the density of the matrices
  produced will grow unacceptably.
We overcome this problem by occasionally applying another sparsification routine that
  substantially decreases the number of nonzero blocks.
The cost of this sparsification routine is that it requires solving systems of
  equations in sparse  ${\bdd}$ matrices.
We, of course, do this recursively.

Our sparsification procedure begins by generalizing the
  observation that graph Laplacians can be sparsified by
  sampling edges with probabilities determined by their effective 
  resistances \cite{SpielmanS08:journal}.
There is a block analog of leverage scores \eqref{eqn:tau} that provides
  probabilities of sampling blocks so that the resulting
  sampled matrix approximates the original and has $O (n \log n)$ nonzero blocks
  with high probability.
To compute this block analog of leverage scores we employ a recently
  introduced procedure of Cohen \textit{et. al.} \cite{cohen2014uniform}.
  
Once we generalize their results to block matrices, they show that we can obtain
  sufficiently good estimates of the block leverage scores by computing leverage
  scores in a {\bdd} matrix obtained by randomly subsampling blocks of the original.
The block leverage scores in this matrix are obtained by solving a logarithmic number
  of linear equations in this subsampled matrix.
Thus, our sparsification procedure requires constructing a solver for a subsampled
  matrix and then applying that solver a logarithmic number of times.
We compute this solver recursively.

There is a tradeoff between the number of nonzero blocks in the subsampled
  system and in the resulting approximation of the original matrix.
If the original matrix has $m$ nonzero blocks and we subsample to a system
  of $m / K$ nonzero blocks, then we obtain an $\epsilon$-approximation
  of the original matrix with $O (K \epsilon^{-2} n \log n )$
  nonzero blocks.

The details of the analysis of the undersampling procedure appear
  in Section \ref{sec:undersampling}.

\subsection{The main algorithm}\label{sec:overviewAlg}

We now explain how we prove Theorem~\ref{thm:recursive}.
The details supporting this exposition appear in Section~\ref{sec:recursive}.
Our main goal is to control the density of the Schur complement
chain as we repeatedly invoke Lemma~\ref{lem:approxSchur}.


Starting from some $\MM^{(0)}$, we compute sets $F_{i}$
  (via calls to \textsc{bDDSubset}), approximate
  solvers $\ZZ^{(i)}$ (via \textsc{Jacobi}), 
  and approximations of Schur complements $\MM^{(i)}$
  (via \schurApx),
  until we obtain a matrix $\MM^{(i)}$ such that its dimension is a smaller
  than that of $\MM^{(0)}$ by a large constant factor (like 4).
While the dimension of $\MM^{(i)}$ is much smaller, its number of nonzero
  blocks is potentially larger by an even larger factor.
We use the procedure described in the previous section to
  sparsify it.
This sparsification procedure produces a sparse approximation of the matrix
  at the cost of solving systems of equations in a subsampled version of that matrix.
Some care is required to balance the cost of the resulting recursion.

We now sketch an analysis of a nearly linear time algorithm that results
  from a simple choice of parameters.
We optimize the parameter choice and analysis in Section~\ref{sec:recursive}.
Let $n$ be the dimension of $\MM^{(0)}$ and let $m$ be its number of nonzero blocks.
To begin, assume that $m \leq n \Delta \log^{3} n$, for a $\Delta$ to be specified later.
We call this the \textit{sparse} case, and address the case of dense $\MM$ later.

We consider fixing $\epsilon_{i} = c / \log n$ for all $i$, for some constant $c$.
As the depth of the Schur complement chain is $O (\log n)$, this results
  in a solver with constant accuracy.
A constant number of iterations of the procedure described above are required
  to produce an $\MM^{(i)}$ whose dimension is a factor of 4 smaller than $\MM^{(0)}$.
Lemma~\ref{lem:approxSchur} tells us that the edge density of this
$\MM^{(i)}$ is potentially
 higher than that of $\MM^{(0)}$ by a factor of
\[\textstyle
O\left( \left( \epsilon^{-1} \log\log \epsilon^{-1}  \right)
	^{O\left(\log\log \epsilon^{-1} \right)} \right)
= \exp \left( O( \log\log^2{n}) \right) = n^{o(1)}.
\]
Set $\Delta$ to be this factor.
We use the sparsification procedure from the previous section to guarantee
  that no matrix in the chain has density higher than that of this matrix,
  which is upper bounded by $(m/n) \Delta  = \Delta^{2} \log^{3} n$.
  
Setting $K = 2 \Delta$, the subsampling produces a matrix
  of density half that of $\MM^{(0)}$, and it produces a sparse approximation
  of $\MM^{(i)}$ of density $O (K  \epsilon_{i}^{-2} \log n) 
  \leq O (\Delta \log^{3} n)$, which, by setting constants appropriately,
  we can also force to be half that of $\MM^{(0)}$.
In order to perform the sparsification procedure, we need to construct
  a Schur complement chain for the subsampled matrix, and then use it to
  solve $O (\log n)$ systems of linear equations.
The cost of using this chain to solve equations in the subsampled system
  is at most $O (n \Delta^{2} \log^{4} n)$,
  and the cost of using the solutions to these equations to sparsify $\MM^{(i)}$
  is $O (m \log n)$.
The cost of the calls to \textsc{bDDSubset} and  \schurApx \ 
  are proportional to the number of edges in the matrices,
  which is $O (n \Delta \log^4{n})$.

We repeat this procedure all the way down the chain, only
  using sparsification when the dimension of $\MM^{(i)}$
  shrinks by a factor of $4$.
Since none of the matrices that we generate have density higher than
 $\Delta \log^{3}n$, we remain in the sparse case.
Let $T_{sparse} (n)$ be the time required to construct a solver chain
  on systems of size $n$  with $m \leq n \Delta \log^3{n}$.
We obtain the following recurrence
\[
  T_{sparse} (n) \leq 2 T_{sparse}  (n/4) + n \Delta^{2} \log^{4} n + m \log n + m \Delta 
  \leq  2 T_{sparse}  (n/4) + n \Delta^{2} \log^{4} n + n \Delta  \log n + n \Delta^{2},
\]
which gives
\[
  T (n)_{sparse} \leq O (n \Delta^{2} + n \Delta^{2} \log^{4} n)
  \leq n^{1 + o (1)}.
\]
To handle the case of dense $\MM$, we repeatedly sparsify while keeping $n$
  fixed until we obtain a matrix with fewer than $n \Delta \log^{3} n$ edges,
  at which point we switch to the algorithm described above.
The running time of this algorithm on a graph with $m$ edges, $T_{dense}(m)$,
  satisfies the recurrence
\[
T_{dense}(m) \leq
\begin{cases}
  T_{sparse} (n) & \qquad \text{if } m \leq n \Delta \log^3{n}, \text{ and}\\
2 T_{dense}\left(m / 2\right)  + n \Delta^{2} \log^{4} n + m \log n + m \Delta & \qquad \text{otherwise}.
\end{cases}
\]
Thus $T_{dense} (m)$  is upper bounded bounded by $O(mn^{o(1)} + n^{1 + o(1)})$.

We tighten this bound in Section~\ref{sec:recursive} 
  to prove  Theorem~\ref{thm:recursive}
by
  carefully choosing the parameters to accompany  a sequence $\eepsilon$
  that starts constant and decreases slowly.



\section{Summary}\label{sec:summary}

We introduce a new approach to solving systems of linear equations
  that gives the first nearly linear time algorithms for solving
  systems
in connection Laplacians and the first proof that
  connection Laplacians have linear-sized approximate inverses.
This was unknown even for graph Laplacians.

Our algorithms build on ideas introduced in \cite{PengS14}
  and are a break from those used in the previous work on solving systems
  of equations in graph Laplacians \cite{Vaidya,SpielmanTengLinsolve,KMP1,KMP2,KOSZ,CohenKMPPRX}.
Those algorithms all rest on \textit{support theory} \cite{SupportGraph},
  originally introduced by Vaidya \cite{Vaidya}, and rely on the fact
  that the Laplacian of one edge is approximated by the Laplacian
  of a path between its endpoints.
No analogous fact is true for connection Laplacians, even those
  with complex entries for $r = 1$.

Instead, our algorithms rely on many new ideas, the first being
  that of sparsifying the matrices that appear during elimination.
Other critical ideas are finding {\abdd} subsets of vertices to eliminate
  in bulk, approximating Schur complements without computing
  them explicitly, and the use of sub-sampling to sparsify
  in a recursive fashion.
To efficiently compute approximations of the Schur complements,
  we introduce a new operation that transforms a matrix
  into one with the same 
  Schur complement but a much better conditioned upper block
  \eqref{eq:schurSquare:overview}.
To obtain the sharp bounds in Theorem \ref{thm:recursive},
  we exploit a new linear-time algorithm for 
  constructing linear-sized sparse approximations to
  implicitly represented weighted cliques whose edge weights
  are products of weights at vertices (Section \ref{sec:weightedExp}), and extend this to the analog for {\bdd} matrices   
  (Section \ref{sec:product-clique-sparsify}). 


\bibliographystyle{alpha}
\bibliography{stoc16}

\newcommand{\etalchar}[1]{$^{#1}$}
\begin{thebibliography}{ANKKS{\etalchar{+}}12}

\bibitem[ABFM14]{alexeev2014phase}
Boris Alexeev, Afonso~S Bandeira, Matthew Fickus, and Dustin~G Mixon.
\newblock Phase retrieval with polarization.
\newblock {\em SIAM Journal on Imaging Sciences}, 7(1):35--66, 2014.

\bibitem[ANKKS{\etalchar{+}}12]{arie2012global}
Mica Arie-Nachimson, Shahar~Z Kovalsky, Ira Kemelmacher-Shlizerman, Amit
  Singer, and Ronen Basri.
\newblock Global motion estimation from point matches.
\newblock In {\em 3D Imaging, Modeling, Processing, Visualization and
  Transmission (3DIMPVT), 2012 Second International Conference on}, pages
  81--88. IEEE, 2012.

\bibitem[AT11]{avron2011effective}
Haim Avron and Sivan Toledo.
\newblock Effective stiffness: Generalizing effective resistance sampling to
  finite element matrices.
\newblock {\em arXiv preprint arXiv:1110.4437}, 2011.

\bibitem[AW02]{ahlswede2002strong}
Rudolf Ahlswede and Andreas Winter.
\newblock Strong converse for identification via quantum channels.
\newblock {\em Information Theory, IEEE Transactions on}, 48(3):569--579, 2002.

\bibitem[BGH{\etalchar{+}}06]{SupportGraph}
M.~Bern, J.~Gilbert, B.~Hendrickson, N.~Nguyen, and S.~Toledo.
\newblock Support-graph preconditioners.
\newblock {\em SIAM J. Matrix Anal. \& Appl}, 27(4):930--951, 2006.

\bibitem[BSS13]{connection}
Afonso~S Bandeira, Amit Singer, and Daniel~A Spielman.
\newblock A cheeger inequality for the graph connection laplacian.
\newblock {\em SIAM Journal on Matrix Analysis and Applications},
  34(4):1611--1630, 2013.

\bibitem[CKM{\etalchar{+}}14]{CohenKMPPRX}
Michael~B. Cohen, Rasmus Kyng, Gary~L. Miller, Jakub~W. Pachocki, Richard Peng,
  Anup~B. Rao, and Shen~Chen Xu.
\newblock Solving sdd linear systems in nearly mlog1/2n time.
\newblock In {\em Proceedings of the 46th Annual ACM Symposium on Theory of
  Computing}, STOC '14, pages 343--352, New York, NY, USA, 2014. ACM.

\bibitem[CLM{\etalchar{+}}14]{cohen2014uniform}
Michael~B Cohen, Yin~Tat Lee, Cameron Musco, Christopher Musco, Richard Peng,
  and Aaron Sidford.
\newblock Uniform sampling for matrix approximation.
\newblock {\em arXiv preprint arXiv:1408.5099}, 2014.

\bibitem[dCSHS11]{CarliSilvaHS}
Marcel~K. de~Carli~Silva, Nicholas J.~A. Harvey, and Cristiane~M. Sato.
\newblock Sparse sums of positive semidefinite matrices.
\newblock {\em CoRR}, abs/1107.0088, 2011.

\bibitem[KFS13]{krishnan2013efficient}
Dilip Krishnan, Raanan Fattal, and Richard Szeliski.
\newblock Efficient preconditioning of laplacian matrices for computer
  graphics.
\newblock {\em ACM Transactions on Graphics (TOG)}, 32(4):142, 2013.

\bibitem[KMP10]{KMP1}
I.~Koutis, G.L. Miller, and R.~Peng.
\newblock Approaching optimality for solving {SDD} linear systems.
\newblock In {\em Foundations of Computer Science (FOCS), 2010 51st Annual IEEE
  Symposium on}, pages 235 --244, 2010.

\bibitem[KMP11]{KMP2}
I.~Koutis, G.L. Miller, and R.~Peng.
\newblock A nearly-$m \log n$ time solver for {SDD} linear systems.
\newblock In {\em Foundations of Computer Science (FOCS), 2011 52nd Annual IEEE
  Symposium on}, pages 590--598, 2011.

\bibitem[KOSZ13]{KOSZ}
Jonathan~A Kelner, Lorenzo Orecchia, Aaron Sidford, and Zeyuan~Allen Zhu.
\newblock A simple, combinatorial algorithm for solving sdd systems in
  nearly-linear time.
\newblock In {\em Proceedings of the 45th annual ACM symposium on Symposium on
  theory of computing}, pages 911--920. ACM, 2013.

\bibitem[LMP13]{li2013iterative}
Mu~Li, Gary~L Miller, and Rongkun Peng.
\newblock Iterative row sampling.
\newblock In {\em Foundations of Computer Science (FOCS), 2013 IEEE 54th Annual
  Symposium on}, pages 127--136. IEEE, 2013.

\bibitem[LPS88]{LPS}
A.~Lubotzky, R.~Phillips, and P.~Sarnak.
\newblock {Ramanujan} graphs.
\newblock {\em Combinatorica}, 8(3):261--277, 1988.

\bibitem[LPS15]{LeePengSpielman}
Yin~Tat Lee, Richard Peng, and Daniel~A. Spielman.
\newblock Sparsified cholesky solvers for {SDD} linear systems.
\newblock {\em CoRR}, abs/1506.08204, 2015.

\bibitem[Mar88]{Margulis88}
G.~A. Margulis.
\newblock Explicit group theoretical constructions of combinatorial schemes and
  their application to the design of expanders and concentrators.
\newblock {\em Problems of Information Transmission}, 24(1):39--46, July 1988.

\bibitem[MP13]{MillerP13}
Gary~L. Miller and Richard Peng.
\newblock Approximate maximum flow on separable undirected graphs.
\newblock In {\em Proceedings of the Twenty-Fourth Annual ACM-SIAM Symposium on
  Discrete Algorithms}, pages 1151--1170. SIAM, 2013.

\bibitem[MSS15]{IF4}
Adam~W Marcus, Nikhil Srivastava, and Daniel~A Spielman.
\newblock Interlacing families {IV}: {Bipartite} {Ramanujan} graphs of all
  sizes.
\newblock {\em arXiv preprint arXiv:1505.08010}, 2015.
\newblock to appear in FOCS 2015.

\bibitem[MTW14]{marchesini2014alternating}
Stefano Marchesini, Yu-Chao Tu, and Hau-tieng Wu.
\newblock Alternating projection, ptychographic imaging and phase
  synchronization.
\newblock {\em arXiv preprint arXiv:1402.0550}, 2014.

\bibitem[MV77]{ICC}
J.~A. Meijerink and H.~A. van~der Vorst.
\newblock An iterative solution method for linear systems of which the
  coefficient matrix is a symmetric $m$-matrix.
\newblock {\em Mathematics of Computation}, 31(137):148--162, 1977.

\bibitem[OSB15]{stable2015}
Onur Özyeşil, Amit Singer, and Ronen Basri.
\newblock Stable camera motion estimation using convex programming.
\newblock {\em SIAM Journal on Imaging Sciences}, 8(2):1220--1262, 2015.

\bibitem[PS14]{PengS14}
Richard Peng and Daniel~A. Spielman.
\newblock An efficient parallel solver for {SDD} linear systems.
\newblock In {\em Symposium on Theory of Computing, {STOC} 2014, New York, NY,
  USA, May 31 - June 03, 2014}, pages 333--342, 2014.

\bibitem[SS11a]{singer2011three}
Amit Singer and Yoel Shkolnisky.
\newblock Three-dimensional structure determination from common lines in
  cryo-em by eigenvectors and semidefinite programming.
\newblock {\em SIAM journal on imaging sciences}, 4(2):543--572, 2011.

\bibitem[SS11b]{SpielmanS08:journal}
D.~Spielman and N.~Srivastava.
\newblock Graph sparsification by effective resistances.
\newblock {\em SIAM Journal on Computing}, 40(6):1913--1926, 2011.

\bibitem[SS11c]{spielman2011graph}
Daniel~A Spielman and Nikhil Srivastava.
\newblock Graph sparsification by effective resistances.
\newblock {\em SIAM Journal on Computing}, 40(6):1913--1926, 2011.

\bibitem[SS12]{shkolnisky2012viewing}
Yoel Shkolnisky and Amit Singer.
\newblock Viewing direction estimation in cryo-em using synchronization.
\newblock {\em SIAM journal on imaging sciences}, 5(3):1088--1110, 2012.

\bibitem[ST14]{SpielmanTengLinsolve}
Daniel~A. Spielman and Shang-Hua Teng.
\newblock Nearly-linear time algorithms for preconditioning and solving
  symmetric, diagonally dominant linear systems.
\newblock {\em SIAM. J. Matrix Anal. \& Appl.}, 35:835–885, 2014.

\bibitem[SW12]{singer2012vector}
Amit Singer and H-T Wu.
\newblock Vector diffusion maps and the connection laplacian.
\newblock {\em Communications on pure and applied mathematics}, 65(8), 2012.

\bibitem[Tch36]{Tchudakoff}
Nikolai Tchudakoff.
\newblock On the difference between two neighbouring prime numbers.
\newblock {\em Rec. Math. [Mat. Sbornik] N.S.}, 1(6):799--814, 1936.

\bibitem[TOS00]{trottenberg2000multigrid}
Ulrich Trottenberg, Cornelius~W Oosterlee, and Anton Schuller.
\newblock {\em Multigrid}.
\newblock Academic press, 2000.

\bibitem[Tro12]{tropp2012user}
Joel~A Tropp.
\newblock User-friendly tail bounds for sums of random matrices.
\newblock {\em Foundations of Computational Mathematics}, 12(4):389--434, 2012.

\bibitem[Vai90]{Vaidya}
Pravin~M. Vaidya.
\newblock Solving linear equations with symmetric diagonally dominant matrices
  by constructing good preconditioners.
\newblock Unpublished manuscript UIUC 1990. A talk based on the manuscript was
  presented at the IMA Workshop on Graph Theory and Sparse Matrix Computation,
  October 1991, Minneapolis., 1990.

\bibitem[ZS14]{zhao2014rotationally}
Zhizhen Zhao and Amit Singer.
\newblock Rotationally invariant image representation for viewing direction
  classification in cryo-em.
\newblock {\em Journal of structural biology}, 186(1):153--166, 2014.

\end{thebibliography}

\appendix

\newpage
\tableofcontents
\thispagestyle{empty}

\newpage

\section{Background}
\begin{proof} (of Claim~\ref{clm:LZ3L_new})
Note that
\[
\norm{\xx-\ZZ b}^2_{\MM} = \xx^{\dg} \MM x - 2 \xx^{\dg} \MM \ZZ \MM \xx
	+  \xx^{\dg} \MM \ZZ \MM \ZZ \MM \xx
\]
Since all nonzero eigenvalues of $\MM$ are at least $\mu$, the eigenvalues of 
$\MM^{1/2} (\MM + \epsilon \mu \II)^{-1}\MM^{1/2}$ lie between $1/(1+\epsilon)$ and 1.
Using $\ZZ \approx_{\epsilon} (\MM + \epsilon \mu \II)^{-1}$, we see that
the eigenvalues of $\MM^{1/2} \ZZ \MM^{1/2}$ lie between $e^{-2\epsilon}$ and $e^{\epsilon}$.
Using $0<\epsilon<1/2$, we have
\[
\norm{\xx-\ZZ \bb}^2_{\MM} \leq (1 - 2 e^{-2\epsilon} + e^{2\epsilon}) \xx^{\dg} \MM \xx
	\leq 6 \epsilon \norm{\xx}^2_{\MM}.
\]
\end{proof}

\begin{claim}\label{clm:M3Z3}
Let $\AA$ be a matrix of condition number $\kappa$ and let $\AA \approx_{\epsilon} \BB$
  for $\epsilon \leq (56 \kappa^{3})^{-1}$.
Then,
  $\AA^{3} \approx_{28 \kappa^{3} \epsilon} \BB^{3}$.
\end{claim}
\begin{proof}
First, observe that $\AA \approx_{\epsilon} \BB $
  implies that $\norm{\BB} \leq (1 + e^{\epsilon}) \norm{\AA} \leq 2 \norm{\AA}$.
It also implies that $\norm{\AA - \BB} \leq 2 \epsilon \norm{\AA}$.
As 
\[
\AA^{2} - \BB^{2} = 
 \frac{1}{2} (\AA - \BB) (\AA + \BB)
+ \frac{1}{2} (\AA + \BB) (\AA - \BB),
\]
$\norm{\AA^{2} - \BB^{2}} \leq 6 \epsilon  \norm{\AA}^{2}$.
Similarly, as
\[
\AA^{3} - \BB^{3} = 
 \frac{1}{2} (\AA - \BB) (\AA^{2} + \BB^{2})
+ \frac{1}{2} (\AA + \BB) (\AA^{2} - \BB^{2}),
\]
$\norm{\AA^{3} - \BB^{3}} \leq 28 \epsilon  \norm{\AA}^{3}$.

Let $\kappa = \norm{\AA} / \lambda_{min} (\AA)$.
The above relation implies that
\[
\norm{\AA^{3} - \BB^{3}} \leq 28 \epsilon  \kappa^{3} \lambda_{min} (\AA)^{3}
=
28 \epsilon  \kappa^{3} \lambda_{min} (\AA^{3}).
\]
Thus, as $28 \epsilon \kappa^{3} \leq 1/2$,
\[
  \AA^{3} \approx_{56 \epsilon \kappa^{3}} \BB^{3}.
\]
\end{proof}

\begin{proof}
(of Claim~\ref{clm:LZ3L})

By Claim~\ref{clm:M3Z3}, using $\delta = \frac{\epsilon}{56\kappa^{3}}$,
\[
\MM \ZZ^{3} \MM \approx_{\epsilon} \MM (\MM+\epsilon \mu I)^{-3} \MM.
\]

$\MM$ has an eigendecomposition in the same basis as $ (\MM+\epsilon
\mu I)^{-3}$, and so it follows that $\MM (\MM+\epsilon \mu I)^{-3}
\MM$ has the same eigenbasis, and the same null space as $\MM$.

When $\lambda^{-1}$ is the eigenvalue of $\MM^{+}$ of
an eigenvector $\vv$, the corresponding eigenvalue of  
$\MM (\MM+\epsilon \mu I)^{-3} \MM$ is
$
\beta \defeq 
\frac{\lambda^{2}}
{(\lambda + \epsilon \mu)^{3} }
$ and
\[
\lambda^{-1} > \beta
\geq\frac{\lambda^{2}}
{(1+\epsilon)^{3}\lambda^{3}  } = e^{-3\epsilon} \lambda^{-1}.
\]
So $\MM (\MM+\epsilon \mu I)^{-3} \MM \approx_{3 \epsilon} \MM^{+}$,
and $\MM \ZZ^{3} \MM \approx_{4 \epsilon} \MM^{+}$.
\end{proof}

\begin{fact}
\label{fact:sum-of-unitary}
For every $\dd \in \complex^{r \times r},$ we can find
$\QQ^{(1)}, \QQ^{(2)} \in \complex^{r \times r}$ where
\[(\QQ^{(1)}) (\QQ^{(1)})^{\dg} = (\QQ^{(1)})^{\dg} (\QQ^{(1)}) = (\QQ^{(2)})
(\QQ^{(2)})^{\dg} = (\QQ^{(2)})^{\dg} (\QQ^{(2)}) = \id_{r},\]
such that
\[\dd = \frac{1}{2}\norm{\dd}(\QQ^{(1)} + \QQ^{(2)}).\]
\end{fact}
\begin{proof}
  Let $\hat{\dd} = \frac{1}{\norm{\dd}} \dd.$ Thus,
  $\norm{\hat{\dd}} = 1.$ Write $\hat{\dd}$ using its singular value
  decomposition as $\UU \DD \VV^{\dg},$ where
  $\DD, \UU, \VV \in \complex^{r \times r},$ $\DD$ is a real diagonal
  matrix with the singular values of $\dd$ on the diagonal, and
  $\UU \UU^{\dg} = \UU^{\dg} \UU = \VV \VV^{\dg} = \VV^{\dg} \VV = \id_{r}.$ Since
  $\norm{\hat{\dd}} = 1,$ we have $\DD_{j,j} \in [0,1]$ for all
  $j \in [r].$ Thus, there exists a real $\theta_{j}$ such that
  $\cos \theta_{j} = \DD_{j,j}.$ If we let $\DD^{(1)}, \DD^{(2)}$ be
  diagonal matrices usch that
  $\DD^{(1)}_{j,j} = \exp(i \theta_{j}), \DD^{(2)}_{j,j} = \exp(-i
  \theta_{j}),$
  we have $\DD = \frac{1}{2} (\DD^{(1)} + \DD^{(2)}).$ Moreover,
  \[(\DD^{(1)}) (\DD^{(1)})^{\dg} = (\DD^{(1)})^{\dg} (\DD^{(1)}) =
  (\DD^{(2)}) (\DD^{(2)})^{\dg} = (\DD^{(2)})^{\dg} (\DD^{(2)}) = \id_{r}.\]
  Letting $\QQ^{(k)} = \UU \DD^{(k)} \VV^{*}$ for $k = 1,2,$ we get
  $\hat{\dd} = \frac{1}{2} (\QQ^{(1)} + \QQ^{(2)})$ and hence
  $\dd = \frac{1}{2} \norm{\dd} (\QQ^{(1)} + \QQ^{(2)}).$ Moreover,
\[(\QQ^{(1)}) (\QQ^{(1)})^{\dg} = (\QQ^{(1)})^{\dg} (\QQ^{(1)}) = (\QQ^{(2)})
(\QQ^{(2)})^{\dg} = (\QQ^{(2)})^{\dg} (\QQ^{(2)}) = \id_{r}.\]
\end{proof}

\begin{fact}[Lemma B.1. from~\cite{MillerP13}]
\label{fact:schurLoewner}
If $\MM$ and $\Mtil$ are positive semidefinite
matrices satisfying $\MM \preceq \Mtil$,
then
\[
\schur{\MM}{F} \preceq \schur{\Mtil}{F}.
\]
\end{fact}

This fact can be proven via an energy minimization
definition of Schur complement.
More details on this formulation can be found in~\cite{MillerP13}.

\section{Block Diagonally Dominant Matrices}
In this section, we prove a few basic facts about {\bdd} matrices. The
following lemma gives an equivalent definition of {\bdd} matrices.
\begin{lemma}
\label{lem:extra-diagonal}
  A Hermitian block-matrix
  $\MM \in \left({\complex^{r \times r}}\right)^{n \times n}$ is
  {\bdd} iff it can be written as $\DD - \AA$ where $\DD$ is
  block-diagonal, $\AA$ is Hermitian, and
  $ \blk{\DD}{i,i} \succeq \id_{r}\sum_{j} \norm{\blk{\AA}{i,j}},$ for
  all $i \in V.$
\end{lemma}
\begin{proof}
The only if direction is easy. For {\bdd} matrix $\MM,$ if we let
$\DD$ be the block diagonal matrix such that $\blk{\DD}{i,i} =
\blk{\MM}{i,i},$ and $\AA$ be the matrix such that 
\begin{align*}
\blk{\AA}{i,j} = 
\begin{cases}
0 & i = j \\
\blk{\MM}{i,j} & i \neq j
\end{cases}
\end{align*}
It is immediate that $\AA$ is Hermitian. Moreover, $\MM$ is {\bdd} implies that $
\blk{\DD}{i,i} \succeq \id_{r}\sum_{j} \norm{\blk{\AA}{i,j}}.$

For the if direction. Suppose we have $\DD, \AA$ such that $\DD$ is
block-diagonal, and for all $i$ $
\blk{\DD}{i,i} \succeq \id_{r}\sum_{j} \norm{\blk{\AA}{i,j}}.$ Thus,
letting $\MM = \DD - \AA,$ we have, for all $i \in V,$
\begin{align*}
\id_{r} \cdot \sum_{j \neq i} \norm{\blk{\MM}{i,j}} 
= \id_{r} \cdot \sum_{j \neq i} \norm{\blk{\AA}{i,j}} 
\pleq \blk{\DD}{i,i} - \id_{r} \cdot \norm{\blk{\AA}{i,i}} 
\pleq \blk{\DD}{i,i} - \blk{\AA}{i,i}
= \blk{\MM}{i,i},
\end{align*}
where we used that since $\AA, \MM$ are Hermitian,
$\DD$ is also Hermitian, and thus $\blk{\DD}{i,i}, \blk{\AA}{i,i}$ are Hermitian.
\end{proof}

This immediately implies the corollary that flipping the sign
of off-diagonal blocks preserves {\bdd}-ness.
\begin{corollary}
\label{cor:bsdd-sign-flip}
  Given a {\bdd} matrix $\MM,$ write it as $\DD - \AA,$ where $\DD$
  is a block-diagonal, and $\AA$ has its diagonal blocks as
  zero. Then, $\DD + \AA$ is also PSD.
\end{corollary}
\begin{proof}
  First observe that for all $i,$
  $\blk{(\DD + \AA)}{i,i} = \blk{(\DD - \AA)}{i,i} = \blk{\MM}{i,i},$
  i.e., their diagonal blocks are identical. Moreover, for all $i \neq
  j,$ we have
  $\norm{\blk{(\DD + \AA)}{i,j}} = \norm{\blk{(\DD - \AA)}{i,j}} =
  \norm{\blk{\MM}{i,j}}.$
Thus $\DD + \AA$ is also {\bdd}, and hence PSD.
\end{proof}

Next, we show that the class of {\bdd} matrices is closed under Schur complement.
\begin{lemma}
\label{lem:schurClosed}
The class of {\bdd} matrices is closed under Schur complement.
\end{lemma}
\begin{proof}
  Since Schur complementation does not depend on the order of indices
  eliminated, it suffices to prove that for any {\bdd} matrix
  $\MM \in (\complex^{r \times r})^{n \times n},$ $\schur{\MM}{1}$ is
  a {\bdd} matrix. Let $C = V \setminus \{1\}.$

  We have
  $\schur{\MM}{1} = \blk{\MM}{C,C} - \blk{\MM}{C,1}\blk{\MM}{1,1}^{-1}
  \blk{\MM}{1,C}.$
  Let $\DD \in (\complex^{r \times r})^{C \times C}$ be the block
  diagonal matrix such that $\blk{\DD}{i,i} = \blk{\MM}{i,i}$ for
  $i \in C.$ Expressing $\schur{\MM}{1}$ as
  $\DD + (\blk{\MM}{C,C} - \DD + \blk{\MM}{C,1}\blk{\MM}{1,1}^{-1}
  \blk{\MM}{1,C}),$ we have for any $i \in C,$
\begin{align*}
& \sum_{j \in C} \norm{ \blk{(\blk{\MM}{C,C} - \DD + \blk{\MM}{C,1}\blk{\MM}{1,1}^{-1}
  \blk{\MM}{1,C})}{i,j} } 
\le  \left( \sum_{j \in C: j \neq i} \norm{\blk{\MM}{i,j}} \right) + \sum_{j \in C} \norm{ 
\blk{\MM}{i,1}\blk{\MM}{1,1}^{-1} \blk{\MM}{1,j}} \\
& \qquad \qquad \qquad \le  \left( \sum_{j \in C: j \neq i} \norm{\blk{\MM}{i,j}} \right) + \norm{\blk{\MM}{i,1} } 
\norm{\blk{\MM}{1,1}^{-1}}
\sum_{j \in C}
\norm{ \blk{\MM}{1,j}} \\
& \qquad \qquad \qquad \le \left( \sum_{j \in C: j \neq i}
  \norm{\blk{\MM}{i,j}} \right) + \norm{\blk{\MM}{i,1} } = \sum_{j \in V: j \neq i}
  \norm{\blk{\MM}{i,j}} , 
\end{align*}
where the last inequality uses
$\norm{\blk{\MM}{1,1}^{-1}} \left( \sum_{j \neq 1} \norm{
    \blk{\MM}{1,j}} \right) \le 1,$
since
$\blk{\MM}{1,1} \pgeq \id_{r} \cdot \sum_{j \neq 1}
\norm{\blk{\MM}{1,j}}.$
Thus using Lemma~\ref{lem:extra-diagonal}, we have $\schur{\MM}{1} =
\DD - ( - (\blk{\MM}{C,C} - \DD + \blk{\MM}{C,1}\blk{\MM}{1,1}^{-1}
  \blk{\MM}{1,C}))$ is {\bdd}.
\end{proof}
The next definition describes
 a special form that we can express
any {\bdd} matrix in, which will occasionally be useful.

\begin{definition}
\label{def:unitaryB}
A matrix $\BB \in (\complex^{r \times r})^{n \times m}$ is called a
\emph{unitary edge-vertex transfer matrix}, when each block column of
$\blk{\BB}{e}$ has
exactly two nonzero blocks
$\UU_{e}, \VV_{e} \in \complex^{r \times r} $ s.t.
$\UU_{e} \UU_{e}^{\dg} = \VV_{e} \VV_{e}^{\dg} = w_{e} \id_{r}$ where $w_{e}  \geq 0$.
\end{definition}

\begin{lemma}
\label{lem:bsdd-factorization}
  Every {\bdd} matrix $\MM \in (\complex^{r \times r})^{n \times n}$
  with $m$ nonzero off-diagonal blocks can be written as $\XX+\BB \BB^{\dg}$ where
  $\BB \in (\complex^{r \times r})^{n \times 2m}$ is a unitary
  edge-vertex transfer matrix
  and $\XX$ is a block diagonal PSD matrix.
  This implies that every {\bdd} matrix is PSD.
  Furthermore, for every block diagonal matrix $\YY$ s.t. $\MM-\YY$ is
  {\bdd}, we have $\XX \pgeq \YY$.
  This decomposition
can be found in $O(m)$ time and $O(\log n)$ depth.
  
\end{lemma}
\begin{proof}
  Consider a pair $\{i,j\} \in V \times V$ such that $i \neq j,$
  and $\blk{\MM}{i,j} \neq 0.$ Using Fact~\ref{fact:sum-of-unitary},
  we can write such a $\blk{\MM}{i,j}$ as
$\frac{1}{2}\norm{\blk{\MM}{i,j}}(\QQ_{\{i,j\}}^{(1)} + \QQ_{\{i,j\}}^{(2)}),$ where
$\QQ_{\{i,j\}}^{(1)} (\QQ_{\{i,j\}}^{(1)})^{\dg} = \QQ_{\{i,j\}}^{(2)}
(\QQ_{\{i,j\}}^{(2)})^{\dg} = \id_{r}.$ we construct two vectors vector
  $\BB^{(1)}_{\{i,j\}}, \BB^{(2)}_{\{i,j\}} \in (\complex^{r \times
    r})^n$ such that for $k=1,2,$ $\blk{\left(\BB_{\{i,j\}}^{(k)} \right)}{i} =
  \frac{1}{\sqrt{2}}\norm{\blk{\MM}{i,j}}^{1/2} \id_r,$  $\blk{\left(\BB_{\{i,j\}}^{(k)} \right)}{j} =
  \frac{1}{\sqrt{2}}\norm{\blk{\MM}{i,j}}^{1/2} \left( \QQ^{(k)}_{\{i,j\}} \right)^{\dg},$ and all other
  blocks are zero. We can verify that for all $k,\ell \in V,$
\begin{align*}
\blk{\left( \BB^{(1)}_{\{i,j\}} (\BB^{(1)}_{\{i,j\}})^{\dg}  +
  \BB^{(2)}_{\{i,j\}} (\BB^{(2)}_{\{i,j\}})^{\dg} \right)}{k,\ell} = 
\begin{cases}
\norm{\blk{\MM}{i,j}} \id_r & k = \ell = i, \\
\norm{\blk{\MM}{i,j}} \id_r & k = \ell = j, \\
\blk{\MM}{i,j} & k = i, \ell = j, \\
\blk{\MM}{i,j}^{\dg} = \blk{\MM}{j,i} & k = j, \ell = i, \\
0 & \textrm{otherwise.}
\end{cases}
\end{align*}
We let
$\XX = \MM - \sum_{\{i,j\}: i \neq j} \left( \BB^{(1)}_{\{i,j\}}
  (\BB^{(1)}_{\{i,j\}})^{\dg} + \BB^{(2)}_{\{i,j\}}
  (\BB^{(2)}_{\{i,j\}})^{\dg} \right),$
which must be block-diagonal. We now show that for all $i \in V,$ the
block $\blk{\XX}{i,i}$ is PSD. We have for all $i \in V,$
\[\blk{\XX}{i,i} = \blk{\MM}{i,i} - \id_r \cdot \sum_{j: j \neq i}
\norm{\blk{\MM}{i,j}} \pgeq 0,\]
where the last inequality holds since $\MM$ is {\bdd}.



Thus, if we define $\BB \in  (\complex^{r \times r})^{n \times
  2m}$ such that its columns are all the vectors 
$\BB^{(1)}_{\{i,j\}}, \BB^{(2)}_{\{i,j\}}$ defined above, we have $\MM
=\XX +  \BB \BB^{\dg},$ and
every column of $\BB$ has exactly 2 nonzero blocks. 

To show that for every block diagonal $\YY$ s.t. $\MM-\YY$ is
{\bdd}, $\XX \pgeq \YY$, first consider applying
the decomposition described above to $\MM-\YY$ instead of $\MM$.
Since the construction of $\BB$
only depends on the off-diagonal blocks,
we get $\MM - \YY = \boldsymbol{\Lambda} + \BB{\BB^{\dg}}$, where
$\boldsymbol{\Lambda}$ is block diagonal and PSD.
So,  $\XX - \YY = \MM -\BB{\BB^{\dg}} - \YY = \boldsymbol{\Lambda} \pgeq 0$.

 It is immediate that the decomposition
can be found in $O(m)$ time and $O(\log n)$ depth.

\end{proof}


\section{Schur Complement Chains}
\label{sec:chains}
In this section, we give a proof of Lemma~\ref{lem:apply_chain}. We
restate the lemma here for convenience.
\applychain*

The pseudocode for procedure $\textsc{ApplyChain}$ that uses an
$\eepsilon$-vertex sparsifier chain to approximately solve a system of
equations in $\MM^{(0)}$ is given in Figure~\ref{fig:applyChain}.
\begin{figure}[h]
\begin{algbox} $\xx^{(0)} = \textsc{ApplyChain}\left( (\MM^{(1)}, \ZZ^{(1)}) , \ldots, (\MM^{(d)},
  \ZZ^{(d)}); F_1,  \ldots, F_{d} \right)$

\begin{enumerate}
  \item Initialize $\bb^{(0)} \leftarrow \bb.$
	\item For i = $1, \ldots, d$
		\begin{enumerate}
			\item $\blk{\xx}{F_i}^{(i-1)} \leftarrow \ZZ^{(i)} \blk{\bb}{F_i}^{(i-1)}$, 
			\item $\bb^{(i)} \leftarrow \blk{\bb}{C_i}^{(i-1)} - \blk{\MM}{C_i,F_i}^{(i-1)} \blk{\xx}{F_i}^{(i-1)}$.
		\end{enumerate}
	\item $\xx^{(d)} \leftarrow \left( \MM^{(d)} \right)^{-1} \bb^{(d)}$.
	\item For i = $d, \ldots, 1$
		\begin{enumerate}
			\item $\blk{\xx}{C_i}^{(i-1)} \leftarrow \xx^{(i)}$.
			\item $\blk{\xx}{F_i}^{(i-1)} \leftarrow \blk{\xx}{F_i}^{(i-1)} -  \ZZ^{(i)} \blk{\MM}{F_i, C_i}^{(i-1)} \xx^{(i)}$.
		\end{enumerate}
	\end{enumerate}
\end{algbox}

\caption{Solver Algorithm using Vertex Sparsifier Chain}

\label{fig:applyChain}
\end{figure}

\begin{proof}
We begin by observing that the output vector $\xx^{(0)}$ is a linear transformation
  of the input vector $\bb^{(0)}$.
Let $\WW^{(0)}$ be the matrix that realizes this transformation.
Similarly, for $1 \leq i \leq d$, define $\WW^{(i)}$  to be the matrix so that
\[
  \xx^{(i)} = \WW^{(i)} \bb^{(i)}.
\]
An examination of the algorithm reveals that
\begin{equation}\label{eqn:apply_chain1}
  \WW^{(d)} = \left(\MM^{(d)} \right)^{-1},
\end{equation}
and for $1 \le i \le d,$
\begin{equation}\label{eqn:apply_chaini}
  \WW^{(i-1)}
 =
  \left[
\begin{array}{cc}
\II & -\ZZ^{(i)} \blk{\MM}{F_{i}, C_{i}}^{(i-1)} \\
0 & \II
\end{array}
\right]
\left[
\begin{array}{cc}
\ZZ^{(i)} & 0 \\
0 & \WW^{(i)}
\end{array}
\right]
\left[
\begin{array}{cc}
\II & 0\\
-\blk{\MM}{C_{i}, F_{i}}^{(i-1)} \ZZ^{(i)}  & \II
\end{array}
\right].
\end{equation}

We will now prove by backwards induction on $i$ that
\[
\left(\WW^{(i)} \right)^{-1} \approx_{\sum_{j = i+1}^{d} 2\epsilon_j} \MM^{(i)}.
\]
The base case of $i = d$ follows from \eqref{eqn:apply_chain1}.
Using the definition of an $\eepsilon$-SCC, we know that
$0 \preceq (\ZZ^{(i)})^{-1} -
  \blk{\MM}{F_{i},F_{i}}^{(i-1)} \preceq \epsilon_{i}\cdot
  \schur{\MM^{(i-1)}}{C_{i}}.$ We show in
  Lemma~\ref{lem:block-sub-error} that this implies
\[
  \left[
\begin{array}{cc}
\II & -\ZZ^{(i)} \blk{\MM}{F_{i},C_{i}}^{(i-1)}\\
0 & \II
\end{array}
\right]
\left[
\begin{array}{cc}
\ZZ^{(i)} & 0 \\
0 & \schur{\MM^{(i-1)}}{F_{i}}^{-1}
\end{array}
\right]
  \left[
\begin{array}{cc}
\II & 0\\
- \blk{\MM}{C_{i}, F_{i}}^{(i-1)} \ZZ^{(i)} & \II
\end{array}
\right]
\approx_{\epsilon_{i}}
\left(\MM^{(i-1)} \right)^{-1}.
\]
As $\MM^{(i)} \approx_{\epsilon_{i}} \schur{\MM^{(i-1)}}{F_{i}}$,
\[
  \left[
\begin{array}{cc}
\II & -\ZZ^{(i)} \blk{\MM}{F_{i},C_{i}}^{(i-1)}\\
0 & \II
\end{array}
\right]
\left[
\begin{array}{cc}
\ZZ^{(i)} & 0 \\
0 & \left(\MM^{(i)}\right)^{-1}
\end{array}
\right]
  \left[
\begin{array}{cc}
\II & 0\\
- \blk{\MM}{C_{i}, F_{i}}^{(i-1)} \ZZ^{(i)} & \II
\end{array}
\right]
\approx_{2\epsilon_{i}}
\left(\MM^{(i-1)} \right)^{-1}.
\]
By combining this identity with \eqref{eqn:apply_chaini}
  and our inductive hypothesis, we obtain
\[
  \WW^{(i-1)}
\approx_{\sum_{j = i}^{d} 2 \epsilon_{j}}
\left(\MM^{(i-1)} \right)^{-1}.
\]
Thus, by induction, we obtain
\[
  \WW^{(0)}
\approx_{\sum_{j = 1}^{d} 2 \epsilon_{j}}
\left(\MM^{(0)} \right)^{-1}.
\]

The whole algorithm involves a constant number of applications of
$\ZZ^{(i)}, \blk{\MM}{F_{i}, C_{i}}^{(i-1)},$ and
$ \blk{\MM}{C_{i}, F_{i}}^{(i-1)}.$ We observe that in order to
compute $\blk{\MM}{F_{i}, C_{i}}^{(i-1)} x_{C_{i}},$ using a
multiplication procedure for $\MM^{(i-1)},$ we can pad $x_{C_{i}}$
with zeros, multiply by $\MM^{(i-1)},$ and read off the answer on the
indices in $C_{i}.$ Similarly, we can multiply vectors with
$ \blk{\MM}{C_{i}, F_{i}}^{(i-1)}.$ This immediately gives the claimed
bounds on work and depth.
\end{proof}

We now prove the deferred claims from the above proof.
\begin{lemma}\label{lem:block-sub-error}
Let $\MM \in (\complex^{r \times r})^{n \times n}$ be a {\bdd} matrix,
$F \subseteq V$ be a subset of the indices, and $\ZZ \in
(\complex^{r \times r})^{|F| \times |F|}$ be an hermitian operator satisfying
$0 \preceq \ZZ^{-1} -
  \blk{\MM}{F,F} \preceq \epsilon\cdot
  \schur{\MM}{C}.$
Then,
\[
\left[
\begin{array}{cc}
\II & -\ZZ \blk{\MM}{F,C}\\
0 & \II
\end{array}
\right]
\left[
\begin{array}{cc}
\ZZ & 0 \\
0 & \schur{\MM}{F}^{-1}
\end{array}
\right]
\left[
\begin{array}{cc}
\II & 0\\
-\blk{\MM}{C,F} \ZZ & \II
\end{array}
\right]
\approx_{\epsilon} \MM^{-1}.
\]
\end{lemma}

\begin{proof}
Define
\[
  \MMhat = 
\left[
\begin{array}{cc}
(\ZZ)^{-1} &\blk{\MM}{F,C}\\
\blk{\MM}{C,F}   & \blk{\MM}{C,C}
\end{array}
\right].
\]
Using Lemma~\ref{lem:block-sub-equiv}, we know that the assumption on
$\ZZ$ is equivalent to
\[
\MM \pleq \MMhat
\pleq \left(1+ \epsilon\right) \MM.
\]
By Eq.~\eqref{eqn:blockInverse} and Fact~\ref{fact:orderInverse}, this implies
\[
\MM^{-1}
\pgeq 
\left[
\begin{array}{cc}
\II & -\ZZ \blk{\MM}{F,C}\\
0 & \II
\end{array}
\right]
\left[
\begin{array}{cc}
\ZZ & 0 \\
0 & \schur{\MMhat}{F}^{-1}
\end{array}
\right]
\left[
\begin{array}{cc}
\II & 0\\
-\blk{\MM}{C,F}  \ZZ & \II
\end{array}
\right]
\pgeq 
(1+\epsilon)^{-1} \MM^{-1}.
\]
From Facts \ref{fact:schurLoewner} and \ref{fact:orderInverse}, we know that
\[
\schur{\MM}{F}^{-1}
\pgeq 
\schur{\MMhat}{F}^{-1}
\pgeq 
(1+\epsilon)^{-1} \schur{{\MM}}{F}^{-1}.
\]
When we use Fact \ref{fact:orderCAC}  to substitute this inequality into the one above,
  we obtain
\[
(1+\epsilon )\MM^{-1}
\pgeq 
\left[
\begin{array}{cc}
\II & -\ZZ \blk{\MM}{F,C}\\
0 & \II
\end{array}
\right]
\left[
\begin{array}{cc}
\ZZ & 0 \\
0 & \schur{\MM}{F}^{-1}
\end{array}
\right]
\left[
\begin{array}{cc}
\II & 0\\
-\blk{\MM}{C,F}  \ZZ & \II
\end{array}
\right]
\pgeq 
(1+\epsilon)^{-1} \MM^{-1},
\]
which implies the lemma.
\end{proof}

\begin{lemma}
\label{lem:block-sub-equiv}
Given a {\bdd} matrix $\MM,$ a partition of its indices $(F,C)$ such
that $\blk{\MM}{F,F}, \blk{\MM}{C,C}$ are invertible, and
an invertible hermitian operator $\ZZ \in (\complex^{r \times r})^{|F|
  \times |F|},$ the following two conditions are equivalent:
\begin{enumerate}
\item $0 \preceq (\ZZ)^{-1} -
  \blk{\MM}{F,F} \preceq \epsilon\cdot \schur{\MM}{C}.$
\item \[
\MM  \pleq
\left[ 
\begin{array}{cc}
(\ZZ)^{-1} & \blk{\MM}{F, C}\\
\blk{\MM}{C,F}  & \blk{\MM}{C,C} 
\end{array} \right]
\pleq
(1+\epsilon) \MM.
\]
\end{enumerate}
\end{lemma}
\begin{proof}
Writing 
\[M = \left[ 
\begin{array}{cc}
\blk{\MM}{F,F} & \blk{\MM}{F, C}\\
\blk{\MM}{C,F}  & \blk{\MM}{C,C} 
\end{array} \right],
\]
condition 2 is equivalent to
\[
0  \pleq
\left[ 
\begin{array}{cc}
(\ZZ)^{-1} - \blk{\MM}{F,F} & 0\\
0  & 0
\end{array} \right]
\pleq
\epsilon\MM.
\]
The left inequality in this statement is equivalent to the left
inequality in condition 1. Thus, it suffices to prove the right sides
are equivalent.

To this end, using $\blk{\MM}{F,C} = \blk{\MM}{C,F}^{\dg},$ it
suffices to prove that
$\forall x \in (\complex^{r})^{|F|}, y \in (\complex^{r})^{|C|},$
\[x^{\dg}((\ZZ)^{-1} - \blk{\MM}{F,F} )x \le \epsilon
(x^{\dg}\blk{\MM}{F,F}x + 2 x^{\dg} \blk{\MM}{F,C} y +
y^{\dg}\blk{\MM}{C,C}y).\]
This is equivalent to proving $\forall x \in (\complex^{r})^{|F|},$
\[x^{\dg}((\ZZ)^{-1} - \blk{\MM}{F,F} )x \le \epsilon \inf_{y \in (\complex^{r})^{|C|}}
(x^{\dg}\blk{\MM}{F,F}x + 2 x^{\dg} \blk{\MM}{F,C} y +
y^{\dg}\blk{\MM}{C,C}y).\]
Since $\blk{\MM}{C,C} \pgeq 0,$ the rhs is a convex function of $y.$
The minimum is achieved at $y =
-\blk{\MM}{C,C}^{-1}\blk{\MM}{F,C}^{\dg}x,$ and we obtain the
equivalent condition
\[x^{\dg}((\ZZ)^{-1} - \blk{\MM}{F,F} )x \le \epsilon
x^{\dg}(\blk{\MM}{F,F}  -
\blk{\MM}{F,C}\blk{\MM}{C,C}^{-1}\blk{\MM}{F,C}^{\dg})x.\]
Since $\schur{\MM}{C} = \blk{\MM}{F,F} -
\blk{\MM}{F,C}\blk{\MM}{C,C}^{-1}\blk{\MM}{C,F} = \blk{\MM}{F,F}  -
\blk{\MM}{F,C}\blk{\MM}{C,C}^{-1}\blk{\MM}{F,C}^{\dg},$ we obtain our claim.
\end{proof}


\section{Finding $\alpha$-{\bdd} Subsets}
\label{sec:absdd}

In this section we check that a simple randomized sampling procedure
leads to {\abdd} subsets. 
Specifically we will prove Lemma~\ref{lem:abddSubset}:

\abddSubsetAlgo*

Pseudocode for this routine is given in Figure~\ref{fig:randF}.
\begin{figure}[h]
\begin{algbox} $F=\textsc{bDDSubset}(\MM , \alpha)$, where $\MM$
is a {\bdd} matrix with $n$ rows.
\begin{enumerate}
\item Let $F'$ be a uniform random subset of $\setof{1, \dots , n}$ of size
$\frac{n}{4(1+\alpha)}$.
\label{ln:randSubset}
\item Set \[
F=\left\{ i\in F'\text{ such that }
\sum_{j:j \neq i}
\norm{\blk{\MM}{i,j}}
\ge (1+\alpha)  \sum_{j \in F':j \neq i}
\norm{\blk{\MM}{i,j}} \right\}.
\]
\label{ln:pickSubset}
\item If $|F| < \frac{n}{8(1 + \alpha)}$, goto Step~\ref{ln:randSubset}.
\item Return $F$
\end{enumerate}
\end{algbox}
\caption{Routine for finding an $\alpha$-strongly block diagonally dominant
submatrix}

\label{fig:randF}
\end{figure}
We first show that the set returned is guaranteed to be $\alpha$-strongly block diagonally dominant.

\begin{lemma}
If \textsc{bSDDSubset} terminates,
it returns $F$ such that $\blk{\MM}{F,F}$ is {\abdd}.
\end{lemma}

\begin{proof}
Consider some $i \in F$, the criteria for including $i$ in $F$
in Step~\ref{ln:pickSubset} gives:
\[
\sum_{j:j \neq i} \norm{\blk{\MM}{i,j}}
\geq (1+\alpha) \sum_{j \in F', j \neq i} \norm{\blk{\MM}{i,j}}
\geq (1+\alpha) \sum_{j \in F, j \neq i} \norm{\blk{\MM}{i,j}},
\]
where the last inequality follows since $F$ is a subset of $F'$.

Incorporating this into the definition of $\MM$ being {\bdd} gives
\[
\blk{\MM}{i,i} \pgeq \id_{r} \cdot \sum_{j:j \neq i} \norm{\blk{\MM}{i,j}}
\geq (1+\alpha) \id_{r} \cdot \sum_{j \in F, j \neq i} \norm{\blk{\MM}{i,j}}.
\]
which means $\blk{\MM}{F,F}$ is {\abdd}.
\end{proof}

It remains to show that the algorithm finds a big $F$ quickly.
This can be done by upper bounding the expected size of $F$,
or the probability of a single index being in $F$.

\begin{lemma}
\label{lem:expectedSize}
For any $i$, we have
\[
\prob{}{i \in F' \text{ and } i  \notin F} \leq \frac{1}{16 \left(1+\alpha\right)}.
\]
\end{lemma}

\begin{proof}
This event only happens if $i \in F'$ and
%
\begin{equation}\label{eqn:subsetSimple}
  \sum_{j \in F', j \not = i} \norm{\blk{\MM}{i,j}}
>
  \frac{1}{1+\alpha }   \sum_{j : j \neq i}\norm{\blk{\MM}{i,j}}.
\end{equation}

Conditioning on $i$ being selected initially, or $i \in F'$,
the probability that each other $j \not = i$ is in $F'$ is
\[
  \frac{1}{n-1}  \left(\frac{n}{4 (1+\alpha)}-1 \right) < \frac{1}{4(\alpha+1)},
\]
which gives:
\[
  \expec{}{\sum_{j \in F', j \not = i } \norm{\blk{\MM}{i,j}} \Big| i \in F'}
<
  \frac{1}{4 (1+\alpha)}  \sum_{j : j \not = i} \norm{\blk{\MM}{i,j}}.
\]
Markov's inequality then gives:
\[
  \prob{}{
\sum_{j \in F', j \not = i }  \norm{\blk{\MM}{i,j}} 
>
  \frac{1}{1+\alpha}  \sum_{j: j \neq i}\norm{\blk{\MM}{i,j}} 
\Big| i \in F'
  }
< 1/4,
\]
and thus
\[
  \prob{}{i \in F' \text{ and } i  \notin F} = \prob{}{i \not \in F | i \in F'} \prob{}{i \in F'} 
<
  \frac{1}{4} \frac{1}{4 (1+\alpha)}
=
  \frac{1}{16 (1+\alpha)}.
\]
\end{proof}

Combining these two bounds gives Lemma~\ref{lem:abddSubset}.

\begin{proof}(of Lemma~\ref{lem:abddSubset})
Applying Linearity of Expectation to Lemma~\ref{lem:expectedSize} gives
\[
\expec{}{\left| F' \setminus F \right|} \leq \frac{n}{16 \left( 1 + \alpha \right)}. 
\]
Markov's inequality then gives
\[
  \prob{}{\left| F' \setminus F \right| \geq \frac{n}{8 (1+\alpha)}} < 1/2.
\]
So, with probability at least $1/2$, $\sizeof{F} \geq n / (8 (1+\alpha))$,
  and the algorithm will pass the test in line 3.
Thus, the expected number of iterations made by the algorithm is at most $2$.
The claimed bounds on the expected work and depth of the algorithm follow.
\end{proof}

\section{Jacobi Iteration on {\abdd} Matrices}
\label{sec:jacobi}

From an {\abdd} set $F$, we will construct 
  an operator $\ZZ$ that approximates
  $\blk{\MM}{F,F}^{-1}$  and that can be applied quickly.
 Specifically, we will show:

\jacobi*

Pseudocode of this routine is given in Figure~\ref{fig:jacobi}.
\begin{figure}[h]
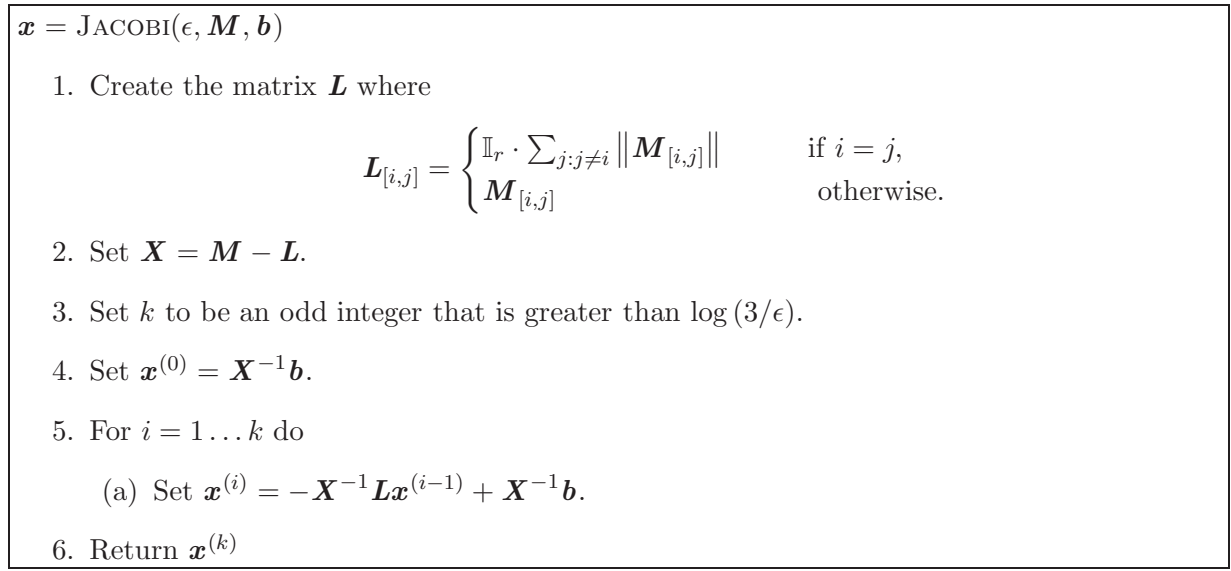

\begin{algbox} $\xx = \textsc{Jacobi}(\epsilon, \MM, \bb) $
\begin{enumerate}
	\item Create the matrix $\LL$ where
		\[
			\blk{\LL}{i, j} = 
			\begin{cases}
				\id_r \cdot \sum_{j:j \neq i} \norm{\blk{\MM}{i,j}} & \qquad \text{if } i = j, \\
				\blk{\MM}{i,j} & \qquad \text{ otherwise.}\\
			\end{cases}
		\]
	\item Set $\XX = \MM - \LL$.
	\item Set $k$ to be an odd integer that is greater than $\log \left( 3 / \epsilon \right)$.
	\item Set $\xx^{(0)} = \XX^{-1} \bb$.
	\item For $i = 1 \ldots k$ do
		\begin{enumerate}
			\item Set $\xx^{(i)} = -\XX^{-1} \LL \xx^{(i - 1)} + \XX^{-1} \bb$.
		\end{enumerate}
	\item Return $\xx^{(k)}$
\end{enumerate}
\end{algbox}
\caption{Jacobi Iteration for Solving Linear Systems in an {\abdd} Matrix}
\label{fig:jacobi}
\end{figure}

We first verify that any {\abdd} matrix has a good block-diagonal
preconditioner.

\split*

\begin{proof}
Write $\LL = \YY - \AA$ where $\YY$ is block-diagonal
and $\AA$ has its diagonal blocks as zeros. Note that $\LL$ is {\bdd}
by definition. Thus, by Corollary~\ref{lem:bsdd-factorization}, $\LL =
\YY - \AA \pgeq 0.$

Using Lemma~\ref{cor:bsdd-sign-flip}, we know that
$\YY + \AA \pgeq 0$, or $\YY \pgeq -\AA$. This implies $2 \YY \pgeq \LL$.

As $\MM$ is $\alpha$-strongly diagonally dominant and
$\blk{\MM}{i, i} = \blk{\XX}{i,i} + \blk{\YY}{i, i}$, we have
\[
\blk{\left( \XX + \YY \right) }{i, i} \pgeq \left( 1 + \alpha \right)
\id_{r} \cdot \sum_{j \neq i} \norm{\blk{\MM}{i, j}}
= \left( 1 + \alpha \right) \blk{\YY}{i, i}.
\]
Manipulating this then gives:
\[
\XX \pgeq \alpha \YY \pgeq \frac{\alpha}{2} \LL .
\]
\end{proof}

We now move on to measuring the quality of the operator generated
by \textsc{Jacobi}.
It can be checked that running it $k$ steps gives the operator
\begin{equation}
\ZZ^{(k)} \defeq  \sum_{i = 0}^{k} \blk{\XX}{F,F}^{-1} \left(-\blk{\LL}{F,F} \blk{\XX}{F,F}^{-1} \right)^{i},
\label{eqn:defZ}
\end{equation}
which is equivalent to evaluating a truncation of the Neumann series for $\MM^{-1}$

\begin{lemma}\label{lem:seriesError}
Let $\MM$ be a matrix with splitting $\MM = \XX + \LL$
where $0 \pleq \LL \pleq \beta \XX$ for some parameter $1 > \beta > 0$.
Then, for odd $k$ and for $\ZZ^{(k)}$ as defined in \eqref{eqn:defZ} we have:
\begin{equation}\label{eqn:lightBlock2}
\XX + \LL
\preceq (\ZZ^{(k)})^{-1} \preceq
\XX  + \left( 1 + \delta  \right) \LL
\end{equation}
where
\[
  \delta = \beta^{k} \frac{1+\beta}{1-\beta^{k+1}}.
\]
\end{lemma}
\begin{proof}
The left-hand inequality is equivalent to the statement that
 all the eigenvalues of
 $\ZZ^{(k)} (\XX + \LL)$ are at most $1$
   (see \cite[Lemma 2.2]{SupportGraph} or
  \cite[Proposition 3.3]{SpielmanTengLinsolve}).
To see that this is the case, expand
\begin{align*}
\ZZ^{(k)} (\XX + \LL)
& = 
\left(\sum_{i=0}^{k} \XX^{-1} (-\LL \XX^{-1})^{i} \right)
(\XX + \LL)
\\
& = 
\sum_{i=0}^{k} (-\XX^{-1} \LL)^{i}
-
\sum_{i=1}^{k+1} (-\XX^{-1} \LL)^{i}
\\
& =
\II - (\XX^{-1} \LL)^{k+1} .
\end{align*}
As all the eigenvalues of an even power of a matrix are nonnegative,
 all of the eigenvalues of this last matrix are at most $1$.

Similarly, the other inequality is equivalent to the assertion
 that all of the eigenvalues of 
 $\ZZ^{(k)} (\XX + (1+\delta ) \LL)$
 are at least one.
Expanding this product yields
\[
\left(\sum_{i=0}^{k} \XX^{-1} (-\LL \XX^{-1})^{i} \right)
(\XX + (1+\delta ) \LL)
= 
\II - (\XX^{-1} \LL)^{k+1} 
+ \delta  \sum_{i=0}^{k} (-1)^{i} (\XX^{-1} \LL)^{i+1} 
\]
The eigenvalues of this matrix are precisely the numbers
\begin{equation}\label{eqn:jacobiError}
1 - \lambda^{k+1} + \delta  \sum_{i=0}^{k} (-1)^{i} \lambda^{i+1},
\end{equation}
where $\lambda$ ranges over the eigenvalues of 
 $\XX^{-1} \LL$.
The assumption $\LL \pleq \beta \XX$ implies
 that the eigenvalues of  $\XX^{-1} \LL$
  are at most $\beta$, so $0 \leq \lambda \leq \beta$.
We have chosen the value of $\delta$ precisely to guarantee that,
 under this condition on $\lambda$, the value of \eqref{eqn:jacobiError}
 is at least $1$.
\end{proof}

This error crucially depends only on $\LL$, which for any choice of $F$
can be upper bounded by $\MM$.
Propagating the error this way allows us to prove the guarantees for
\textsc{Jacobi}

\begin{proof} (of Theorem~\ref{thm:jacobi})
Consider the matrix $\blk{\LL}{F, F}$ generated when calling
$\textsc{Jacobi}$ with $\blk{\MM}{F, F}$.
$\MM$ being {\bdd} means for each $i$ we have
\[
\blk{\MM}{i, i} \pgeq \id_r \cdot \sum_{j \neq i} \norm{\blk{\MM}{i, j}},
\]
which means for all $i \in F$
\[
\blk{\MM}{i, i} - \id_r \cdot \sum_{j \neq i, j \in F}  \norm{\blk{\MM}{i, j}}
\pgeq \id_r \cdot \sum_{j \notin F} \norm{\blk{\MM}{i, j}}.
\]
Therefore if we extend $\blk{\LL}{F, F}$ onto the full matrix by
putting zeros everywhere else, we have $\LL \pleq \MM$.
Fact~\ref{fact:schurLoewner} then gives
$\blk{\LL}{F, F} \pleq \schur{\MM }{V\setminus F}$.

Lemma~\ref{lem:split} gives that $\blk{\XX}{F, F} \pgeq \frac{\alpha}{2} \blk{\LL}{F, F}$.
As $\alpha \geq 4$, we can invoke Lemma~\ref{lem:seriesError} with
$\beta = \frac{1}{2}$.
Since $\frac{1+\beta}{1-\beta^{k+1}} \leq (\frac{3}{2}) / (\frac{1}{2}) \leq 3$,
our choice of $k = \log(3 / \epsilon)$ gives the desired error.
Each of these steps involve a matrix vector multiplication in $\blk{\LL}{F, F}$
and two linear system solves in $\blk{\XX}{F, F}$.
The former takes $O(m)$ work and $O(\log{n})$ depth since the blocks of
$\blk{\LL}{F, F}$ are a subset of the blocks of $\MM$, while the latter
takes $O(n)$ work and $O(\log{n})$ depth due to $\blk{\XX}{F,F}$ being
block-diagonal.
\end{proof}

\section{Existence of Linear-Sized Approximate Inverses}
\label{sec:existence}

In this section we prove Theorem \ref{thm:UDU}, which tells us that
  every $\bdd$ matrix has a linear-sized approximate inverse.
In particular, this implies that for every $\bdd$ matrix
  there is a linear-time algorithm that approximately solves
  systems of equations in that matrix.
To save space, we will not dwell on this algorithm, but rather
  will develop the linear-sized approximate inverses directly.
There is some cost in doing so: there are very large constants in
  the linear-sized approximate inverses that are not present
  in a simpler linear-time solver.

To obtain a $\UU^{T} \DD \UU$ factorization from an $\eepsilon$-SCC
  in which each $\blk{\MM^{(i)}}{F_{i} F_{i}}$ is $4$-{\bdd},
  we employ the procedure in Figure~\ref{fig:UDU}.

\begin{figure}[h]
\begin{algbox}
$(\DD, \UU) = \textsc{Decompose}\left(\MM^{(1)}, \dots , \MM^{(d)}, F_{1}, \dots , F_{d-1}\right)$,
 where each $\MM^{(i)}$ is a {\bdd} matrix
 and each $\blk{\MM^{(i)}}{F_{i} F_{i}}$ is $4$-{\bdd}.
\begin{enumerate}

\item Use \eqref{eqn:defZ} to compute the matrix $\ZZ^{(i)}$ such that
  $\textsc{Jacobi}(\epsilon_{i}, \MM, \bb) = \ZZ^{(i)} \bb$.

\item For each $i < d$, write $\MM^{(i)} = \XX^{(i)} + \LL^{(i)}$,
  where $\XX^{(i)}$ is block diagonal and $\LL^{(i)}$ is $\bdd$,
   as in \textsc{Jacobi}.

\item Let $\XX^{(d)} = \id_{\abs{C_{d - 1}}}$
  and let $\UUhat$ be the upper-triangular Cholesky factor of $\MM^{(d)}$.

\item Let $\DD$ be the block diagonal matrix with $\blk{\DD}{F_{i},F_{i}} = \XX^{(i)}$,
  for $1 \leq i < d$, and $\blk{\DD}{C_{d-1} ,C_{d-1}} = \id_{\abs{C_{d-1}}}$.

\item Let $\UU$ be the upper-triangular matrix with $1s$ on the diagonal,
  $\blk{\UU}{C_{d-1}, C_{d-1}} =  \UUhat$, and
  $\blk{\UU}{F_{i}, C_{i}} = \ZZ^{(i)} \blk{\MM^{(i)}}{F_{i}, C_{i}}$,
  for $1 \leq i < d$.
\end{enumerate}
\end{algbox}

\caption{Converting a vertex sparsifier chain into $\UU$ and $\DD$.}
\label{fig:UDU}
\end{figure}

\begin{lemma}\label{lem:uduApprox}
On input an $\eepsilon$-SCC of $\MM^{(0)}$
  in which  each $\blk{\MM^{(i)}}{F_{i} F_{i}}$ is $4$-{\bdd},
  the algorithm \textsc{Decompose} produces matrices $\DD$ and $\UU$
  such that
\[
  \UU^{T} \DD \UU \approx_{\gamma} \MM ,
\]
where
\[
  \gamma \leq 2 \sum_{i=0}^{d-1} \epsilon_{i} + \max_{i} \epsilon_{i} + 1/2.
\]
\end{lemma}


\begin{proof}
Consider the inverse of the operator $\WW = \WW^{(1)}$ realized by the algorithm
  \textsc{ApplyChain},
  and the operators $\WW^{(i)}$ that appear in the proof of Lemma~\ref{lem:apply_chain}.

We have
\[
 \left(  \WW^{(i)} \right)^{-1}
 =
\left[
\begin{array}{cc}
\II & 0\\
\blk{\MM}{C_{i},F_{i}} \ZZ^{(i)}  & \II
\end{array}
\right]
\left[
\begin{array}{cc}
\left(\ZZ^{(i)}  \right)^{-1} & 0 \\
0 & \left(\WW^{(i+1)} \right)^{-1}
\end{array}
\right]
  \left[
\begin{array}{cc}
\II & \ZZ^{(i)} \blk{\MM}{F_{i},C_{i}}\\
0 & \II
\end{array}
\right],
\]
and
\[
 \left(  \WW^{(d)} \right)^{-1}
  =
 \MM^{(d)}
 = \UUhat^{T}  \UUhat.
\]
After expanding and multiplying the matrices in this recursive factorization,
  we obtain
\[
 \left( \WW^{(1)} \right)^{-1}
= 
  \UU^{T}
\begin{bmatrix}
\left(\ZZ^{(1)} \right)^{-1} &  \dots & 0 & 0\\
0 & \ddots   & 0 & 0\\
0 &  \ldots & \left(\ZZ^{(d - 1)} \right)^{-1} & 0 \\
0 &  \ldots & 0 & \id_{\abs{C_{d-1}} }
\end{bmatrix}
  \UU.
\]
Moreover, we know that this latter matrix is a $2 \sum_{i=0}^{d-1} \epsilon_{i}$
  approximation of $\MM$.
It remains to determine the impact of replacing the matrix in the middle of
  this expression with $\DD$.

Lemma~\ref{lem:seriesError} implies that each
  $\blk{\MM}{F_{i}F_{i}} \approx_{\epsilon_{i}} \left(\ZZ^{(i)}\right)^{-1}$
  and Lemma \ref{lem:split}  implies that
  $\XX^{(i)} \approx_{1/2} \blk{\MM}{F_{i}F_{i}}$.
So, the loss in approximation quality when we substitute
  the diagonal matrices is $\max_{i} \epsilon_{i} + 1/2$.
\end{proof}

Invoking this  decomposition procedure in conjunction with the
the near-optimal sparsification routine from Theorem~\ref{thm:blockBSS}
gives a nearly-linear work routine.
Repeatedly picking subsets using Lemma~\ref{lem:subsetLowDeg}
gives then gives the linear sized decomposition.

\thmUDU*

\begin{proof}
We set $\alpha = 4$ throughout and
  $\epsilon_{i} = 1/8 (i+2)^{2}$.
Theorem~\ref{thm:blockBSS} then guarantees that the average number of
  nonzero blocks in each column of $\MM^{(i)}$ is at most
  $10 r / \epsilon_{i}^{2} = 640 r (i+2)^{4}$.
If we now apply Lemma~\ref{lem:subsetLowDeg} to find $4$-diagonally dominant
  subsets $F_{i}$ of each $\MM^{(i)}$,
  we find that each such subset contains at least a $1/80$ fraction of the block columns
  of its matrix and that each
  column and row of $\MM^{(i)}$ indexed by $F$ has at most $1280 r (i+2)^{4}$
  nonzero entries.
This implies that each row of
  $\ZZ^{(i)} \blk{\MM}{F_{i},C_{i}}$
  has at most
  $(1280 r (i+2)^{4})^{k_{i}+1}$ nonzero entries.

Let $n_{i}$ denote the number of block columns of $\MM^{(i)}$.
By induction, we know that
\[
  n_{i} \leq n \left(1 - \frac{1}{80} \right)^{i-1}.
\]
So, the total number of nonzero blocks in $\UU$ is at most
\[
  \sum_{i=1}^{d} n_{i} (1280 r (i+2)^{4})^{k_{i}+1}
\leq 
  n \sum_{i=1}^{d} \left(1 - \frac{1}{80} \right)^{i-1} (1280 r (i+2)^{4})^{k_{i}+1}.
\]
We will show that the term multiplying $n$ in this later expression is upper bounded by
  a constant.
To see this, note that $k_{i} \leq 1 + \log (2 \epsilon_{i}^{-1} )\leq \nu \log (i+1)$
  for some constant $\nu$.
So, there is some other constant $\mu$ for which
\[
(1280 r (i+2)^{4})^{k_{i}+1}
\leq 
\exp (\mu \log^{2} (i+1)).
\]
This implies that the sum is at most
\[
  \sum_{i \geq 1} \exp (\mu \log^{2} (i+1) - i / 80),
\]
which is bounded by a constant.

To bound the quality of the approximation, we compute
\[
  2 \sum_{i} \epsilon_{i} + \max_{i} \epsilon_{i} + 1/2
=
  2 \sum_{i} 1/8 (i+2)^{2} + 1/72 + 1/2
<
   3/4. 
\]
The claimed bound on the work to perform backwards and forwards substitution with $\UU$
  is standard: these operations require work linear in the number of nonzero entries of $\UU$.
The bound on the depth follows from the fact that the substitutions can be performed level-by-level,
  take depth $O (\log n)$ for each level, and the number of levels, $d$, is logarithmic in $n$.
\end{proof}


\newcommand{\eps}{\epsilon}
\newcommand{\twoCover}[1]{\ensuremath{{#1}^{(K2)}}}
\newcommand{\prodClique}{\textsc{CliqueSparsification}}
\newcommand{\bipClique}{\textsc{BipartiteCliqueSparsification}}

\section{Spectral Vertex Sparsification Algorithm}
\label{sec:vertexSparsify}

In this section, we give a proof of the following lemma that
immediately implies Lemma~\ref{lem:approxSchur}. 
\begin{lemma}
\label{lem:approxSchur:extended}
Let $\MM$ be a {\bdd} matrix with index set $V$, and $m$ nonzero
blocks. Let $F \subseteq V$ be such that $\blk{\MM}{F,F}$ is {\abdd}
for some $\alpha \ge 4.$ The algorithm $\schurApx(\MM,F,\epsilon)$,
returns a matrix $\MMtil_{SC}$ s.t.
\begin{enumerate}
\item $\MMtil_{SC}$ has $O(m ( \epsilon^{-1}
  \log \log \epsilon^{-1} )^{O(\log \log \epsilon^{-1})} )$ nonzero blocks, and
\item  $\MMtil_{SC} \approx_{\epsilon} \schur{\MM}{F}$,
\end{enumerate}
in
$O(m ( \epsilon^{-1} \log \log \epsilon^{-1} )^{O(\log \log
  \epsilon^{-1})} )$
work and $O(\log n (\log \log \epsilon^{-1}) )$ depth.

Moreover, if $\TT$ is a matrix such that only the submatrix
$\blk{\TT}{C,C}$ is nonzero, and $\MM + \TT$ is {\bdd}, then
$\schurApx(\MM+\TT, F, \epsilon) = \schurApx(\MM, F,
\epsilon)+\blk{\TT}{C,C}$.
\end{lemma}
%
We
show how to sparsify the Schur complement of $\MM$ after eliminating a
set of indices $F$ such that $\blk{\MM}{F,F}$ is an {\abdd}
matrix. The procedure {\schurApx} is described in
Figure~\ref{fig:approxSchur} and uses two key subroutines
$\schurSquare$ and $\lastStep.$ $\schurSquare$ allows us to
approximate $\schur{\MM}{F}$ as the Schur complement of another matrix
$\MM_{1}$ such that $\blk{\MM_{1}}{F,F}$ is roughly
$\alpha^{2}$-\bdd. {\lastStep} allows us to approximate
$\schur{\MM}{F}$ with $O(\epsilon)$ error, when $\blk{\MM}{F,F}$ is
roughly $1/\epsilon$-\bdd.
\begin{figure}[ht]
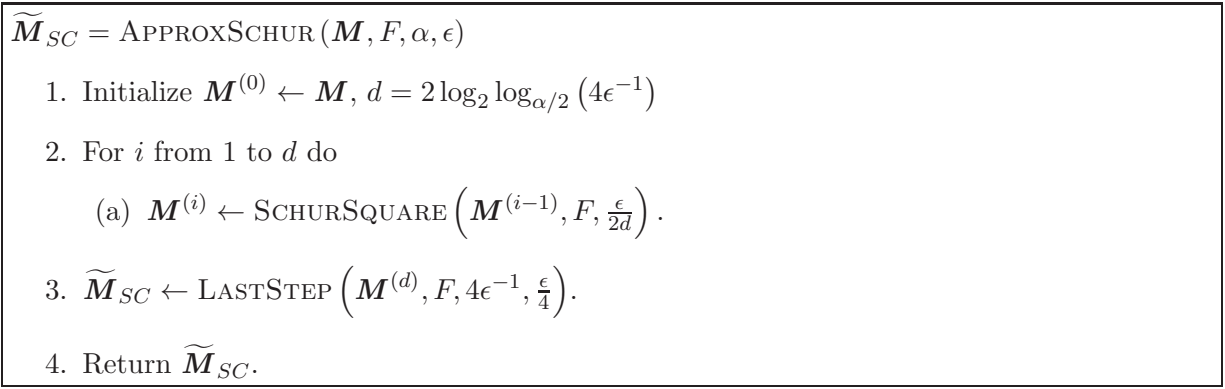


\begin{algbox}
$\MMtil_{SC} = \schurApx\left(\MM, F, \alpha, \epsilon \right)$
\begin{enumerate}
\item Initialize $\MM^{(0)} \leftarrow \MM$,
$d = 2 \log_2 \log_{\alpha/2}\left(4\epsilon^{-1}\right)$
\item For $i$ from $1$ to $d$ do
\begin{enumerate}
\item $\MM^{(i)} \leftarrow
  \schurSquare\left(\MM^{(i-1)},F, \frac{\epsilon}{2d} \right).$
\end{enumerate}
\item $\MMtil_{SC} \leftarrow \lastStep\left(\MM^{(d)}, F,
    4\epsilon^{-1}, \frac{\epsilon}{4}\right)$.
\item Return $\MMtil_{SC}$.
\end{enumerate}
\end{algbox}

\caption{Pseudocode for Computing Spectral Vertex Sparsifiers}

\label{fig:approxSchur}

\end{figure}
\begin{lemma}
\label{lem:schurSquare}
  Let $\MM$ be an {\bdd} matrix with index set $V$, and $m$ nonzero
  blocks and let $F \subseteq V$ be such that $\blk{\MM}{F,F}$ is
  {\abdd} for some $\alpha \ge 4.$ Given $\epsilon < \nfrac{1}{2},$ the algorithm
  $\schurSquare(\MM,F,\epsilon)$, returns a {\bdd} matrix
  $\MM_1$ in $O(m \epsilon^{-4})$ work and $O(\log n)$ depth,
  such that
\begin{enumerate}
\item  $\schur{\MM}{F} \approx_{\epsilon} \schur{\MM_{1}}{F},$
  and
\item $\blk{(\MM_{1})}{F,F}$ is $\nfrac{\alpha^2}{2}$-\bdd.
\item $\MM_1$ has $O(m\epsilon^{-4})$ nonzero blocks, 
\end{enumerate}
If $\TT$ is a matrix such that only the submatrix
$\blk{\TT}{C,C}$ is nonzero, and $\MM + \TT$ is {\bdd},
then $\schurSquare(\MM+\TT,F,\epsilon) =
\schurSquare(\MM,F,\epsilon)+\blk{\TT}{C,C}$.
\end{lemma}

We can repeatedly applying the above lemma, to approximate
$\schur{\MM}{F}$ as $\schur{\MM_{1}}{F},$ where $\MM_{1}$
is {$O(\eps^{-1})$-\bdd}. {\lastStep} allows us to approximate the
Schur complement for such a strongly block diagonally
dominant matrix $\MM_{1}.$ The guarantees of {\lastStep} are given by
the following lemma.
\begin{lemma}
\label{lem:lastStep}
  Let $\MM$ be an {\bdd} matrix with index set $V$, and $m$ nonzero
  blocks and let $F \subseteq V$ be such that $\blk{\MM}{F,F}$ is
  {\abdd} for some $\alpha \ge 4.$
  There exist a procedure $\lastStep$ such that
  $\lastStep(\MM, F, \alpha, \epsilon)$
  returns in $O(m\epsilon^{-8})$ work and $O(\log n)$ depth a matrix
  $\MMtil_{SC}$ s.t. $\MMtil_{SC}$ has $O(m\epsilon^{-8})$
  nonzero blocks and $\MMtil_{SC} \approx_{\epsilon + 2/\alpha}
  \schur{\MM}{F}$.
If $\TT$ is a matrix such that only the submatrix
$\blk{\TT}{C,C}$ is nonzero, and $\MM + \TT$ is {\bdd},
then $\lastStep(\MM+\TT, F, \alpha, \epsilon) = \lastStep(\MM, F, \alpha, \epsilon)+\blk{\TT}{C,C}$.
\end{lemma}

Combining the above two lemmas, we obtain a proof of Lemma~\ref{lem:approxSchur}.
\begin{proof}\emph{(of Lemma~\ref{lem:approxSchur:extended}).}
By induction, after 
$i$ steps of the main loop in \textsc{ApproxSchur},
\[
\schur{\MM}{F} \approx_{\frac{\epsilon i}{2d}}\schur{\MM^{(i)}}{F}.
\]

Lemma~\ref{lem:schurSquare} also implies that $\MM^{(i)}$ is
$2(\frac{\alpha}{2})^{2^i}$-\bdd. Thus, we have that $\MM^{(d)}_{FF}$
is $8\epsilon^{-1}$-strongly diagonally dominant at the last
step. Hence, Lemma~\ref{lem:lastStep} then gives
\[
\schur{\MM^{(d)}}{F} \approx_{\frac{1}{2} \epsilon} \lastStep\left(\MM^{(d)}, F,
    4\epsilon^{-1}, \frac{\epsilon}{4}\right).
\]
Composing this bound with the guarantees of the iterations
then gives the bound on overall error.

The property that if $\TT$ is a matrix such that only the submatrix
$\blk{\TT}{C,C}$ is nonzero, and $\MM + \TT$ is {\bdd}, then
$\schurApx(\MM+\TT, F, \epsilon) = \schurApx(\MM, F,
\epsilon)+\blk{\TT}{C,C}$ follows from Lemma~\ref{lem:schurSquare}
and~\ref{lem:lastStep},
which ensure that this property holds for all our calls to
{\schurSquare} and {\lastStep}.

The work of these steps, and the size of the output graph
follow from Lemma~\ref{lem:schurSquare} and~\ref{lem:lastStep}.
\end{proof}

\subsection{Iterative Squaring and Sparsification}
In this section, we give a proof of Lemma~\ref{lem:schurSquare}. At
the core of procedure is a \emph{squaring} identity that preserves
Schur complements, and efficient sparsification of special classes of
{\bdd} matrices that we call product demand block-Laplacians.
\begin{definition}
\label{def:block-demand}
The product demand block-Laplacian of a vector
$\dd \in (\complex^{r \times r})^{n}$, is a {\bdd} matrix
$\LL_{G(\dd)} \in (\complex^{r \times r})^{n \times n},$ defined as
\[
\blk{(\LL_{G(\dd)})}{i,j} =
\begin{cases}
-\blk{\dd}{i} \blk{\dd}{j}^{\dg} & \quad i\neq j,\\
\id_r \cdot \norm{\blk{\dd}{i}} \sum_{k: k \neq i} \norm{\blk{\dd}{k}} & \quad \textrm{otherwise.}
\end{cases}
\]
\end{definition}
\begin{definition}
\label{def:bip-block-demand}
Given a vector $\dd \in (\complex^{r \times r})^{n}$ and an index set
$F \subseteq V,$ let $C = V \setminus F$. The bipartite product
demand block-Laplacian of $(\dd,F)$ is a {\bdd} matrix
$\LL_{G(\dd),F} \in (\complex^{r \times r})^{n \times n},$ defined as
\[ 
\blk{(\LL_{G(\dd,F)})}{i,j} =
\begin{cases}
\id_r \cdot \norm{\blk{\dd}{i}} \sum_{k \in C} \norm{\blk{\dd}{k}} & \quad i = j,
\text{ and } i \in F\\
\id_r \cdot \norm{\blk{\dd}{i}} \sum_{k \in F} \norm{\blk{\dd}{k}} & \quad i = j,
\text{ and } i \in C\\
0 & \quad i\neq j, \text{ and } ( i , j \in F  \text{ or }   i, j \in C)\\
-\blk{\dd}{i} \blk{\dd}{j}^{\dg} & \quad i\neq j, 
\text{otherwise.}
\end{cases}
\]
\end{definition}

In Section~\ref{sec:product-clique-sparsify}, we prove the following
lemmas that allows us to efficiently construct sparse approximations to
these matrices.
\begin{lemma}
\label{lem:product-clique-sparsify}
There is a routine $\prodClique(\dd,\eps)$ such that for any demand
vector $\dd \in (\complex^{r \times r})^{n},$ and $\epsilon > 0,$
$\prodClique(\dd, \epsilon)$ returns in $O(n \epsilon^{-4})$ work and
$O(\log{n})$ depth a {\bdd} matrix $\LL_{H}$ with
$O(n \epsilon^{-4})$ nonzero blocks such that
\[
\LL_{H} \approx_{\epsilon} \LL_{G(\dd)}.
\]
\end{lemma}

\begin{lemma}
\label{lem:bipartite-clique-sparsify}
There is a routine $\bipClique(\dd,F,\eps)$ such that for any demand
vector $\dd \in (\complex^{r \times r})^{n},$ $\epsilon > 0,$
and $F \subseteq V$,
$\bipClique(\dd, \epsilon)$ returns in $O(n \epsilon^{-4})$ work and
$O(\log{n})$ depth a {\bdd} matrix $\LL_{H}$ with
$O(n \epsilon^{-4})$ nonzero blocks such that
\[
\LL_{H} \approx_{\epsilon} \LL_{G(\dd,F)}.
\]
Moreover, $\blk{(\LL_{H})}{F,F},$ and $\blk{(\LL_{H})}{C,C}$ are
block-diagonal, where $C = V\setminus F.$
\end{lemma}
We now use these efficient sparse approximations to give a proof of
the guarantees of the procedure $\schurSquare.$

\begin{figure}[ht]
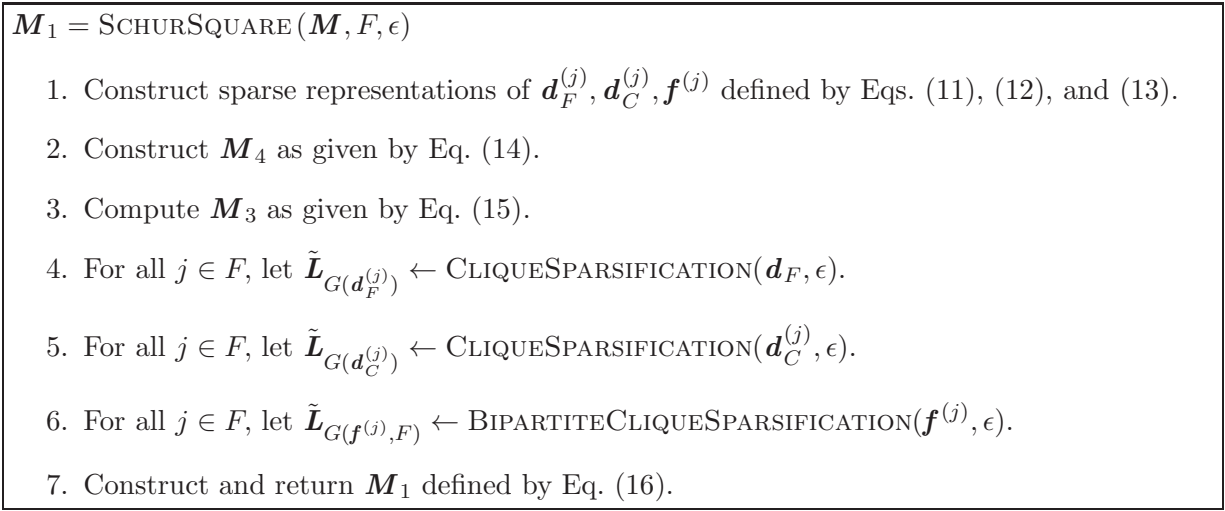


\begin{algbox}
$\MM_{1} = \schurSquare\left(\MM, F, \epsilon \right)$
\begin{enumerate}
\item Construct sparse representations of
  $\dd_F^{(j)}, \dd_C^{(j)}, \ff^{(j)}$
  defined by
  Eqs.~\eqref{eq:schurSquare:clique1},~\eqref{eq:schurSquare:clique2},
  and~\eqref{eq:schurSquare:bipartite}.
\item Construct $\MM_{4}$ as given by Eq.~\eqref{eq:schurSquare:diagonal}.
\item Compute $\MM_{3}$ as given by Eq.~\eqref{eq:schurSquare:remaining}.
\item For all $j \in F,$ let $\tilde{\LL}_{G(\dd_{F}^{(j)})}
  \leftarrow \prodClique(\dd_{F},\eps).$
\item For all $j \in F,$ let $\tilde{\LL}_{G(\dd_{C}^{(j)})}
  \leftarrow \prodClique(\dd^{(j)}_{C},\eps).$
\item For all $j \in F,$ let $\tilde{\LL}_{G(\ff^{(j)},F)} 
  \leftarrow \bipClique(\ff^{(j)},\eps).$
\item Construct and return $\MM_{1}$ defined by Eq.~\eqref{eq:schurSquare:output}.
\end{enumerate}
\end{algbox}

\caption{Pseudocode for procedure {\schurSquare} Iterative Squaring and Sparsification}

\label{fig:schurSquare}

\end{figure}

\begin{proof}\emph{(of Lemma~\ref{lem:schurSquare})}
Let $C$ denote $V\setminus F.$ Write $\blk{\MM}{F,F}$ as $\DD - \AA,$ where
$\DD$ is a block-diagonal, and $\AA$ has its diagonal blocks as
zero. The proof is based on the identity
$\schur{\MM}{F} = \schur{\MM_{2}}{F},$ where $\MM_{2}$ is the
following matrix.
\begin{equation}
\label{eq:schurSquare}
\MM_{2} = \frac{1}{2}
\left[
\begin{array}{cc}
\DD - \AA \DD^{-1} \AA &
\blk{\MM}{F,C} + \AA \DD^{-1}\blk{\MM}{F,C}\\
\blk{\MM}{C,F} + \blk{\MM}{C,F}\DD^{-1} \AA  
&  2 \blk{\MM}{C,C} - \blk{\MM}{C,F} \DD^{-1} \blk{\MM}{F,C}
\end{array}
\right].
\end{equation}
It is straightforward to prove that $\MM_{2}$ satisfies
$\schur{\MM}{F} = \schur{\MM_{2}}{F}$ (see
Lemma~\ref{lem:schurSquare:equality}).  
It is also straightforward to show that $\DD - \AA \DD^{-1} \AA$ is
{$((\alpha+1)^2-1)$-\bdd}. Thus, $\MM_{2}$ satisfies the first two
requirements.

However,
$\MM_{2}$ is likely to be a dense matrix, and we will not construct it
in full. The key observation is that
$\MM_{2}$ can be written as a sum of an explicit sparse {\bdd}
matrix, and several product demand block-Laplacians. Formally, for
every $j \in F,$ we define $\dd_F^{(j)}, \dd_{C}^{(j)} \in (\complex^{r \times
  r})^{n}$ as follows
\begin{align}
\label{eq:schurSquare:clique1}
\blk{(\dd_{F}^{(j)})}{i} 
& = 
\begin{cases}
\blk{\AA}{i,j}\blk{\DD}{j,j}^{-\nfrac{1}{2}} & i \in F, \\
0 & i \in C.
\end{cases}
\\
\label{eq:schurSquare:clique2}
 \blk{(\dd_{C}^{(j)})}{i} 
& = 
\begin{cases}
0 & i \in F, \\
\blk{\MM}{i,j}\blk{\DD}{j,j}^{-\nfrac{1}{2}} & i \in C.
\end{cases}
\end{align}
For every $j \in F,$ we need to define a bipartite product demand
block-Laplacian given by
$\ff^{(j)} \in (\complex^{r \times r})^{n} $ defined as
\begin{equation}
\label{eq:schurSquare:bipartite}
\blk{\ff}{i}^{(j)} = 
\begin{cases}
\blk{\AA}{i,j}\blk{\DD}{j,j}^{-\nfrac{1}{2}} & i \in F, \\
-\blk{\MM}{i,j}\blk{\DD}{j,j}^{-\nfrac{1}{2}} & i \in C.
\end{cases}
\end{equation}
Letting $\deg_j$ denote the number of nonzero blocks in
$\blk{\MM}{j},$ the number of nonzero blocks in each of
$\dd_F^{(j)}, \dd_C^{(j)}, \ff^{(j)}$ is at most $\deg_j.$ We
construct a sparse representation of each of them using $O(\deg_j)$
work. Thus, the total number of nonzero blocks in all these
block-vectors is at most $O(m)$ and we can explicitly construct
sparse representations for them using $O(m)$ work and $O(\log n)$
depth.

We can now express $\MM_{2}$ as 
\[\textstyle \MM_{2} = \MM_{3} + \frac{1}{2}\sum_{j \in F}\LL_{G(\dd_{F}^{(j)})}
+ \frac{1}{2}\sum_{j \in F}\LL_{G(\dd_{C}^{(j)})} + \frac{1}{2}\sum_{j \in
  F}\LL_{G(\ff^{(j)},F)}.\]

For all $i,j \in V,$ we compute
$\beta_{i,j} = \norm{\ff_{i}^{(j)}}.$ 
Define $\MM_{4}$ to be the block diagonal matrix such that
\begin{equation}
\label{eq:schurSquare:diagonal}
  \textstyle 
\blk{\MM_{4}}{i,i} = \frac{1}{2} \id_r \cdot
  \sum_{j \in F} 
  \beta_{i,j}
\left( \left(\sum_{k \in V}
\beta_{k,j}\right) - \beta_{i,j}\right).
\end{equation}
Since at most $m$ of $\beta_{i,j}$ are nonzero, and we can compute
$\beta_{i,j}$ and construct $\MM_{4}$ using using $O(m)$ work and
$O(\log n)$ depth. It is easy to verify that we can express $\MM_{3}$ explicitly as
\begin{equation}
\label{eq:schurSquare:remaining}
  \textstyle 
\MM_{3} = \frac{1}{2}
\left[
\begin{array}{cc}
\DD &
\blk{\MM}{F,C}\\
\blk{\MM}{C,F} 
&  2 \blk{\MM}{C,C}
\end{array}
\right] - \MM_{4}.
\end{equation}

For each $j,$ we use Lemmas~\ref{lem:product-clique-sparsify}
and~\ref{lem:bipartite-clique-sparsify} to construct, in
$O(\deg_j\eps^{-4})$ work and $O(\log n)$ depth, {\bdd} matrices
$\tilde{\LL}_{G(\dd_{F}^{(j)})}, \tilde{\LL}_{G(\dd_{C}^{(j)})},
\tilde{\LL}_{G(\ff^{(j)},F)},$
each with $O(\deg_j\eps^{-4})$ nonzero blocks such that
\[\tilde{\LL}_{G(\dd_{F}^{(j)})} \approx_{\epsilon} \LL_{G(\dd_{F}^{(j)})},
\quad \tilde{\LL}_{G(\dd_{C}^{(j)})} \approx_{\epsilon}
\LL_{G(\dd_{C}^{(j)})},\ \textrm{  and  }
\  \tilde{\LL}_{G(\ff^{(j)},F)} \approx_{\epsilon} \LL_{G(\ff^{(j)},F)}.\]

We can now construct the matrix 
\begin{equation}
\label{eq:schurSquare:output}
  \textstyle 
\MM_{1} = \MM_{3} + \frac{1}{2}\sum_{j \in F}\tilde{\LL}_{G(\dd_{F}^{(j)})}
+ \frac{1}{2}\sum_{j \in F}\tilde{\LL}_{G(\dd_{C}^{(j)})} + \frac{1}{2}\sum_{j \in
  F}\tilde{\LL}_{G(\ff^{(j)},F)} ,
\end{equation}
in $O(m\eps^{-4})$ time and $O(\log n)$ depth. $\MM_{1}$ has $O(m\eps^{-4})$
nonzero blocks and is our required matrix.

We first show that $\MM_{3}$ is {\bdd}.  Using
$\beta_{k,j} \le \norm{\blk{\DD}{j,j}^{-\nfrac{1}{2}} \blk{\MM}{i,k}}$
for all $j \in F, k \in V,$ we have for all $i \in F,$
\begin{equation}
\label{eq:schurSquare:smallremaining}
 2 \norm{\blk{\MM_{4}}{i,i}} \leq
  \sum_{j \in F} 
  \beta_{i,j}
\sum_{k \in V}
\beta_{k,j}
  \leq
    \sum_{j \in F} \norm{\blk{\AA}{i,j}} \norm{\blk{\DD}{j,j}^{-\nfrac{1}{2}}}^{2} 
    \sum_{k \in V} \norm{\blk{\MM}{j,k}} 
  \leq    \sum_{j \in F} \norm{\blk{\AA}{i,j}}
,
\end{equation}
where the last inequality uses that $\MM$ is {\bdd}. Again, using
$\MM$ is {\bdd}, for all $i \in F,$ we have 
\[\blk{\DD}{i,i} \pgeq \id_r \sum_{j \in F} \norm{\blk{\AA}{i,j}} +
\id_r \sum_{k \in C} \norm{\blk{\MM}{i,k}}.\]
Now combining with Eqs.~\eqref{eq:schurSquare:remaining}
and~\eqref{eq:schurSquare:smallremaining}, we obtain
$\blk{\MM_{3}}{i,i} \pgeq \id_r \sum_{k \in C} \norm{\blk{\MM}{i,k}}.$
Thus, $\MM_{3}$ is {\bdd}, and hence $\MM_{1}$ is {\bdd}.

We also have $\MM_{1} \approx_{\eps} \MM_2.$ Thus, using
Fact~\ref{fact:schurLoewner}, we obtain
$\schur{\MM_1}{F} \approx_{\eps} \schur{\MM_2}{F}.$ It remains to show
that $\blk{\MM_1}{F,F}$ is $\alpha^2/2$-\bdd. 

The key observation is that only the matrices
$\tilde{\LL}_{G(\dd_{C}^{(j)})}$ contribute to off-diagonal blocks in
$\blk{\MM_{1}}{F,F}.$ By construction,
\begin{align*}
\forall i,j \in F, \quad \blk{\LL_{G(\dd_{F}^{(j)})}}{i,i}
= \id_{r}\sum_{k \in F: k \neq i} \beta_{i,j}\beta_{k,j}.
\end{align*}
Thus, for all $i \in F$
\begin{align}
\label{eq:clique-diags}
\sum_{j \in
  F}\norm{\blk{\LL_{G(\dd_{F}^{(j)})}}{i,i}}
   \leq
     \sum_{j \in F} \sum_{k \in F}
    \beta_{i,j}\beta_{k,j}
& \leq
    \sum_{j \in F} \norm{\blk{\AA}{i,j}} \norm{\blk{\DD}{j,j}^{-\nfrac{1}{2}}}^{2} 
    \sum_{k \in F} \norm{\blk{\AA}{j,k}} \nonumber\\
& \leq
(1+\alpha)^{-1} \id_r \cdot    \sum_{j \in F} \norm{\blk{\AA}{i,j}},
\end{align}
where the last inequality follows since $\DD - \AA = \blk{\MM}{F,F}$
is {\abdd}.  Moreover $\DD - \AA$ being {\abdd} implies that for each
$i \in F,$
$\blk{\DD}{i,i} \pgeq (\alpha+1)\id_{r} \sum_{j \in F}
\norm{\blk{\AA}{i,j}}.$
Thus, using Eqs.~\eqref{eq:schurSquare:remaining}
and~\eqref{eq:schurSquare:smallremaining}, we obtain,
$2\blk{\MM_3}{i,i} \pgeq \alpha \id_{r} \sum_{j \in F}
\norm{\blk{\AA}{i,j}}.$

Since for each $j,$
$\tilde{\LL}_{G(\dd^{(j)})} \approx_{\eps}
\tilde{\LL}_{G(\dd^{(j)})},$
we have
$\blk{(\tilde{\LL}_{G(\dd^{(j)})})}{i,i} \pleq e^{\eps}
\blk{({\LL}_{G(\dd^{(j)})})}{i,i}.$
Thus, $\sum_{j} \tilde{\LL}_{G(\dd_{F}^{(j)})}$ is {\bdd} with sums of
norms of off-diagonal blocks at most
$e^{\eps} (1+\alpha)^{-1} \sum_{j \in F} \norm{\blk{\AA}{i,j}}.$ Since
these are the only matrices that contribute to off-diagonal blocks in
$\blk{\MM_{1}}{F,F},$ we get that $\blk{\MM_{1}}{F,F}$ is at least
$\frac{\alpha } {e^{\eps} (1+\alpha)^{-1} }$-\bdd. Using
$\eps < \nfrac{1}{2}$ gives us our claim.

It is easy to verify that our transformations maintain that if $\TT$
is a matrix that is only
nonzero inside the submatrix $\blk{\TT}{C,C}$, and $\MM + \TT$ is {\bdd},
then $\schurSquare(\MM+\TT,F,\epsilon) =
\schurSquare(\MM,F,\epsilon)+\blk{\TT}{C,C}$.
\end{proof}

We now prove the claims deferred from the above proof.
\begin{lemma}
\label{lem:schurSquare:equality}
The matrix $\MM_{2}$ defined by Eq.~\eqref{eq:schurSquare} satisfies
$\schur{\MM_{2}}{F} = \schur{\MM}{F}.$
\end{lemma}
\begin{proof}
  We need the following identity from~\cite{PengS14}:
\begin{equation}\label{eqn:schurSquare:PSidentity}
  (\DD - \AA)^{-1}
 =
 \nfrac{1}{2} \cdot
(
  \DD^{-1}
 + 
  \left(\II + \DD^{-1} \AA\right)
  \left(\DD - \AA \DD^{-1} \AA\right)^{-1}
  \left(\II + \AA \DD^{-1}\right)
 ).
\end{equation}
We have,
\begin{align*}
\schur{\MM_{2}}{F} & = \blk{\MM}{C,C} - \nfrac{1}{2}\cdot
                     \blk{\MM}{C,F} \DD^{-1} \blk{\MM}{F,C}  \\
& \qquad - \nfrac{1}{2}\cdot\blk{\MM}{C,F} (\II + \DD^{-1} \AA
  )(\DD - \AA \DD^{-1} \AA)^{-1}(\II + \AA \DD^{-1}) \blk{\MM}{F,C} \\
& = \blk{\MM}{C,C} \\
& \qquad - \nfrac{1}{2}\cdot\blk{\MM}{C,F} (\DD^{-1}  + (\II + \DD^{-1} \AA  )(\DD - \AA \DD^{-1} \AA)^{-1}(\II + \AA \DD^{-1})) \blk{\MM}{F,C} \\
& = \blk{\MM}{C,C} - \blk{\MM}{C,F} (\DD - \AA)^{-1} \blk{\MM}{F,C} \\
& = \schur{\MM}{F}.
\end{align*}
\end{proof}


\subsection{Schur Complement w.r.t. Highly {\abdd} Submatrices}
\label{sec:lastStep}

\begin{figure}[ht]

\begin{algbox}
$\MMtil_{SC}  = \lastStep\left(\MM, F, \alpha, \epsilon \right)$
\begin{enumerate}
\item Compute $\XX$, $\DD$, and $\AA$ as given by equations
  \eqref{eq:laststepX}, \eqref{eq:laststepD}, and \eqref{eq:laststepA}.
\item Construct $\YY$ as defined by equations
  \eqref{eq:lastStepYCross}, \eqref{eq:lastStepYF}, and \eqref{eq:lastStepYC}.
\item Construct sparse vectors 
  $\dd^{(j)}, \gg^{(j)}$ for all $j \in F$ as
  defined by equations~\eqref{eq:lastStepBipVec}
and~\eqref{eq:lastStepProdVec}.
\item For all $j \in F$,
 let $\tilde{\LL}_{G(\dd^{(j)},F)} \leftarrow \bipClique( \dd^{(j)}, F, \epsilon / 2)$
\item For all $j \in F$,
let $\tilde{\LL}_{G(\gg^{(j)} )} \leftarrow \prodClique( \gg^{(j)}, \epsilon / 2)$.
\item Compute $\RR$ as given by equation~\eqref{eq:lastStepR}.
\item Construct sparse vectors 
  $\rr^{(j)}$ for all $j \in F$ as
  defined by equations~\eqref{eq:lastStepRVec}.
\item For all $j \in F$,
let $\tilde{\LL}_{G(\rr^{(j)} )} \leftarrow \prodClique( \rr^{(j)}, \epsilon / 2)$.
\item
Compute $\SS$ as given by equation~\eqref{eq:lastStepS}.
\item
Return
\[
\MMtil_{SC} = \SS+\sum_{i}\tilde{\LL}_{G(\rr^{(i)} )}.
\]
\end{enumerate}
\end{algbox}

\caption{Pseudocode for procedure {\lastStep}: Computing an
  approximate Schur complement  w.r.t. a highly {\abdd} submatrices.}

\label{fig:lastStep}

\end{figure}

In this section, we describe the {\lastStep} procedure for computing an
approximate Schur complement of a {\bdd} matrix $\MM$ with a highly {\abdd}
submatrix $\blk{\MM}{F,F}$.
{\lastStep} is the final step of the {\schurApx} algorithm.
The key element of the procedure is a formula for approximating the
inverse of $\blk{\MM}{F,F}$ that is leveraged to approximate
the Schur complement of  $\MM$ as the Schur complement of matrix with
the $FF$ submatrix being block diagonal.

We prove guarantees for the {\lastStep} algorithm as stated in
Lemma~\ref{lem:lastStep}.

One could attempt to deal with the highly {\abdd} matrix at the last
step by directly replacing it with its diagonal, but this is
problematic.
Consider the case where $F$ contains $u$ and $v$
with a weight $\epsilon$ edge between them, and
$u$ and $v$ are connected to $u'$ and $v'$ in $C$
by weight $1$ edges respectively.
Keeping only the diagonal results in a Schur complement
that disconnects $u'$ and $v'$.
This however can be fixed by taking a step of random
walk within $F$.

Given a {\bdd} matrix $\MM$, s.t. $\blk{\MM}{F,F}$ is {\abdd} we
define a block diagonal matrix $\XX \in (\complex^{r \times r})^{|F|
  \times |F|} $ s.t. for each $i \in F$
\begin{equation}
\label{eq:laststepX}
\blk{\XX}{i,i} = \frac{\alpha}{\alpha+1} \blk{\MM}{i,i}
\end{equation}
and another  block diagonal matrix $\DD \in (\complex^{r \times r})^{|F|
  \times |F|} $ s.t. for each $i \in F$
\begin{equation}
\label{eq:laststepD}
\blk{\DD}{i,i} = \frac{1}{\alpha+1} \blk{\MM}{i,i},
\end{equation}
and we define a matrix $\AA \in (\complex^{r \times r})^{|F|
  \times |F|} $
\begin{align}
\label{eq:laststepA}
\blk{\AA}{i,j} =
\begin{cases}
0 & i = j
\\
-\blk{\MM}{i,j} & \text{ otherwise.}
\end{cases}
\end{align}

Thus $\blk{\MM}{F,F}  = \XX + \DD - \AA$.
One can check that because $\MM$ is {\bdd} and
$\blk{\MM}{F,F}$ is {\abdd}, it follows that $\DD - \AA$ is {\bdd}
and the matrix 
\[
\begin{bmatrix}
  \XX & \blk{\MM}{F,C} \\
   \blk{\MM}{C,F}  &  \blk{\MM}{C,C} 
\end{bmatrix}
\]
is also {\bdd}.

We will consider the linear operator
\begin{equation}
\ZZ^{(last)}
\defeq
\frac{1}{2} \XX^{-1}
	+ \frac{1}{2} \XX^{-1} \left(  \XX - \DD + \AA \right) \XX^{-1}
		\left( \XX - \DD + \AA \right) \XX^{-1}
\label{eq:zzFinal}.
\end{equation}
We define
\[
\MM^{(last)}
=
\begin{pmatrix}
\left(\ZZ^{(last)} \right)^{-1} & \blk{\MM}{F,C}\\
\blk{\MM}{F,C} & \blk{\MM}{C,C}
\end{pmatrix}
\]
\begin{lemma}\label{lem:goodMlast}
\[
\MM \pleq 
\MM^{(last)}
\pleq 
\left(1 +\frac{2}{\alpha}\right) \MM .
\]
\end{lemma}
We defer the proof of Lemma~\ref{lem:goodMlast} to Section~\ref{sec:lastStepDeferred}.

To utilize $\ZZ^{(last)}$, define
\begin{equation}
\MM^{\left(last\right)}_{2} \defeq
\left[
\begin{array}{cc}
 \frac{1}{2} \XX &  \frac{1}{2}\left( \XX - \DD + \AA\right) \XX^{-1} \blk{\MM}{F,C}\\
 \frac{1}{2} \blk{\MM}{C,F}   \XX ^{-1} \left( \XX - \DD + \AA\right) &
\blk{\MM}{C,C} - \frac{1}{2}\blk{\MM}{C,F}  \XX^{-1} \blk{\MM}{F,C}.
\end{array}
\right]
\label{eqn:lastSecondHalf}
\end{equation}
and note
\begin{equation}
  \label{eq:lastStepSchurEquals}
  \schur{\MM^{\left(last\right)}}{F} = 
   \schur{\MM^{\left(last\right)}_{2}}{F} 
\end{equation}
Lemma~\ref{lem:goodMlast} tells us that for large enough $\alpha$, we can
approximate the Schur complement of $\MM$ by approximating the the Schur
complement of $\MM^{(last)}_{2}$.

The next lemma tells us that $M_{2}$ is {\bdd} and that we can write
the matrix as a sum of an explicit {\bdd} matrix and sparse
implicitly represented product demand
block-Laplacians and bipartite product demand
block-Laplacians.

\begin{lemma}
\label{lem:lastCliqueSplit}
Consider a {\bdd} matrix $\MM$, where $\blk{\MM}{F,F}$ is
  {\abdd} for some $\alpha \ge 4$, and
let $\MM^{(last)}_{2}$ be the associated matrix defined by
equation~\eqref{eqn:lastSecondHalf}. 
Let $m$ be the number of nonzero blocks of $\MM$.

For $j \in F$, we define $\dd^{(i)} \in (\complex^{r
    \times r})^{n}$
\begin{equation}
\label{eq:lastStepBipVec}
\blk{\dd}{i}^{(j)} = 
\begin{cases}
\blk{ \AA  }{i,j} {\blk{\XX}{j,j}^{-\nfrac{1}{2}}} & \textrm{ for } i \in F \\
\blk{ \MM }{i,j} {\blk{\XX}{j,j}^{-\nfrac{1}{2}}} & \textrm{ for } i \in C
\end{cases}
\end{equation}
For $j \in F$, we define $\gg^{(j)} \in (\complex^{r
    \times r})^{n}$
\begin{equation}
\label{eq:lastStepProdVec}
\blk{\gg}{i}^{(j)} = 
\begin{cases}
0 & \textrm{ for } i \in F \\
\blk{ \MM }{i,j} {\blk{\XX}{j,j}^{-\nfrac{1}{2}}} & \textrm{ for } i
\in C
\end{cases}
\end{equation}
Then
\[
\MM^{(last)}_{2} = \YY + \frac{1}{2} \sum_{j \in F} \LL_{G(\gg^{(i)} )} +
\frac{1}{2} \sum_{j \in F} \LL_{G(\dd^{(i)},F)}
\]
where $\YY$ is {\bdd} and has $O(m)$ nonzero blocks, and the total
number of nonzero blocks in $\dd^{(i)}$ and $\gg^{(i)}$ for all $i$
combined is also $O(m)$.

$\YY$ as well as $\dd^{(i)}$ and $\gg^{(i)}$ for all $i$ can be
computed in $O(m)$ time and $O(\log n)$ depth.

If $\TT$ is a matrix that is only
nonzero inside the submatrix $\blk{\TT}{C,C}$,
then if we apply the transformation of Eq.~\ref{eqn:lastSecondHalf},
to $\MM+\TT$ instead of $\MM$, we find $(\MM+ \TT)^{(last)}_{2}  =
\MM^{(last)}_{2} + \blk{\TT}{C,C}$, and in
particular $\YY(\MM + \TT) = \YY(\MM) + \blk{\TT}{C,C}$.

\end{lemma}
We defer the proof of Lemma~\ref{lem:lastCliqueSplit} to Section~\ref{sec:lastStepDeferred}.
\begin{proof}
\emph{(of Lemma~\ref{lem:lastStep})}
The procedure $\textsc{LastStep}(M, F, \alpha, \epsilon)$ first computes 
$\YY$ and $\dd^{(i)}$ and $\gg^{(i)}$ for all $i$ s.t.
\begin{equation}
\label{eq:mlast2def}
\MM^{(last)}_{2} = \YY + \frac{1}{2} \sum_{j \in F} \LL_{G(\gg^{(i)} )} +
\frac{1}{2} \sum_{j \in F} \LL_{G(\dd^{(i)},F)},
\end{equation}
where $\MM^{(last)}_{2} \approx_{2/\alpha} \MM$.
Let $n_{\dd^{(i)}}$ and $n_{\gg^{(i)}}$ denote the number of nonzero
blocks in each $\dd^{(i)}$ and $\gg^{(i)}$.
By Lemma~\ref{lem:product-clique-sparsify},
we may call $\prodClique( \gg^{(i)}, \epsilon / 2)$
in $O(n_{\gg^{(i)}}\epsilon^{-4})$ time to sparsify
$\LL_{G(\gg^{(i)} )}$, producing a {\bdd} matrix
$\tilde{\LL}_{G(\gg^{(i)} )} \approx_{\epsilon/2} \LL_{G(\gg^{(i)} )} $ with
$O(n_{\gg^{(i)}}\epsilon^{-4})$ nonzero blocks.
By Lemma~\ref{lem:bipartite-clique-sparsify},
we may call $\bipClique( \dd^{(i)}, F, \epsilon / 2)$
in $O(n_{\dd^{(i)}}\epsilon^{-4})$ time to sparsify
$\LL_{G(\dd^{(i)},F )}$, producing a {\bdd} matrix
$\tilde{\LL}_{G(\dd^{(i)},F)} \approx_{\epsilon/2} \LL_{G(\dd^{(i)},F )}$ with
$O(n_{\dd^{(i)}}\epsilon^{-4})$ nonzero blocks.
The total running time of this is
$O(\epsilon^{-4}\sum_{i} (n_{\dd^{(i)}} + n_{\gg^{(i)}} )) =
O(\epsilon^{-4} m)$,
and the total number of nonzero blocks in 
\[
\frac{1}{2} \sum_{j \in F} \tilde{\LL}_{G(\gg^{(i)} )} +
\frac{1}{2} \sum_{j \in F} \tilde{\LL}_{G(\dd^{(i)},F)},
\]
is $O(\epsilon^{-4}(\sum_{i} n_{\dd^{(i)}} + n_{\gg^{(i)}} )) =
O(\epsilon^{-4} m)$. 
We define 
\begin{equation}
\label{eq:lastStepR}
\RR = \YY + \frac{1}{2} \sum_{j \in F} \tilde{\LL}_{G(\gg^{(i)} )} +
\frac{1}{2} \sum_{j \in F} \tilde{\LL}_{G(\dd^{(i)},F)}.
\end{equation}
which we can compute in $O(\epsilon^{-4} m)$ time and
$O(\log n)$ depth.
We have $\MM \approx_{2/\alpha} \MM^{(last)}_{2}$ and
$\MM^{(last)}_{2} \approx_{\epsilon/2} \RR$,
so that $\MM \approx_{2/\alpha+\epsilon/2}\RR$.
It follows from Fact~\ref{fact:schurLoewner} that $\schur{\MM}{F}
\approx_{2/\alpha+\epsilon/2} \schur{\RR}{F}$.

Because the sparsifiers computed by \bipClique preserve the graph bipartition, 
 $\blk{\RR}{F,F}$ is block diagonal.

We can use the block diagonal structure of 
$\blk{\RR}{F,F}$ quickly compute a
sparse approximation to $\schur{\RR}{F}$.

For $j \in F$, we define $\gg^{(j)} \in (\complex^{r
    \times r})^{|C|}$
\begin{equation}
\label{eq:lastStepRVec}
\rr^{(j)}_{i} = 
\blk{ \RR }{i,j}{\blk{\RR}{j,j}^{-\nfrac{1}{2}}} \textrm{ for } i
\in C
\end{equation}
Then
\[
\schur{\RR}{F}
=
\blk{\RR}{C,C} - \blk{\RR}{C,F} \blk{\RR}{F,F}^{-1} \blk{\RR}{F,C}
=
 \SS + \sum_{j \in F} \LL_{G(\rr^{(i)} )}.
\]
where
\begin{align}
\label{eq:lastStepS}
 \SS =
&
\blk{\RR}{C,C} - \diagop_{i \in C}\left( \id_{r} \sum_{j \in F}
      \sum_{k \in C\setminus{\setof{i}}} \norm{ \blk{\RR}{i,j} {\blk{\RR}{j,j}^{-\nfrac{1}{2}}}
      } \norm{ \blk{\RR}{k,j} {\blk{\RR}{j,j}^{-\nfrac{1}{2}}} }
      \right)
  \\
&-
\diagop_{i \in C}\sum_{j \in F}\left( \blk{\RR}{i,j} \blk{\RR}{j,j}^{-1}
      \blk{\RR}{j,i}\right).
\end{align}
Let us define, for each block row $i \in F$, 
$\rho_{i} = \sum_{j \in C} \norm{\blk{ \RR
  }{i,j} }$, and for each block row $i \in C$, 
$\rho_{i} = \sum_{j \in F} \norm{\blk{ \RR
  }{i,j} }$.
From $\RR$ being {\bdd}, we then conclude  
for each $i \in F$
\[
\blk{\RR}{i,i} \pgeq \id_{r} \rho_{i} 
\]
and for each $i \in C$
\[
\blk{\RR}{i,i} \pgeq \id_{r} \left( \rho_{i} + \sum_{j \in C}\norm{\blk{\RR}{i,j}} \right).
\]

With this in mind, we check that each for $i \in C$ of $\SS$
is {\bdd}.

The diagonal block satisfies
\begin{align*}
\blk{\SS}{i,i}
&\pgeq \id_{r} \left(
\sum_{j \in C}\norm{\blk{\RR}{i,j}}+
\rho_{i}
- \sum_{j \in F}
      \sum_{k \in C} \norm{ \blk{\RR}{i,j} \blk{\RR^{-1/2}}{j,j}
      } \norm{ \blk{\RR}{k,j} \blk{\RR^{-1/2}}{j,j} }
\right)
\\
&\pgeq \id_{r} \left(
\sum_{j \in C}\norm{\blk{\RR}{i,j}}+
\rho_{i}
- \sum_{j \in F}
      \sum_{k \in C} 
        \frac{1}{\rho_{j}} \norm{ \blk{\RR}{i,j} }
      \norm{ \blk{\RR}{k,j}} 
\right)
\\
&\pgeq \id_{r} \left(
\sum_{j \in C}\norm{\blk{\RR}{i,j}}+
\rho_{i}
- \sum_{j \in F}
        \frac{1}{\rho_{j}} \norm{ \blk{\RR}{i,j} }
\right)
\\
&\pgeq
\id_{r} \sum_{j \in C}\norm{\blk{\RR}{i,j}}
\end{align*}

We can compute $\SS$ and all $\rr^{(j)}_{i}$ in $O(m\epsilon^{-4})$
time and $O(\log n)$ depth, since this is an upper bound to the number of nonzero blocks
in $\RR$.

Let $n_{\rr^{(i)}}$ denote the number of nonzero
blocks in $\rr^{(i)}$. By Lemma~\ref{lem:product-clique-sparsify},
we may call $\prodClique( \rr^{(i)}, \epsilon / 2)$
in $O(n_{\rr^{(i)}}\epsilon^{-4})$ time to sparsify
$\LL_{G(\rr^{(i)} )}$, producing a {\bdd} matrix
$\tilde{\LL}_{G(\rr^{(i)} )} \approx_{\epsilon/2} \LL_{G(\rr^{(i)} )} $ with
$O(n_{\rr^{(i)}}\epsilon^{-4})$ nonzero blocks.

The total number of nonzero blocks in 
\[
\SS+\sum_{i}\tilde{\LL}_{G(\rr^{(i)} )}
\]
is also $O(\epsilon^{-8} m)$, and 
$
\SS+\sum_{i}\tilde{\LL}_{G(\rr^{(i)} )}
\approx_{\epsilon/2}
\schur{\RR}{F}.
$
So
\[
\MMtil_{SC} = \SS+\sum_{i}\tilde{\LL}_{G(\rr^{(i)} )}
\approx_{\epsilon+2/\alpha}
\schur{\MM}{F}
\]

The total amount of work to compute these $\tilde{\LL}_{G(\rr^{(i)}
  )}$ is  $O(\epsilon^{-4}(\sum_{i} n_{\rr^{(i)}})) =
O(\epsilon^{-8} m)$. The depth for this computation is
$O(\log n)$.

Suppose $\TT$ is a matrix that is only
nonzero inside the submatrix $\blk{\TT}{C,C}$, and $\MM + \TT$ is
{\bdd}.
We can show that $\lastStep(\MM+\TT, F, \alpha, \epsilon) = \lastStep(\MM, F,
\alpha, \epsilon)+\blk{\TT}{C,C}$, by first noting that this type of property
holds for $\YY$ by Lemma~\ref{lem:lastCliqueSplit}, and
from this concluding that 
similarly if we consider $\RR$ as a function of $\MM$ then 
$\blk{\RR}{C,C}(\MM + \TT) = \blk{\RR}{C,C}(\MM) + \blk{\TT}{C,C} $,
and finally considering $\MMtil_{SC}$ as a function of $\MM$, 
we can then easily show that $(\boldsymbol{\widetilde{M + T}})_{SC} =
\MMtil_{SC} + \blk{\TT}{C,C} $.
\end{proof}

\subsection{Deferred Proofs from Section~\ref{sec:lastStep}}
\label{sec:lastStepDeferred}

To help us prove Lemma~\ref{lem:goodMlast}, we first prove the next lemma.
\begin{lemma}
\label{lem:threeStep}
If $\blk{\MM}{F,F} = \XX +\DD - \AA $ be a {\abdd} matrix for some $\alpha \geq 4$, then the operator
$\ZZ^{(last)}$ as defined in Equation~\ref{eq:zzFinal} satisfies:
\[
\blk{\MM}{F,F} \pleq \left( \ZZ^{(last)} \right)^{-1}
\pleq \blk{\MM}{F,F} + \frac{2}{\alpha}  \left(\DD - \AA\right).
\]
\end{lemma}

\begin{proof}
Composing both sides by $\XX^{-1/2}$
and substituting in $\mathcal{L}= \XX^{-1/2} \left(\DD - \AA\right) \XX^{-1/2}$ means it
suffices to show
\[
\II + \mathcal{L}
\pleq \left( \frac{1}{2} \II
	+ \frac{1}{2} \left( \II - \mathcal{L}\right)^{2} \right)^{-1}
	\pleq \II + \mathcal{L}+  \frac{2}{\alpha} \mathcal{L}.
\]
We can use the fact that $\blk{\MM}{F,F}$ is $\alpha$-strongly diagonally dominant
to show $0 \pleq \LL \pleq \frac{2}{\alpha} \XX$, and equivalently
$0 \pleq \mathcal{L}\pleq \frac{2}{\alpha} \II$, as follows:

Firstly, $\DD - \AA \pgeq 0$ as $\DD - \AA$  is {\bdd},
similarly, $\DD + \AA \pgeq 0$ as  $\DD + \AA$ is also {\bdd}.
From the latter $\DD \pgeq - \AA$, so   $2 \DD \pgeq \DD - \AA$.
Finally, $\XX \pgeq \alpha \DD \pgeq \frac{\alpha}{2} \left( \DD - \AA \right)$.

As $\mathcal{L}$ and $\II$ commute, the spectral theorem
means it suffices to show this for any scalar $0 \leq t \leq \frac{2}{\alpha}$.
Note that
\[
\frac{1}{2}
	+ \frac{1}{2} \left( 1 - t \right)^{2}
	= 1 - t + \frac{1}{2} t^2
\]
Taking the difference between the inverse of this
and the `true' value of $1 + t$ gives:
\[
\left( 1 - t + \frac{1}{2} t^2 \right)^{-1} - \left( 1 + t \right)
= \frac{1 - \left( 1 + t \right) \left( 1 - t + \frac{1}{2} t^2\right)}{1 - t + \frac{1}{2}t^2 }
= \frac{\frac{1}{2 } t^2 \left( 1 - t \right) }
	{1 - t + \frac{1}{2} t^2}
\]
Incorporating the assumption that $0 \leq t \leq \frac{2}{\alpha}$
and $\alpha \geq 4$ gives
that the denominator is at least
\[
1 - \frac{2}{\alpha} \geq \frac{1}{2},
\]
and the numerator term can be bounded by
\[
0 \leq \frac{t^2}{2} \left( 1 - t\right) \leq \frac{t}{\alpha}.
\]
Combining these two bounds then gives the result.
\end{proof}

Lemma~\ref{lem:goodMlast} allows us to extend the approximation of $\blk{\MM}{F,F}$
  by the inverse of $\ZZ^{(last)}$ to the entire matrix $\MM$.

\begin{proof}
\emph{(of Lemma~\ref{lem:goodMlast})}
Recall that when a matrix $\TT$ is PSD,
\begin{equation}
\label{eq:blockSubstitute}
\begin{pmatrix}
 \TT  & 0\\
0 & 0
\end{pmatrix}
\pgeq 0.
\end{equation}

The left-hand inequality of our lemma follows immediately from
  Eq.~\ref{eq:blockSubstitute} and
  the left-hand side of the guarantee of Lemma~\ref{lem:threeStep}.
To prove the right-hand inequality we apply
  Eq.~\ref{eq:blockSubstitute} and the right-hand side of the guarantee of Lemma~\ref{lem:threeStep}.
  to conclude
\[
\begin{pmatrix}
\left(\ZZ^{(last)} \right)^{-1} & \blk{\MM}{F,C}\\
\blk{\MM}{F,C} & \blk{\MM}{C,C}
\end{pmatrix}
\pleq 
\begin{pmatrix}
\blk{\MM}{F,F} + \frac{2}{\alpha}  \LL   & \blk{\MM}{F,C}\\
\blk{\MM}{F,C} & \blk{\MM}{C,C}
\end{pmatrix}
=
\MM + 
\frac{2}{\alpha}  \begin{pmatrix}
 \LL   & 0\\
0 & 0
\end{pmatrix}.
\]
The matrix 
\[
\begin{bmatrix}
  \XX & \blk{\MM}{F,C} \\
   \blk{\MM}{C,F}  &  \blk{\MM}{C,C} 
\end{bmatrix}
\]
is {\bdd} and hence PSD. It follows that 
\[
\begin{pmatrix}
 \DD - \AA  & 0\\
0 & 0
\end{pmatrix}
\pleq \MM,
\]
by which we may conclude that
\[
\MM + 
\frac{2}{\alpha}  \begin{pmatrix}
 \DD - \AA  & 0\\
0 & 0
\end{pmatrix}
\pleq 
\MM + \frac{2}{\alpha}  \MM .
\]
\end{proof}
\begin{proof}
\emph{(of Lemma~\ref{lem:lastCliqueSplit})}

Each product demand clique $\LL_{G(\gg^{(i)} )}$ and bipartite product
demand clique $\LL_{G(\dd^{(i)},F)}$ is
{\bdd}.

We now have to find an expression for $\YY$ and show that $\YY$ is {\bdd}.
Let us write $\YY$ in terms of its blocks
From
\begin{align*}
\YY
& = 
\MM-
\left(
\frac{1}{2} \sum_{j \in F} \LL_{G(\gg^{(i)} )} +
\frac{1}{2} \sum_{j \in F} \LL_{G(\dd^{(i)},F)}
\right)
\end{align*}
it then follows by a simple check that
\begin{equation}
\label{eq:lastStepYCross}
\blk{\YY}{C,F}^{\dg} =
\blk{\YY}{F,C} =\frac{1}{2} (\XX-\DD) \XX^{-1} \blk{\MM}{F,C} 
=\frac{1}{2}  \left(1-\frac{1}{\alpha}\right) \blk{\MM}{F,C}
\end{equation}
and that
\begin{align}
\label{eq:lastStepYF}
\blk{\YY}{F,F} 
= & 
\frac{1}{2} \XX
- \frac{1}{2}\diagop_{i \in F}\left( \id_{r} \sum_{j \in F}
  \sum_{k \in C}
  \norm{\blk{ \AA }{i,j}{\blk{\XX}{j,j}^{-\nfrac{1}{2}}}  }
  \norm{ \blk{\MM}{k,j} \blk{\XX^{-\nfrac{1}{2}}}{j,j}} 
      \right)
\end{align}
and
\begin{align}
\label{eq:lastStepYC}
\blk{\YY}{C,C} 
= & 
\blk{\MM}{C,C}
- \frac{1}{2}\diagop_{i \in C}\left( \blk{\MM}{i,F} \XX^{-1}
      \blk{\MM}{F,i} \right)
\nonumber
\\
&
- \frac{1}{2}\diagop_{i \in C}\left( \id_{r} \sum_{j \in F}
      \sum_{k \in C\setminus{\setof{i}}} \norm{ \blk{\MM}{i,j} {\blk{\XX}{j,j}^{-\nfrac{1}{2}}}
      } \norm{ \blk{\MM}{k,j} {\blk{\XX}{j,j}^{-\nfrac{1}{2}}} }
      \right)
\nonumber
\\
&
- \frac{1}{2}\diagop_{i \in C}\left( \id_{r} \sum_{j \in F}
      \sum_{k \in F} \norm{ \blk{\MM}{i,j} {\blk{\XX}{j,j}^{-\nfrac{1}{2}}}
      } \norm{\blk{ \AA }{k,j}\blk{\XX^{-\nfrac{1}{2}}}{j,j}  }
      \right).
\end{align}

Next, we check that $\YY$ is {\bdd}.
First, let us define, for each block row $i \in F$, 
$\rho_{i} = \sum_{j \in F} \norm{\blk{ \AA
  }{i,j} }$.
From $\blk{\MM}{F,F}$ being {\abdd}, we then conclude  
$ \blk{\MM}{i,i} \pgeq \id_{r} (1+\alpha) \rho_{i}$.
And from $\MM$ being {\bdd}, we find that
$ \sum_{j \in C} \norm{\blk{ \MM
  }{i,j} } \leq \alpha \rho_{i}$.
We now check that each block row $i \in F$ of $\YY$ is {\bdd}.
The off-diagonal block norm sum is at most $\frac{1}{2}
\left(1-\frac{1}{\alpha}\right) \alpha \rho_{i}$.

The diagonal satisfies
\begin{align*}
\blk{\YY}{i,i}
& \pgeq \id_{r} \left( \frac{1}{2}\alpha \rho_{i}
-\frac{1}{2}\left( \sum_{j \in F}
  \sum_{k \in C}
  \norm{\blk{ \AA }{i,j}{\blk{\XX}{j,j}^{-\nfrac{1}{2}}}  }
  \norm{ \blk{\MM}{k,j} {\blk{\XX}{j,j}^{-\nfrac{1}{2}}}} 
      \right)
\right)
\\
& \pgeq \id_{r} \left(
\frac{1}{2} \alpha \rho_{i}
-\frac{1}{2}\left( \sum_{j \in F}
  \sum_{k \in C}
  \frac{1}{\alpha \rho_{j}}\norm{\blk{ \AA }{i,j} }
  \norm{ \blk{\MM}{k,j}} 
\right)
\right)
\\
& \pgeq \id_{r} \left(
\frac{1}{2} \alpha \rho_{i}
-\frac{1}{2}\left( \sum_{j \in F}
\norm{\blk{ \AA }{i,j} }
\right)
\right)
=
\id_{r} \left(
\frac{1}{2}
\left(1-\frac{1}{\alpha}\right) \alpha \rho_{i}.
\right)
\end{align*}
So the block rows with $i \in F$ are {\bdd}.

Next we check the block rows for $i \in C$.
Let us define for each $i \in C$,
$
\rho_{i} = \sum_{j \in F} \norm{\blk{ \MM
  }{i,j} }.
$
Thus for $\YY$, the sum of block norms of row $i$ over columns $j \in F$ 
\[
\sum_{j \in F} \norm{\blk{ \YY
  }{i,j} } = \frac{1}{2}  \left(1-\frac{1}{\alpha}\right)  \sum_{j \in F} \norm{\blk{ \MM
  }{i,j} } = \frac{1}{2}  \left(1-\frac{1}{\alpha}\right) \rho_{i}.
\]
For $\YY$, the sum of block norms of row $i$ over columns $j \in C$
is
$\sum_{j \in C\setminus\setof{i}} \norm{\blk{ \MM
  }{i,j} }$.
So the total sum of the blocks norms of off-diagonals is
\[
\frac{1}{2} \rho_{i} \left(1-\frac{1}{\alpha}\right)
+
\sum_{j \in C\setminus\setof{i}} \norm{\blk{ \MM }{i,j} }
\]
The diagonal block satisfies
\begin{align*}
\blk{\YY}{i,i}
&\pgeq \id_{r} \left(
\rho_{i}
- \frac{1}{2}\norm{\blk{\MM}{i,F} \XX^{-1}
      \blk{\MM}{F,i}}
\right.
\\
&
- \frac{1}{2}\left( \sum_{j \in F}
      \sum_{k \in C\setminus{\setof{i}}} \norm{ \blk{\MM}{i,j} {\blk{\XX}{j,j}^{-\nfrac{1}{2}}}
      } \norm{ \blk{\MM}{k,j} {\blk{\XX}{j,j}^{-\nfrac{1}{2}}} }
      \right)
\\
&
\left.
- \frac{1}{2}\left( \sum_{j \in F}
      \sum_{k \in F} \norm{ \blk{\MM}{i,j} {\blk{\XX}{j,j}^{-\nfrac{1}{2}}}
      } \norm{\blk{ \AA }{k,j}{\blk{\XX}{j,j}^{-\nfrac{1}{2}}}  }
      \right)
\right)
\\
&
\pgeq \id_{r} \left(
\rho_{i}
+
\sum_{j \in C\setminus\setof{i}} \norm{\blk{ \MM }{i,j} }
\right.
\\
&
- \frac{1}{2}\left( \sum_{j \in F}
      \sum_{k \in C} \frac{1}{\alpha\rho_{j}}\norm{ \blk{\MM}{i,j} }\norm{ \blk{\MM}{k,j} } 
      \right)
\\
&
\left.
- \frac{1}{2}\left( \sum_{j \in F}
  \sum_{k \in F} \frac{1}{\alpha\rho_{j}} \norm{ \blk{\MM}{i,j} }
  \norm{\blk{\AA}{k,j}  }
      \right)
\right)
\\
&
\pgeq \id_{r} \left(
\rho_{i}
+
\sum_{j \in C\setminus\setof{i}} \norm{\blk{ \MM }{i,j} }
- \frac{1}{2} \rho_{i}
- \frac{1}{2}\frac{1}{\alpha}  \rho_{i}
\right)
\\
&= \id_{r} \left(
\frac{1}{2} \rho_{i} \left(1-\frac{1}{\alpha}\right)
+
\sum_{j \in C\setminus\setof{i}} \norm{\blk{ \MM }{i,j} }
\right) 
.
\end{align*}
For these block rows are also {\bdd}, and hence $\YY$ is {\bdd}.
It is clear from the definitions that 
$\YY$ as well as $\dd^{(i)}$ and $\gg^{(i)}$ for all $i$ can be
computed in $O(m)$ time and $O(\log n)$ depth.

It is easy to verify that if $\TT$ is a matrix that is only
nonzero inside the submatrix $\blk{\TT}{C,C}$,
then if we apply the transformation of Eq.~\ref{eqn:lastSecondHalf},
to $\MM+\TT$ instead of $\MM$, we find $(\MM+ \TT)^{(last)}_{2}  =
\MM^{(last)}_{2} + \blk{\TT}{C,C}$, and in
particular $\YY(\MM + \TT) = \YY(\MM) + \blk{\TT}{C,C}$.
\end{proof}

\subsection{Sparsifying Product Demand Block-Laplacians}
\label{sec:product-clique-sparsify}
In this section, we show how to efficiently sparsify product demand
block-Laplacians and their bipartite analogs. We prove the following
key lemma later in this section that allows us to transfer results on
graph sparsification to sparsifying these product block-Laplacians.
\begin{lemma}
\label{lem:product-block-graphs-sparsification}
  Suppose we have two graphs on $n$ vertices, $H^{(1)}$ and $H^{(2)}$
  such that $H^{(1)} \approx_{\eps} H^{(2)}.$ Given
  $\dd \in (\complex^{r \times r})^{n},$ define the {\bdd} matrix
  $\LL^{(\ell)} \in (\complex^{r \times r})^{n \times n}$ for each
  $\ell=1,2,$ as
\[
\blk{\LL^{(\ell)}}{i,j} =
\begin{cases}
-\frac{h^{(\ell)}_{i,j}}{\norm{\blk{\dd}{i}} \norm{\blk{\dd}{j}}}\blk{\dd}{i} \blk{\dd}{j}^{\dg} & \quad i\neq j,\\
\id_r \cdot \sum_{k: k \neq i} h^{(\ell)}_{i,k} & \quad \textrm{otherwise,}
\end{cases}
\]
where $h^{(\ell)}_{i,j}$ denotes the weight of the edge $i,j$ in
$H^{(\ell)}.$ Then, $\LL^{(1)} \approx_{\eps} \LL^{(2)}.$
\end{lemma}

We now introduce scalar versions of product block-Laplacian matrices
that will be useful.
\begin{definition}
\label{def:demand}
The product demand graph of a vector $\dd \in (\Re_{>0})^{n}$,
$G(\dd)$, is a complete weighted graph on $n$ vertices
whose weight between vertices $i$ and $j$ is given by
$\ww_{ij} = \dd_{i} \dd_{j}.$

The Laplacian of $G(\dd),$ denoted $\LL_{G(\dd)}$ is called the
product demand Laplacian of $\dd.$
\end{definition}
\begin{definition}
\label{def:demand2}
The bipartite product demand graph of two vectors
$\dd^{A} \in (\Re_{>0})^{|A|},$ $\dd^{B} \in (\Re_{>0})^{|B|},$
$G(\dd^{A}, \dd^{B})$, is a weighted bipartite graph on vertices
$A \cup B,$ whose weight between vertices $i \in A$ and $j \in B$ is
given by
$\ww_{ij} = \dd^{A}_{i} \dd^{B}_{j}.$

The Laplacian of $G(\dd^{A}, \dd^{B}),$ denoted
$\LL_{G(\dd^{A}, \dd^{B})}$ is called the bipartite product demand
Laplacian of $(\dd^{A},\dd^{B}).$
\end{definition}

In Section~\ref{sec:weightedExp}, we give results on efficiently
constructing approximations to product demand Laplacians
that can be summarized as follows:
\begin{restatable}[]{lemma}{weightedExp}
\label{lem:weightedExpander}
There is a routine $\textsc{WeightedExpander}(\dd, \epsilon)$ such
that for any $\epsilon > 0,$ and a demand vector
$\dd \in (\Re_{>0})^{n},$
$\textsc{WeightedExpander}(\dd, \epsilon)$ returns in
$O(n \epsilon^{-4})$ work and $O(\log{n})$ depth a graph $H$ with
$O(n \epsilon^{-4})$ edges such that
\[
\LL_{H} \approx_{\epsilon} \LL_{G(\dd)}.
\]
\end{restatable}

\begin{restatable}[]{lemma}{weightedBipartiteExp}
\label{lem:weightedBipartiteExpander}
There is a routine
$\textsc{WeightedBipartiteExpander}(\dd^{A}, \dd^{B} , \epsilon)$ such
that for any demand vectors $\dd^{A}$ and $\dd^{B}$ of total length
$n,$ and a parameter $\epsilon$, it returns in
$O(n \epsilon^{-4})$ work and $O(\log{n})$ depth a bipartite graph $H$
between $A$ and $B$ with $O(n \epsilon^{-4})$ edges such that
\[
\LL_{H} \approx_{\epsilon} \LL_{G(\dd^{A}, \dd^{B})}.
\]
\end{restatable}
Finally, we need to define an operation on graphs:
Given a graph $G,$ define $G^{(K2)}$ to be the graph obtained by
duplicating each vertex in $G,$ and for each edge $(i,j)$ in $G,$ add a $2
\times 2$ bipartite clique between the two copies of $i$ and $j.$

\begin{figure}[ht]

\begin{algbox}
${\LL} = \prodClique\left(\dd, \epsilon \right)$
\begin{enumerate}
\item Initialize $\LL \in (\complex^{r \times r})^{n \times n}$ to 0.
\item Construct $\ww$ where $\ww_i = \norm{\blk{\dd}{i}}.$
\item $H \leftarrow \textsc{WeightedExpander}(\ww,\epsilon).$
\item For each edge $(i,j) \in H,$ with weight $h_{i,j},$
\begin{enumerate}
\item $\blk{\LL}{i,i} \leftarrow \blk{\LL}{i,i} + h_{i,j}\id_r.$ 
\item $\blk{\LL}{j,j} \leftarrow \blk{\LL}{j,j} + h_{i,j}\id_r.$
\item $\blk{\LL}{i,j} \leftarrow \blk{\LL}{i,j} -
  \frac{h_{i,j}}{\ww_{i}\ww_{j}} \blk{\dd}{i}\blk{\dd}{j}^{\dg}.$ 
\item $\blk{\LL}{j,i} \leftarrow \blk{\LL}{j,i} -
  \frac{h_{i,j}}{\ww_{i}\ww_{j}} \blk{\dd}{j}\blk{\dd}{i}^{\dg}.$ 
\end{enumerate}
\item Return $\LL$.
\end{enumerate}
\end{algbox}

\caption{Pseudocode for Sparsifying Product Demand Block-Laplacians}

\label{fig:prodClique}

\end{figure}

\begin{figure}[ht]

\begin{algbox}
${\LL} = \bipClique\left(\bb, F, \epsilon \right)$
\begin{enumerate}
\item Initialize $\LL \in (\complex^{r \times r})^{n \times n}$ to 0.
\item Construct $\ww$ where $\ww_i = \norm{\blk{\dd}{i}}.$ Let $C = V
  \setminus F.$
\item $H \leftarrow \textsc{WeightedBipartiteExpander}(\ww |_{F}, \ww|_{C},\epsilon).$
\item For each edge $(i,j) \in H,$ with weight $h_{i,j},$
\begin{enumerate}
\item $\blk{\LL}{i,i} \leftarrow \blk{\LL}{i,i} + h_{i,j}\id_r.$ 
\item $\blk{\LL}{j,j} \leftarrow \blk{\LL}{j,j} + h_{i,j}\id_r.$
\item $\blk{\LL}{i,j} \leftarrow \blk{\LL}{i,j} -
  \frac{h_{i,j}}{\ww_{i}\ww_{j}}\blk{\dd}{i}\blk{\dd}{j}^{\dg}.$ 
\item $\blk{\LL}{j,i} \leftarrow \blk{\LL}{j,i} -
  \frac{h_{i,j}}{\ww_{i}\ww_{j}} \blk{\dd}{j}\blk{\dd}{i}^{\dg}.$ 
\end{enumerate}
\item Return $\LL$.
\end{enumerate}
\end{algbox}

\caption{Pseudocode for Sparsifying Bipartite Product Demand Block-Laplacians}

\label{fig:prodBipClique}

\end{figure}

We now combine the above construction of sparsifiers fo product demand
graphs with Lemma~\ref{lem:product-block-graphs-sparsification} to
efficiently construct sparse approximations to product demand
block-Laplacians.
\begin{proof}\emph{(of Lemma~\ref{lem:product-clique-sparsify})}
The procedure $\prodClique(\dd,\epsilon)$ returns the matrix $\LL$
where 
\[
\blk{\LL}{i,j} =
\begin{cases}
-\frac{h_{i,j}}{\norm{\blk{\dd}{i}} \norm{\blk{\dd}{j}}}\blk{\dd}{i} \blk{\dd}{j}^{\dg} & \quad i\neq j,\\
\id_r \cdot \sum_{k: k \neq i} h_{i,k} & \quad \textrm{otherwise.}
\end{cases}
\]
By construction $\LL_{H} \approx_{\eps} \LL_{G(\ww)}.$ Thus, applying
Lemma~\ref{lem:product-block-graphs-sparsification}, with $H^{(1)} =
H$ and $H^{(2)} = G(\ww),$ we know that $\LL
\approx_{\eps} \LL^{(2)},$ where $\LL^{(2)}$ is given by 
\[
\blk{\LL^{(2)}}{i,j} =
\begin{cases}
-\frac{\ww_i \ww_j}{\norm{\blk{\dd}{i}} \norm{\blk{\dd}{j}}}\blk{\dd}{i} \blk{\dd}{j}^{\dg} & \quad i\neq j,\\
\id_r \cdot \ww_i \sum_{k: k \neq i} \ww_k & \quad \textrm{otherwise.}
\end{cases}
\]
Since $\ww_i = \norm{\blk{\dd}{i}},$ we have $\LL^{(2)} = \LL_{G(\dd)},$
proving our claim. 
\end{proof}

\begin{proof}\emph{(of Lemma~\ref{lem:bipartite-clique-sparsify})}
The procedure $\bipClique(\dd,F,\epsilon)$ returns the matrix $\LL$
where 
\[
\blk{\LL}{i,j} =
\begin{cases}
-\frac{h_{i,j}}{\norm{\blk{\dd}{i}}\norm{\blk{\dd}{j}}}\blk{\dd}{i} \blk{\dd}{j}^{\dg} & \quad i\neq j,\\
\id_r \cdot \sum_{k: k \neq i} h_{i,k} & \quad \textrm{otherwise.}
\end{cases}
\]
Since $H$ returned by $\textsc{WeightedBipartiteExpander}$ is
guaranteed to be bipartite using
Lemma~\ref{lem:weightedBipartiteExpander}, we obtain that
$\blk{(\LL_{H})}{F,F},$ and $\blk{(\LL_{H})}{C,C}$ are block-diagonal.

By construction $\LL_{H} \approx_{\eps} \LL_{G(\ww|_{F}, \ww|_{C})}.$ Thus, applying
Lemma~\ref{lem:product-block-graphs-sparsification}, with $H^{(1)} =
H$ and $H^{(2)} = G(\ww|_{F}, \ww|_{C}),$ we know that $\LL
\approx_{\eps} \LL^{(2)},$ where $\LL^{(2)}$ is given by 
\[
\blk{\LL^{(2)}}{i,j} =
\begin{cases}
\id_r \cdot \ww_i \sum_{k \in C} \ww_{k} & \quad i = j,
\text{ and } i \in F,\\
\id_r \cdot \ww_{i} \sum_{k \in F} \ww_{k} & \quad i = j,
\text{ and } i \in C,\\
0 & \quad i\neq j, \text{ and } ( i , j \in F  \text{ or }   i, j \in C),\\
-
\frac{\ww_{i}\ww_{j}}{\norm{\blk{\dd}{i}}\norm{\blk{\dd}{j}}}\blk{\dd}{i}
\blk{\dd}{j}^{\dg} & \quad i\neq j, 
\text{otherwise.}
\end{cases}
\]
Since $\ww_i = \norm{\blk{\dd}{i}},$ we have $\LL^{(2)} = \LL_{G(\dd|_F,\dd|_C)},$
proving our claim. 
\end{proof}

Finally, we give a proof of
Lemma~\ref{lem:product-block-graphs-sparsification}.
\begin{proof}\emph{(of Lemma~\ref{lem:product-block-graphs-sparsification})}
Using Fact~\ref{fact:sum-of-unitary}, we can write each $\blk{\dd}{i}$ as
$\frac{1}{2}\norm{\blk{\dd}{i}}(\QQ_{i}^{(1)} + \QQ_{i}^{(2)}),$ where
$\QQ_{i}^{(1)} (\QQ_{i}^{(1)})^{\dg} = \QQ_{i}^{(2)}
(\QQ_{i}^{(2)})^{\dg} = \id_{r}.$ Construct the matrix $\FF \in
(\complex^{r \times r})^{2n \times 2n}$ as follows:
\begin{align*}
  \FF & =
      \left[
\begin{array}{ccccc}
\QQ_{1}^{(1)} & \QQ_{1}^{(2)} &  0 & 0 & \ldots \\
0 & 0 & \QQ_{2}^{(1)} & \QQ_{2}^{(2)} & \ldots \\
\vdots & \vdots & & & \ddots 
\end{array}
\right].
\end{align*}
Now, observe that for each $\ell = 1,2,$
\begin{align*}
\blk{\LL^{(\ell)}}{i,i} 
=
 \frac{1}{2}\left(\sum_{k: k
  \neq i} h^{(\ell)}_{i,k}\right) \cdot \left( \QQ_{i}^{(1)} (\QQ_{i}^{(1)})^{\dg} +
  \QQ_{i}^{(2)} (\QQ_{i}^{(2)})^{\dg} \right) 
= 
 \frac{1}{4} \blk{\left(\FF(\LL_{\twoCover{(H^{(\ell)})}} \otimes \id_{r}) \FF^{\dg} \right)}{i,i},
\end{align*}
where $\LL_{\twoCover{(H^{(\ell)})}}$ is the Laplacian of the graph $\twoCover{(H^{(\ell)})}$
defined above, and $\otimes$ tensor product. Also, for $i \neq j,$
\begin{align*}
\blk{\LL^{(\ell)}}{i,j} 
= -\frac{h^{(\ell)}_{i,j}}{\ww_{i}\ww_{j}}\blk{\dd}{i} \blk{\dd}{j}^{\dg} 
=
 -\frac{h^{(\ell)}_{i,j}}{4} \cdot 
\left( \QQ_{i}^{(1)} +
    \QQ_{i}^{(2)}\right)\left( \QQ_{j}^{(1)} +    \QQ_{j}^{(2)}\right)^{\dg}
= 
 \frac{1}{4} \blk{\left(\FF(\LL_{\twoCover{(H^{(\ell)})}} \otimes \id_{r}) \FF^{\dg} \right)}{i,j},
\end{align*}
Thus, 
\[\LL^{(\ell)} = \frac{1}{4}\FF(\LL_{\twoCover{(H^{(\ell)})}} \otimes \id_{r}) \FF^{\dg}.\]

By assumption, we have $\LL_{H^{(1)}} \approx_{\epsilon} \LL_{H^{(2)}}.$
Using Lemma~\ref{lem:lift-sparsifier}, we get that
$\LL_{\twoCover{H^{(1)}}} \approx_{\epsilon} L_{\twoCover{H^{(2)}}}.$
This implies
\[\LL_{\twoCover{(H^{(1)})}} \otimes \id_{r} \approx_{\epsilon}
\LL_{\twoCover{(H^{(2)})}} \otimes \id_{r},\]
and thus, $\LL^{(1)} \approx_{\eps} \LL^{(2)}.$
\end{proof}
\subsection{Constructing sparsifiers for lifts of graphs}
Given a graph $G,$ define $\twoCover{G}$ to be the graph obtained by
duplicating each vertex in $G,$ and for each edge $(i,j)$ in $G,$ add a $2
\times 2$ bipartite clique between the two copies of $i$ and $j.$
\begin{lemma}
\label{lem:lift-sparsifier}
If $H$ is a sparsifier for $G,$ i.e., $L_{G} \approx_{\eps} L_{H},$
then $\twoCover{H}$ is a sparsifier for $\twoCover{G},$ i.e. $L_{\twoCover{G}} \approx_{\eps} L_{\twoCover{H}}.$
\end{lemma}
\begin{proof}
Since $L_{G} \approx_{\eps} L_{H},$ we have $e_i^{\top} L_{G} e_i
\approx_{\eps} e_i^{\top} L_{H} e_i.$ Thus, if $D_{G}$ denotes the
diagonal matrix of degrees of $G,$ we have $D_{G} \approx_{\eps} D_{H}.$
The Laplacian for $\twoCover{G}$ is 
\[L_{\twoCover{G}} = L_{G} \otimes 
\left[
\begin{array}{cc}
1 & 1 \\
1 & 1 
\end{array}
\right] 
+ D_{G} \otimes 
\left[
\begin{array}{cc}
1 & -1 \\
-1 & 1 
\end{array}
\right].
\]
Similarly, 
\[L_{\twoCover{H}} = L_{H} \otimes 
\left[
\begin{array}{cc}
1 & 1 \\
1 & 1 
\end{array}
\right] 
+ D_{H} \otimes 
\left[
\begin{array}{cc}
1 & -1 \\
-1 & 1 
\end{array}
\right].
\]
Observe that we can write 
\[
L_{G} \otimes 
\left[
\begin{array}{cc}
1 & 1 \\
1 & 1 
\end{array}
\right] = 
\left[
\begin{array}{ccccc}
1 & 1 & 0 & 0 & \ldots \\
0 & 0 & 1 & 1 & \ldots \\
\vdots & \vdots & & \ddots 
\end{array}
\right]
L_{G}
\left[
\begin{array}{ccccc}
1 & 1 & 0 & 0 & \ldots \\
0 & 0 & 1 & 1 & \ldots \\
\vdots & \vdots & & \ddots 
\end{array}
\right]^{\top}.
\]
Since $L_{G} \approx_{\eps} L_{H},$ this implies 
\[L_{G} \otimes 
\left[
\begin{array}{cc}
1 & 1 \\
1 & 1 
\end{array}
\right] \approx_{\eps}
L_{H} \otimes 
\left[
\begin{array}{cc}
1 & 1 \\
1 & 1 
\end{array}
\right].
\]
Similarly, we get,
\[D_{G} \otimes 
\left[
\begin{array}{cc}
1 & -1 \\
-1 & 1 
\end{array}
\right] \approx_{\eps}
D_{H} \otimes 
\left[
\begin{array}{cc}
1 & -1 \\
-1 & 1 
\end{array}
\right].
\]
Adding the above two, we get  $L_{\twoCover{G}} \approx_{\eps} L_{\twoCover{H}}.$
\end{proof}



\section{Estimating Leverage Scores by Undersampling}\label{sec:undersampling}

We will control the densities of all intermediate {\bdd} matrices using
the uniform sampling technique introduced by~\cite{cohen2014uniform}.
It relies on sampling columns
\footnote{The randomized numerical linear algebra literature, e.g. ~\cite{cohen2014uniform},
 typically samples rows instead of columns of matrices.
 We sample columns instead in order to use a more natural set of notations.}
 of matrices by upper bounds of their true leverage scores,
\[
\ttau_{i}\left( \AA \right)  = \aa_i^T \left( \AA \AA^{\dg} \right)^{-1} \aa_i.
\]
These upper bounds are measured w.r.t. a different matrix, giving
generalized leverage scores of the form:
\[
\tau^{\BB}_i(\AA) =
\begin{cases}
\aa_i^{\dg} (\BB \BB^{\dg})^{-1} \aa_i & \text{if } \aa_i\perp\ker(\BB),\\
1 & \text{otherwise}.
\end{cases}
\]
We introduced \emph{unitary edge-vertex transfer matrices} in Definition~\ref{def:unitaryB}.
Sparsifying {\bdd} matrices can be transformed into the more general
setting described above via a unitary edge-vertex transfer matrix,
which is analogous to the edge-vertex incidence matrix. 

Lemma~\ref{lem:bsdd-factorization} proves that every {\bdd} matrix $\MM \in (\complex^{r \times r})^{n \times n}$
  with $m$ nonzero off-diagonal blocks can be written as $\XX+\BB \BB^{\dg}$ where
  $\BB \in (\complex^{r \times r})^{n \times 2m}$ is a unitary
  edge-vertex transfer matrix
  and $\XX$ is a block diagonal PSD matrix. Additionally, for every block diagonal matrix $\YY$ s.t. $\MM-\YY$ is
  {\bdd}, we have $\XX \pgeq \YY$.
  This decomposition
can be found in $O(m)$ time and $O(\log n)$ depth.
We rely on this $\XX$ to detect some cases of high leverage scores
as samples of $\BB$ may have lower rank.

We will reduce the number of nonzero blocks in $\MM$ by
sampling columns blocks from this matrix.
This is more restrictive than sampling individual columns.
Nonetheless, it can be checked via matrix concentration bounds \cite{ahlswede2002strong, tropp2012user}
that it suffices to sample the block by analogs of leverage scores: 

\begin{equation}
\label{eqn:tau}
\ttau_{[i]} = \mathrm{tr}\left(\BB_{[i]}^{\dg} (\XX + \BB
  \BB^{*})^{-1} \BB_{[i]}\right)
\end{equation}

As in~\cite{cohen2014uniform}, we  recursively estimate upper bounds
for these scores, leading to the pseudocode given in Figure~\ref{fig:uniformSample}.

\begin{figure}[ht]
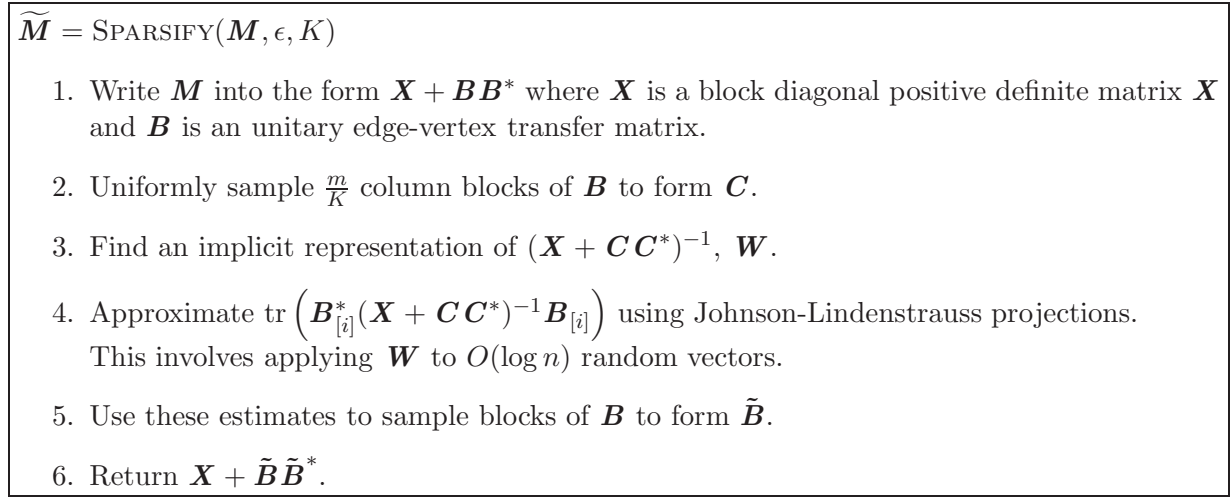


\begin{algbox} $\MMtil = \textsc{Sparsify}(\MM, \epsilon, K) $
\begin{enumerate}
    \item Write $\MM$ into the form $\XX + \BB \BB^{\dg}$ where $\XX$ is a block diagonal
    positive definite matrix $\XX$ and $\BB$ is an unitary edge-vertex transfer matrix.
	\item	Uniformly sample $\frac{m}{K}$ column blocks of $\BB$ to form $\CC$.
	\item	Find an implicit representation of $(\XX + \CC \CC^{\dg})^{-1}$, $\WW$.
	\item Approximate $\mathrm{tr}\left(\BB_{[i]}^{\dg} (\XX + \CC \CC^{\dg})^{-1} \BB_{[i]} \right)$ using Johnson-Lindenstrauss projections. \\This involves applying $\WW$ to $O (\log n)$ random vectors.
	\item	Use these estimates to sample blocks of $\BB$ to form $\BBtil$.
	\item	Return $\XX + \BBtil \BBtil^{\dg}$.
\end{enumerate}

\end{algbox}

\caption{Sparsification via. Uniform Sampling}
\label{fig:uniformSample}
\end{figure}

\begin{lemma}
\label{lem:sparsify}
Given a positive definite {\bdd} matrix $\MM \in (\complex^{r \times r})^{n \times n}$ with $m$ nonzero blocks.

Assume that for any positive definite {\bdd} matrix $\MM' \in (\complex^{r \times
  r})^{n \times n}$ with $\widehat{m}$ nonzero blocks, 
we can find an implicit representation
of a matrix $\WW$ such that $\WW\approx_{1} \left(\MM'\right)^{-1}$
in depth $d_{con}(\widehat{m},n)$ and work $w_{con}(\widehat{m},n)$ and for any vector
$b$, we can evaluate $\WW \bb$ in depth $d_{eval}(\widehat{m},n)$ and work $w_{eval}(\widehat{m},n)$. 

For any $K \geq1$, $1\geq\epsilon>0$,
the algorithm $\textsc{Sparsify}(\MM, \epsilon, K)$ outputs
an explicit positive definite {\bdd} matrix $\MMtil$ with 
$O(K n\log n/\epsilon^{2})$ nonzero blocks and $\MMtil \approx_{\epsilon} \MM$.

Also, this algorithm
runs in $$d_{con}\left(\frac{m}{K},n\right)+O\left(d_{eval} \left(\frac{m}{K},n\right)+\log n \right)$$
depth and $$w_{con}\left(\frac{m}{K},n\right)+O\left(w_{eval} \left(\frac{m}{K},n\right) \log n +m\log n\right)$$
work.
\end{lemma}

The guarantees of this process can be obtained from the main
result from~\cite{cohen2014uniform}:

\begin{theorem}[Theorem 3 from~\cite{cohen2014uniform}]\label{thm:leverage_score_undersampling}
Let $\AA$ be a $n$ by $m$ matrix, and $\alpha\in(0,1]$ a density parameter.
Consider the matrix $\AA'$ consisting of a random $\alpha$ fraction of the columns of $\AA$.
Then, with high probability we have
$\uu_i = \tau_{i}^{\AA'}(\AA)$ is a vector of leverage score overestimates, i.e.
$\tau_{i}(\AA) \le \uu_i$, and
\[
\sum_{i=1}^{m} \uu_{i}\leq\frac{3n}{\alpha}.
\]
\end{theorem}

\begin{proof} (Sketch of Lemma~\ref{lem:sparsify})
Instead of sampling $m / K$ column blocks, consider sampling $m / K$ individual columns
to form $\CChat$.
Theorem~\ref{thm:leverage_score_undersampling} gives that the total leverage scores
of all the columns of $\BB$ w.r.t. $\CChat$ is at most $O(nK r^2)$.

As $\CC$ is formed by taking blocks instead of columns, we have
\[
\XX + \CC \CC^{\dg} \pgeq \XX + \CChat \CChat^{\dg},
\]
so the individual leverage scores computed w.r.t. $\CC$ (and in turn $\WW$)
sums up to less than the ones computed w.r.t. $\CChat$.

Let us use $\bb_{[i],j} \in \complex^{n}$ to denote the $j^{\text{th}}$ \emph{column}
of the $\blk{\BB}{i}$ block column.
Note that 
\[
\BB_{[i]}
\BB_{[i]}^{\dg} = \sum_{j} \bb_{[i],j}\bb_{[i],j}^{\dg}.
\]
 Define
  $\tau_{[i], j}^{\XX + \CC \CC^{*}} = \mathrm{tr}\left(\bb_{[i],j}^{\dg}
    (\XX + \CC \CC^{*})^{-1}  \bb_{[i],j}\right) $.

We then have
\[
 \mathrm{tr}\left(\BB_{[i]}^{\dg} (\XX + \CC \CC^{*})^{-1} \BB_{[i]} \right)
 = \sum_{j} \tau_{[i], j}^{\XX + \CC \CC^{*}},
\]
so the total sum of the estimates we obtain is at most $O(n r^2 K)$,
which gives $\BBtil$
with $O(K n \log{n} \epsilon^{-2})$ blocks.

To approximate $\sum_{j} \tau_{[i], j}^{\XX + \CC \CC^{*}}$, we apply a standard technique given by \cite{spielman2011graph, avron2011effective, li2013iterative}. The rough idea is to write
\[
\tau_{[i], j}^{\XX + \CC \CC^{*}} = \norm{\CC^{*} (\XX + \CC \CC^{*})^{-1} \bb_{[i],j}}^2 + \norm{\sqrt{\XX} (\XX + \CC \CC^{*})^{-1} \bb_{[i],j}}^2
\]
By Johnson-Lindenstrauss Lemma, for a random Gaussian matrix $\GG$ with $\Theta(\log(n))$ rows, we know that
\[
\tau_{[i], j}^{\XX + \CC \CC^{*}} \approx_2 \norm{\GG \CC^{*} (\XX + \CC \CC^{*})^{-1} \bb_{[i],j}}^2 + \norm{\GG \sqrt{\XX} (\XX + \CC \CC^{*})^{-1} \bb_{[i],j}}^2
\]
for high probability. Since $\GG$ has $O(\log(n))$ rows, this can be approximated by applying the approximate inverse $\WW$ to $O(\log(n))$ vectors. Hence, step 3 runs in $O(d_{eval}(\frac{m}{K},n)+\log n)$
depth and $O(w_{eval}(\frac{m}{K},n) \log n +m\log n)$
work.
\end{proof}

We remark that the extra factor of $r^2$ can likely be improved
by modifying the proof of Theorem~\ref{thm:leverage_score_undersampling}
to work with blocks.
We omit this improvement for simplicity, especially since we're already
treating $r$ as a constant.




\section{Recursive Construction of Schur Complement Chains}\label{sec:recursive}

We now give the full details for the recursive algorithm that proves
the running times as stated in Theorem~\ref{thm:recursive}:
\recursive*

This argument can be viewed as a more sophisticated version of
the one in Section~\ref{sec:overviewAlg}.
The main idea is to invoke routines for reducing edges and vertices in
a slightly unbalanced recursion, where several steps of vertex
reductions take place before a single edge reduction.
The routines that we will call are:
\begin{enumerate}
\item \textsc{bDDSubset} given by
Lemma~\ref{lem:abddSubset} proven in Section~\ref{sec:absdd}
for finding a large set of $\alpha$-{\bdd} subset.
\item \textsc{ApproxSchur} given by
Lemma~\ref{lem:approxSchur:extended} proven in Section~\ref{sec:vertexSparsify}
that gives sparse approximations to Schur complements.
\item \textsc{Sparsify} given by Lemma~\ref{lem:sparsify} proven in
Section~\ref{sec:undersampling} that allows us to sparsify a {\bdd}
matrix by recursing on a uniform subsample of its non-zero blocks.
\end{enumerate}
An additional level of complication comes from the fact
that the approximation guarantees of our constructions
rely on gradually decreasing errors down the Schur complement chain.
This means that the density increases faster and larger
reduction factors $K$ are required as the iteration goes on.
 Pseudocode of our algorithm is given in
 Figure~\ref{fig:recursiveConstruct}.

\begin{figure}[ht]

\begin{algbox}

$(\MM^{(1)},\MM^{(2)},\cdots;F_1,F_2,\cdots)=\textsc{RecursiveConstruct}(\MM^{(0)})$
\begin{enumerate}
\item Initialize:
\begin{enumerate}
\item $i\leftarrow 0$, and
\item $k$ and $c$ are parameters to be decided.
\end{enumerate}
\item While $\MM^{(i)}$ has more than $\Theta(1)$ blocks,
\begin{enumerate}
\item If $i \text{ mod } k=0$, Then
\begin{enumerate}
\item $\MM^{(i)}\leftarrow\textsc{Sparsify}(\MM^{(i)},(i+8)^{-2}, 2^{2c k\log^{2}(i+1)}).$
\end{enumerate}
\item $F_{i  + 1} \leftarrow \textsc{bDDSubset} \left( \MM^{(i)}, 4 \right)$.
\item $\MM^{(i + 1)} \leftarrow \textsc{ApproxSchur}(\MM^{(i)}, F_{i + 1}, 4, (i+8)^{-2})$.
\item $i\leftarrow i+1$.
\end{enumerate}
\end{enumerate}
\end{algbox}

\caption{Pseudocode for Recursively Constructing Schur Complement Chains, 
the recursive calls happen through \textsc{Sparsify},
which constructs its own Schur Complement Chains to perform
the sparsification}

\label{fig:recursiveConstruct}

\end{figure}

Note that since $\MM^{(0)}$ may be initially dense, we
first make a recursive call to \textsc{Sparsify} before
computing the approximate Schur complements.

In our analysis, we will use $n^{(i)}$ to denote the number of non-zero
column/row blocks in $\MM^{(i)}$, and $m^{(i)}$ to denote the number
of non-zero blocks in $\MM^{(i)}$.
These are analogous to dimension and number of non-zeros in the matrix.
It is also useful to refer to the steps between
calls to \textsc{Sparsify} as phases.
Specifically, phase $j$ consists of iterations $(j - 1)k$ to $jk - 1$.
We will also use $K_j$ to denote the reduction factor used to
perform the sparsification at the start of phase $j$, aka.
\[
K_j = 2^{2c k\log^{2}((j - 1)k +1)}.
\]

A further technicality is that Lemma~\ref{lem:sparsify} require
a strictly positive definite block-diagonal part to facilitate the detection
of vectors in the null space of the sample.
We do so by padding $\MM$ with a small copy of the identity,
and will check below that this copy stays throughout the course
of this algorithm.
\begin{lemma}
For any {\bdd} matrix $\MM^{(0)}$ expressible as $\xi \id + \BB^{(0)} ( \BB^{(0)})^{\dg}$
for some $\xi > 0$ and unitary edge-vertex transfer matrix $\BB$,
all intermediate matrices $\MM'$ generated during the algorithm
$\textsc{RecursiveConstruct}(\MM^{(0)})$ can be expressed
as $\xi \id + \BB' ( \BB' )^{\dg}$.
\end{lemma}
\begin{proof}
There are two mechanisms by which this recursive algorithm generates
new matrices: through uniform sampling within \textsc{Sparsify}
and via \textsc{ApproxSchur}.
We will show inductively down the algorithmic calls that all matrices satisfy
this property.

Lemma~\ref{lem:sparsify} gives that this $\XX$ is preserved in the sample.

For calls to $\textsc{ApproxSchur}(\MM, C)$, the inductive hypothesis gives
that $\blk{\MM}{C, C}$ is expressible as $\MMhat + \xi \id_{C}$.
The last condition in Lemma~\ref{lem:approxSchur} then gives that the result
also has a $\xi \id_{C}$ part, which gives the inductive hypothesis for it as well.
\end{proof}
As the interaction between this padding and our recursion is minimal,
we will simply state the input conditions as $\MM = \xi \id + \BB \BB^{\dg}$
for some $\xi > 0$.
We will first bound the sizes of the matrices in the Schur complement
chain produced by $\textsc{RecursiveConstruct}$.
\begin{lemma}
\label{lem:recursiveSize}
For any {\bdd} matrix $\MM^{(0)} = \xi \id + \BB^{(0)} ( \BB^{(0)})^{\dg}$
for some $\xi > 0$,
the algorithm $\textsc{RecursiveConstruct}(\MM^{(0)})$ returns a Schur complement
chain $(\MM^{(1)},\MM^{(2)},\cdots;F_1,F_2,\cdots)$ such that:
\begin{enumerate}
\label{item:diagExtra}
\item $n^{(i)} \leq \beta ^{i} n^{(0)}$ for some absolute constant $\beta<1$.
\label{item:vertexBound}
\item The number of non-zero blocks in any iteration $i$ of phase $j$ is at most
  $2^{3ck\log^{2}(jk)} n^{((j - 1)k)}\log n^{((j - 1)k)}$.
\label{item:edgeBound}
\end{enumerate}
\end{lemma}

\begin{proof}

Lemma~\ref{lem:abddSubset} and the choice of $\alpha = 4$
ensures $\left|F_k\right|=\Omega(n^{(i)})$,
which means there is constant $\beta <1$ such that $n^{(i)} \leq \beta^{k} n$.
Furthermore, we do not increase vertex count at any point in this recursion,
and all recursive calls are made to smaller graphs.
Therefore, the recursion terminates.

Lemma \ref{lem:approxSchur:extended} shows that after computing
each approximate Schur complement,
\begin{eqnarray*}
m^{(i+1)} & = & O(m^{(i)}(i^{2}\log(i+8))^{O(\log(i+8))})\\
 & \leq & 2^{O\left(\log^{2}(i+1)\right)}m^{(i)}.
\end{eqnarray*}
Hence, by picking $c$ appropriately we can guarantee that 
the density increases by a factor of most $2^{c \log^{2}(i+1)}$ during each iteration.
This size increase is controlled by calls to \textsc{Sparsify}.
Specifically, Lemma~\ref{lem:sparsify} gives that at the start of phase $j$ we have:
\[
\frac{m^{((j - 1)k)}}{n^{((j - 1)k)}\log n^{((j - 1)k)}}\leq 2^{2c k\log^{2}((j - 1)k )}.
\]
Then, since we go at most $k$ steps without calling \textsc{Sparsify},
this increase in density can be bounded by:
\[
2^{2ck\log^{2}((j - 1) k)}2^{c\log^{2}(jk-1)}
	\cdots2^{c\log^{2}((j - 1)k + 1)} \leq 2^{3ck\log^{2}(jk)}.
\]
\end{proof}

\begin{lemma}
\label{lem:recursiveGuarantee}
For any {\bdd} matrix  $\MM^{(0)} = \xi \id + \BB^{(0)} ( \BB^{(0)})^{\dg}$
for some $\xi > 0$ that has $n$ blocks, the algorithm
$\textsc{RecursiveConstruct}(\MM^{(0)})$ returns a Schur Complement Chain
 $(\MM^{(1)},\MM^{(2)},\cdots;F_1,F_2,\cdots)$ whose
corresponding the linear operator $\WW$ satisfies
\[
\WW\approx_{O(1)}\left(\MM^{(0)}\right)^{-1}.
\]
Also, we can evaluate $\WW \bb$ in $O(\log^{2}n\log\log n)$ depth
and $2^{O(k \log^{2}k)}n\log n
$ work for any vector $\bb$.\end{lemma}

\begin{proof}

We first bound the quality of approximation
between $\WW$ and $\left( \MM^{(0)} \right)^{-1}$.
The approximate Schur complement $\MM^{(i+1)}$ was constructed
so that $\MM^{(i+1)}\approx_{(i+8)^{-2}}\schur{\MM^{(i)}}{F_i}$.
The other source of error, \textsc{Sparsify}, is called only for some $i$.
In those iterations, Lemma~\ref{lem:sparsify} guarantee
that $\MM^{(i)}$ changes only by $(i+8)^{-2}$ factor.
This means overall we have
$\MM^{(i+1)}\approx_{2(i+8)^{-2}}\schur{\MM^{(i)}}{F_i}$.
By Lemma \ref{lem:apply_chain}, we have:
\[
\WW\approx_{1/2+4\sum_{i}(i+8)^{-2}}\left(\MM^{(0)}\right)^{-1},
\]
and it can be checked that $\sum_{i}(i+8)^{-2}$ is a constant.

The cost of \textsc{ApplyChain} is dominated by the sequence of calls
to \textsc{Jacobi}.
As each $F_i$ is chosen to be $4$-{\bdd}, the number of iterations
required is $O( \log ( \epsilon_i )) = O( \log i)$.
As matrix-vector multiplications take $O(\log{n})$ depth, the total
depth can be bounded by.
\[
O\left(\sum_{i}\log i\log n^{(i)}\right)=O\left(\log^{2}n\log\log n\right)
\]

The total work of these steps depend on the number of non-zero
blocks, $m^{(i)}$.
Substituting the bounds from Lemma~\ref{lem:recursiveSize} into 
\textsc{Jacobi} gives a total of:
\[
O\left(\sum_{i}2^{3ck\log^{2}(i+k)}n^{(i)}\log n^{(i)}\log i\right)
\leq 2^{O(k\log^{2}k)}n\log n
\]
where the inequality follows from the fact that the
$n^{(i)}$s are geometrically decreasing.
\end{proof}

This allows us to view the additional overhead of \textsc{Sparsify} as a black
box, and analyze the total cost incurred by the non-recursive parts of
\textsc{RecursiveConstruct} during each phase.

\begin{lemma}
\label{lem:recursivePhase}
The total cost of \textsc{RecursiveConstruct} during phase $j$,
including the linear system solves made by $\textsc{Sparsify}$
at iteration $(j - 1) k$ (but not its recursive invocation to \textsc{RecursiveConstruct}) is
\[
2^{O(k\log^{2}(jk))}n^{((j - 1)k)}\log^2 n^{((j - 1)k)}
\]
in work and
\[
O\left(k\log(jk)\log^2 n^{((j - 1)k)}\right)+O\left(\log^{2}n^{((j - 1)k)}\log\log n^{((j - 1)k)}\right)
\]
in depth.
\end{lemma}

\begin{proof}
Let $m^{(i)}$ and $n^{(i)}$ be the number of non zeros and variables
in $\MM^{(i)}$ before the $\textsc{Sparsify}$ call if there is.
Lemmas~\ref{lem:abddSubset}~and~\ref{lem:approxSchur:extended} show
that the $i^{th}$ iteration takes $O(m^{(i)}+m^{(i+1)})$
work and $O(\log i\log n^{(i)})$ depth. 
By Lemma \ref{lem:recursiveSize}
 the cost during these iterations excluding
the call to \textsc{Sparsify} is:
\begin{eqnarray*}
\sum_{i = (j-1)k}^{jk - 1}
O\left(m^{(i)}+m^{(i+1)}\right) & \leq &
\sum_{i = (j-1)k}^{jk - 1}
2^{O(k\log^{2}i)}n^{(i)}\log n^{(i)}\\
 & \leq & 2^{O(k\log^{2}(jk))}n^{((j - 1)k)}\log n^{((j - 1)k)}
\end{eqnarray*}
work and
\begin{eqnarray*}
\sum_{i = (j-1)k}^{jk - 1}
O\left(\log i\log n^{(i)}\right) & \leq & O(k\log((j -1)k)\log n^{((j - 1)k)})
\end{eqnarray*}
depth.

We now consider the call to \textsc{Sparsify} made during iteration $(j-1)k$.
Access to a fast solver for the sampled {\bdd} matrix is
obtained via recursive calls to $\textsc{RecursiveConstruct}$.
The guarantees of the chain given by Lemma \ref{lem:recursiveGuarantee}
above means each solve takes depth
\[
d_{eval} = \log^{2}n^{((j-1)k)}\log\log n^{((j-1)k)},
\]
and work
\[
w_{eval} = n^{((j-1)k)}\log n^{((j-1)k)}
\]
Incorporating these parameters into Lemma \ref{lem:sparsify} allows
us to bound the overhead from these solves by
\[
2^{O(k \log^{2}(jk)))} n^{((j-1)k)}\log^2 n^{((j-1)k)}
\]
work and
\[
O\left( k\log\left(jk\right)\log n^{(jk)}\right)
+ O\left(\log^{2}n^{((j - 1)k)}\log\log \left( n^{((j - 1)k)} \right) \right)
\]
depth.
\end{proof}

Note that at the end of the $j^{th}$ phase, the time required to construct an extra 
Schur complement chain for the $\textsc{Sparsify}$ call is less than the
remaining cost after the $j^{th}$ phase. This is the reason why we
use $2^{2ck\log^{2}(i + 1)}$ as the reduction factor for the $\textsc{Sparsify}$
call. The following theorem takes account for the recursive call and
show the total running time for the algorithm. 

\begin{lemma}
\label{lem:recursiveCost}
With high probability, $\textsc{RecursiveConstruct}\left(\MM^{(0)}\right)$
takes $2^{O(\log n/k)}$ depth and $m\log n+2^{O(k\log^{2}k)}n\log^2 n$
work.
\end{lemma}

\begin{proof}


Lemma \ref{lem:recursiveSize} shows that the call to \textsc{Sparsify}
at the start of each phase $j$ requires the construction of an extra
Schur complement chain on a matrix with $n^{((j - 1)k)}$ row/column blocks
and at most $2^{ck\log^{2}((j - 1)k)}n^{((j - 2)k)}\log n^{((j - 2)k)}$ non-zeros blocks.
The guarantees of Lemma~\ref{lem:sparsify} gives that the number of
non-zero block in the sparsified matrix is at most
\[
C 2^{2ck\log^{2}((j - 1)k)}n^{((j - 1)k)}\log n^{((j - 1)k)}
\]
for some absolute constant $C$.
Therefore the cost of this additional call is less than the cost
of constructing the rest of the phases during the construction process.
Therefore, every recursive call except the first one
can be viewed as an extra branching factor of $2$ at each subsequent phase.

Depth can be bounded by the total number of recursive invocations
to \textsc{RecursiveConstruct}.
The fact that $n^{(i)}$ is geometrically decreasing means
there are $O(\log n/k)$ phases.
Choosing $k$ so that $k\log^{2}k=o(\log n)$ gives a depth of:
\begin{eqnarray*}
 &  & \sum_{j=1}^{O(\log n/k)}2^{j}\left(k\log(jk)\log n^{(jk)}+\log^{2}n^{(jk)}\log\log n^{(jk)}\right)\\
 & = & 2^{O(\log n/k)}O\left(k\log\log n\log n^{(last)}+\log^{2}n^{(last)}\log\log n^{(last)}\right)\\
 & = & 2^{O(\log n/k)}O\left(k^{2}\log\log n+k^{2}\log k\right)\\
 & = & 2^{O(\log n/k)}k^{2}\log\log(n)\\
 & = & 2^{O(\log n/k)}.
\end{eqnarray*}

For bounding work, we start with the sparse case
since all intermediate matrices generated during the
construction process have density at most $2^{O(k \log^2{k})}$.
In this case, the extra branching factor of $2$ at each phase
can be accounted for by weighting the cost of iteration $j$ by $2^{j}$, giving:
\begin{eqnarray*}
 &  & \sum_{j=1}^{O(\log n/k)}2^{j}
 	\left(2^{O(k \log^{2}(jk))}n^{(jk)}\log^2 n^{(jk)}\right)\\
 & = & 2^{O(k\log^{2}k)}n\log^2 n.
\end{eqnarray*}

For the dense case, the first recursive call to $\textsc{Sparsify}(\MM^{(0)}, O(1), 2^{2ck})$
is made to a graph whose edge count is much less.
This leads to a geometric series, and an overhead of $O(m \log{n})$ work at
each step.
As this can happen at most $O(\log{n})$ times, it gives an additional factor
of $O(\log{n})$ in depth, which is still $2^{O(\log{n} / k)}$.
The work obeys the recurrence:
\[
W(m)_{dense} \leq
\begin{cases}
 2^{O(k\log^{2}k)}n \log^2{n} & \qquad \text{if } m \leq  2^{O(k\log^{2}k)}n\log{n}, \text{ and}\\
W_{dense}\left(m / 2\right)  + m \log{n}  +  2^{O(k\log^{2}k)}n \log^2{n}& \qquad \text{otherwise}.
\end{cases}
\]
which solves to:
\[
W_{dense}(m) \leq O(m\log n)+2^{O(k\log^{2}k)}n\log^2 n.
\]
\end{proof}

To obtain Theorem~\ref{thm:recursive}, we simply choose an appropriate initial
padding and set the parameter $k$.

\begin{proof}(of Theorem~\ref{thm:recursive})
Lemma~\ref{clm:LZ3L_new} allows us to solve the linear system
$\MM + \epsilon \mu \id$ instead.
Invoking Lemma~\ref{lem:recursiveCost} with $k=\log\log\log n$
gives a depth of $2^{O( \log{n} / \log\log\log{n} )}$,
and work of $m \log{n} + 2^{O(\log\log\log{n} \log\log\log\log^4{n})}n \log^2{n}$.
The depth term can be simplified to $n^{o(1)}$ while the work term
can be simplified to $2^{O(\log\log\log^2{n})} = \log^{o(1)}n$.
\end{proof}

\section{Weighted Expander Constructions}
\label{sec:weightedExp}

\def\edg#1{\pmb{\boldsymbol{(}} #1 \pmb{\boldsymbol{)}}_2}
\def\edgu#1{\pmb{\boldsymbol{(}} #1 \pmb{\boldsymbol{)}}}

In this section, we give a linear time algorithm for computing linear
  sized spectral sparsifiers of complete and bipartite
  product demand graphs.
These routines give Lemmas~\ref{lem:weightedBipartiteExpander}~and~\ref{lem:weightedExpander},
which are crucial for controlling the densities of intermediate matrices
in the spectral vertex sparsification algorithm from Section~\ref{sec:vertexSparsify}.
Recall that the \textit{product demand graph} with vertex set $V$ and demands $\dd : V \rightarrow \mathbb{R}_{> 0}$
  is the complete graph
  in which the weight of edge $(u,v)$ is the product $d_{u} d_{v}$.
Similarly, the \textit{bipartite demand graph} with vertex set $U \union V$
  and demands $\dd : U \union V \rightarrow \mathbb{R}_{> 0}$ is the
  complete bipartite graph on which the weight of the edge $(u,v)$ is the product $d_{u} d_{v}$.
Our routines are based on reductions to the unweighted, uniform case.
In particular, we
\begin{itemize}
\item [1.] Split all of the high demand vertices into many vertices that all have the same demand.
  This demand will still be the highest.

\item [2.] Given a graph in which almost all of the vertices have the same highest demand,
  we \begin{itemize}
\item [a.] drop all of the edges between vertices of lower demand,
\item [b.] replace the complete graph between the vertices of highest demand with an expander, and
\item [c.] replace the bipartite graph between the high and low demand vertices with
  a union of stars.
\end{itemize}
\item [3.] To finish, we merge back together the vertices that split off from each original vertex.
\end{itemize}

We start by showing how to construct the expanders that we need for step (2b).
We state formally and analyze the rest of the algorithm for the
  complete case in the following two sections.
We explain how to handle the bipartite case in Section \ref{subsec:bipartite}.

Expanders give good approximations to unweighted complete graphs,
  and our constructions will use the spectrally best expanders---Ramanujan graphs.
These are defined in terms of the eigenvalues of their adjacency matrices.
We recall that the adjacency matrix of every $d$-regular graph has eigenvalue $d$
  with multiplicity $1$ corresponding to the constant eigenvector.
If the graph is bipartite, then it also has an eigenvalue of $-d$ corresponding
  to an eigenvector that takes value $1$ on one side of the bipartition and $-1$
  on the other side.
These are called the \textit{trivial} eigenvalues. 
A $d$-regular graph is called a Ramanujan graph if all of its non-trivial eigenvalues
  have absolute value at most $2 \sqrt{d-1}$.
Ramanujan graphs were constructed independently by Margulis~\cite{Margulis88}
  and Lubotzky, Phillips, and Sarnak~\cite{LPS}.
The following theorem and proposition summarizes part of their results.

\begin{theorem}\label{thm:LPS}
Let $p$ and $q$ be unequal primes congruent to $1$ modulo 4.
If $p$ is a quadratic residue modulo $q$, then there is a non-bipartite
  Ramanujan graph of degree $p+1$ with $q^{2} (q-1)/2$ vertices.
If $p$ is not a quadratic residue modulo $q$, then there is a bipartite
  Ramanujan graph of degree $p+1$ with $q^{2} (q-1)$ vertices.
\end{theorem}

The construction is  explicit.

\begin{proposition}\label{pro:LPS}
If $p < q$, then
  the graph guaranteed to exist by Theorem~\ref{thm:LPS} can be constructed in
  parallel depth $O (\log n)$ and work $O (n)$, where $n$ is its number of vertices.
\end{proposition}
\begin{proof}[Sketch of proof.]
When $p$ is a quadratic residue modulo $q$, the graph is a Cayley graph of
  $PSL (2,Z/qZ)$.
In the other case, it is a Cayley graph of $PGL (2,Z/qZ)$.
In both cases, the generators are determined by the $p+1$ solutions
  to the equation $p = a_{0}^{2} + a_{1}^{2} + a_{2}^{2} + a_{3}^{2}$
  where $a_{0} > 0$ is odd and $a_{1}, a_{2}$, and $a_{3}$ are even.
Clearly, all of the numbers $a_{0}$, $a_{1}$, $a_{2}$ and $a_{3}$
  must be at most $\sqrt{p}$.
So, we can compute a list of all sums $a_{0}^{2} + a_{1}^{2}$
  and all of the sums $a_{2}^{2} + a_{3}^{2}$
  with work $O (p)$, and thus a list of all $p+1$
  solutions with work $O (p^{2}) < O (n)$.

As the construction requires arithmetic modulo $q$, it is convenient
  to compute the entire multiplication table modulo $q$.
This takes time $O (q^{2}) < O (n)$.
The construction also requires the computation of a square root of $-1$
  modulo $q$, which may be computed from the multiplication table.
Given this data, the list of edges attached to each vertex of the graph
  may be produced using linear work and logarithmic depth.
\end{proof}

For our purposes, there are three obstacles to using these graphs:
\begin{itemize}
\item [1.] They do not come in every degree.
\item [2.] They do not come in every number of vertices.
\item [3.] Some are bipartite and some are not.
\end{itemize}
We handle the first two issues by observing that the primes
  congruent to 1 modulo 4 are sufficiently dense.
To address the third issue, we give a procedure to convert a non-bipartite expander into a bipartite expander, and \textit{vice versa}.

An upper bound on the gaps between consecutive primes congruent to 1 modulo 4 can
  be obtained from the following theorem of Tchudakoff.

\begin{theorem}[\cite{Tchudakoff}]
For two integers $a$ and $b$, let
 $p_{i}$ be the $i$th prime congruent to $a$ modulo $b$.
For every $\epsilon > 0$,
\[
p_{i+1} - p_{i} \leq O (p_{i}^{3/4 + \epsilon }).
\]
\end{theorem}

\begin{corollary}\label{cor:tchudakoff}
There exists an $n_{0}$ so that for all $n \geq  n_{0}$
  there is a prime congruent to 1 modulo 4 between $n$ and $2 n$.
\end{corollary}

We now explain how we convert between bipartite and non-bipartite expander graphs.
To convert a non-bipartite expander into a bipartite expander, we take its double-cover.
We recall that if $G = (V,E)$ is a graph with adjacency matrix $\AA$, then its double-cover
  is the graph with adjacency matrix
\[
  \begin{pmatrix}
0 & \AA \\
\AA^{T} & 0
\end{pmatrix}.
\]
It is immediate from this construction that the eigenvalues of the adjacency matrix
  of the double-cover
  are the union of the eigenvalues of $\AA$ with the eigenvalues of $-\AA$.
\begin{proposition}\label{pro:doubleCover}
Let $G$ be a connected, $d$-regular graph in which all matrix eigenvalues
  other than $d$ are bounded in absolute value by $\lambda$.
Then, all non-trivial adjacency matrix eigenvalues of the double-cover of $G$
  are also bounded in absolute value by $\lambda$.
\end{proposition}

To convert a bipartite expander into a non-bipartite expander, we will simply
  collapse the two vertex sets onto one another.
If $G = (U \union V, E)$ is a bipartite graph, 
  we specify how the vertices of $V$ are mapped onto $U$ by a permutation $\pi : V \rightarrow U$.
We then define the \textit{collapse} of $G$ induced by $\pi$ 
  to be the graph with vertex set $U$ 
  and edge set
\[
  \setof{  (u, \pi (v)) : (u,v) \in E }.
\]
Note that the collapse will have self-loops at vertices $u$ for which $(u,v) \in E$
  and $u = \pi (v)$.
We assign a weight of $2$ to every self loop.
When a double-edge would be created, that is when $(\pi (v), \pi^{-1} (u))$ is also an edge in the graph,
  we give the edge a weight of $2$.
Thus, the collapse can be a weighted graph.

\begin{proposition}\label{pro:collapse}
Let $G$ be a $d$-regular bipartite graph with all non-trivial adjacency matrix eigenvalues
  bounded by $\lambda$, and let $H$ be a collapse of $G$.
Then, every vertex in $H$ has weighted degree $2d$ 
   and all adjacency matrix eigenvalues of $H$ other than $d$ are bounded in absolute value by $2 \lambda$.
\end{proposition}
\begin{proof}
To prove the bound on the eigenvalues, let $G$ have adjacency matrix
\[
\begin{pmatrix}
0 & \AA \\
\AA^{T} & 0
\end{pmatrix}.
\]
After possibly rearranging rows and columns, we may assume that
  the adjacency matrix of the collapse is given by
\[
  \AA + \AA^{T}.
\]
Note that the self-loops, if they exist, correspond to diagonal entries of value $2$.
Now, let $\xx$ be a unit vector orthogonal to the all-1s vector.
We have
\[
  \xx^{T} (\AA + \AA^{T}) \xx
=
\begin{pmatrix}
\xx
\\
\xx 
\end{pmatrix}^{T}
\begin{pmatrix}
0 & \AA \\
\AA^{T} & 0
\end{pmatrix}
\begin{pmatrix}
\xx
\\
\xx 
\end{pmatrix}
\leq 
\lambda \norm{
\begin{pmatrix}
\xx
\\
\xx 
\end{pmatrix}
}^{2}
\leq
2 \lambda ,
\]
as the vector $[\xx ;\xx]$ is orthogonal to the eigenvectors of the trivial
  eigenvalues of the adjacency matrix of $G$.
\end{proof}

We now state how bounds on the eigenvalues of the adjacency matrices of graphs
  lead to approximations of complete graphs and complete bipartite graphs.

\begin{proposition}\label{pro:expanderApprox}
Let $G$ be a graph with $n$ vertices, possibly with self-loops and weighted edges,
  such that every vertex of $G$  has weighted degree $d$ and
  such that all non-trivial eigenvalues of the adjacency matrix of $G$
  have absolute value at most $\lambda \leq d/2$.
If $G$ is not bipartite, then
  $(n/d) \LL_{G}$ is an $\epsilon$-approximation of $K_{n}$ for $\epsilon = (2 \ln 2) (\lambda)/d$.
If $G$ is bipartite, then
  $(n/d) \LL_{G}$ is an $\epsilon$-approximation of $K_{n,n }$ for $\epsilon = (2 \ln 2) (\lambda)/d$.
\end{proposition}
\begin{proof}
Let $\AA$ be the adjacency matrix of $G$.
Then,
\[
  \LL_{G} = d \II  - \AA .
\]

In the non-bipartite case, we observe that all of the non-zero eigenvalues
  of $\LL_{K_{n}}$ are $n$,
  so for all vectors $x$ orthogonal to the constant vector,
\[
  x^{T}  \LL_{K_{n}} x = n x^{T}x.
\]
As all of the non-zero eigenvalues of $\LL_{G}$
  are between $d - \lambda$ and $d + \lambda$,
for all vectors $x$ orthogonal to the constant vector
\[
  n \left(1-\frac{\lambda }{d} \right)  x^{T} x 
\leq   x^{T}  (n/d)\LL_{G} x
\leq 
  n \left(1+\frac{\lambda }{d} \right)  x^{T} x.
\]
Thus,
\begin{equation*}
 \left(1-\frac{\lambda }{d} \right) \LL_{K_{n}} 
\pleq \LL_{G} \pleq 
 \left(1+\frac{\lambda }{d} \right) \LL_{K_{n}} .
\end{equation*}

In the bipartite case, we naturally assume that the bipartition is the same in both $G$ and $K_{n,n}$.
Now, let $\xx$ be any vector on the vertex set of $G$.
Both the graphs $K_{n,n}$ and $(n/d) G$ have Laplacian matrix eigenvalue
  $0$ with the constant eigenvector, and eigenvalue $2 n$ with eigenvector
  $[\bvec{1};-\bvec{1}]$.
The other eigenvalues of the Laplacian of $K_{n,n}$ are $n$, while the
  other eigenvalues of the Laplacian of $(n/d) G$ are between
\[
  n \left(1 - \frac{\lambda}{d} \right)
\quad \text{and} \quad 
  n \left(1 + \frac{\lambda}{d} \right).
\]
Thus,
\[
 \left(1-\frac{\lambda }{d} \right) \LL_{K_{n,n}} 
\pleq \LL_{G} \pleq 
 \left(1+\frac{\lambda }{d} \right) \LL_{K_{n,n}} .
\]

The proposition now follows from our choice of $\epsilon$, which guarantees that
\[
  e^{-\epsilon} \leq 1 - \lambda /d
\quad \text{and} \quad 1 + \lambda /d
\leq e^{\epsilon},
\]
provided that $\lambda /d \leq  1/2$.
\end{proof}

\begin{lemma}
\label{lem:explicitExpanders}
There are algorithms that on input $n$ and $\epsilon > n^{-1/6}$
   produce a graph having $O (n/\epsilon^{2})$ edges that is an
  $O (\epsilon)$ approximation of $K_{n'}$ or $K_{n',n'}$
  for some $n \leq n' \leq 8n$.
These algorithms run in $O (\log n)$ depth and $O (n / \epsilon^{2})$ work.
\end{lemma}
\begin{proof}
We first consider the problem of constructing an approximation of $K_{n', n'}$.
By Corollary~\ref{cor:tchudakoff} there
  is a constant $n_{0}$ so that if $n > n_{0}$, then
  there is a prime $q$ that is equivalent to
  $1$ modulo $4$ so that $q^{2} (q-1)$ is between   and $n$ and $8 n$.
Let $q$ be such a prime and let $n' = q^{2} (q-1)$.
Similarly, for $\epsilon$ sufficiently small, there is a prime $p$
  equivalent to $1$ modulo $4$ that is between
  $\epsilon^{-2}/2$ and $\epsilon^{-2}$.
Our algorithm should construct the corresponding Ramanujan graph, as described
  in Theorem~\ref{thm:LPS} and Proposition~\ref{pro:LPS}.
If the graph is bipartite, then Proposition~\ref{pro:expanderApprox} tells us
  that it provides the desired approximation of $K_{n',n'}$.
If the graph is not bipartite, then we form its double cover to obtain
   a bipartite graph and use Proposition~\ref{pro:doubleCover} 
  and Proposition~\ref{pro:expanderApprox} to see that it provides the desired
  approximation of $K_{n',n'}$.

The non-bipartite case is similar, except that we require a prime $q$
  so that $q^{2} (q-1)/2$ is between $n$ and $8 n$, and we use
  a collapse to convert a bipartite expander to a non-bipartite one,
  as analyzed in Proposition~\ref{pro:collapse}.
\end{proof}

For the existence results in Section~\ref{sec:existence}, we just need to know that
  there exist graphs of low degree that are good approximations of complete graphs.
We may obtain them from the recent theorem of Marcus, Spielman and Srivastava
  that there exist bipartite Ramanujan graphs of every degree and number of vertices
  \cite{IF4}.

\begin{lemma}
\label{lem:existExpanders}
For every integer $n$ and even integer $d$, 
  there is a weighted graph on $n$ vertices of degree at most
  $d$  that is a $4 / \sqrt{d} $ approximation
  of $K_{n}$.
\end{lemma}
\begin{proof}
The main theorem of \cite{IF4} tells us that there is a bipartite Ramanujan
  graph on $2n$ vertices of degree $k$ for every $k \leq n$.
By Propositions \ref{pro:collapse} and \ref{pro:expanderApprox},
  a collapse of this graph
  is a weighted graph of degree at most $2k$
  that is a $(4 \ln 2)/\sqrt{k}$ approximation of $K_{n,n}$.
The result now follows by setting $d = 2k$.
\end{proof}

\subsection{Sparsifying Complete Product Demand Graphs}
\label{subsec:complete}

In the rest of the section, we will adapt these expander constructions
to weighted settings.
We start with product demand graphs.

\weightedExp*

Our algorithm for sparsifying complete product demand graphs begins by
  splitting the vertices of highest demands into many vertices.
By \textit{splitting} a vertex, we mean replacing it by many
  vertices whose demands sum to its original demand.
In this way, we obtain a larger product demand graph.
We observe that we can obtain a sparsifier of the original graph by
  sparsifying the larger graph, and then collapsing back together
  the vertices that were split.

\begin{proposition}\label{pro:splitProduct}
Let $G$ be a product demand graph with vertex set 
  $\setof{1, \dots ,n}$ 
  and demands $\dd$,
  and let $\Ghat = (\Vhat, \Ehat )$ be a product demand graph with
  demands $\ddhat$.
If there is a partition of $\Vhat$ into sets $S_{1}, \dots , S_{n}$  
  so that for all $i \in V$, $\sum_{j \in S_{i}} \hat{d}_{j} = d_{i}$,
  then $\Ghat$ is a \textit{splitting} of $G$ and there is a matrix
  $\MM$ so that
\[
  \LL_{G} = \MM \LL_{\Ghat} \MM^{T}.
\]
\end{proposition}
\begin{proof}
The $(i,j)$ entry of matrix $\MM$ is $1$ if and only if $j \in S_{i}$.
Otherwise, it is zero.
\end{proof}

We now show that we can sparsify $G$ by sparsifying $\Ghat$.
 
\begin{proposition}
\label{pro:collapseLoewner}
Let $\Ghat_{1}$ and $\Ghat_{2}$ be graphs on the same vertex set $\Vhat$ such
  that $\Ghat _{1}\approx_{\epsilon}\Ghat _{2}$ for some $\epsilon$.
Let $S_{1}, \dots , S_{n}$ be a partition of $\Vhat$, and let $G_{1}$
  and $G_{2}$ be the graphs obtained by collapsing together all the
  vertices in each set $S_{i}$ and eliminating any self loops that are
  created.
Then
\[
G_{1}\approx_{\epsilon}G_{2}.
\]
\end{proposition}
\begin{proof}
Let $\MM$ be the matrix introduced in Proposition \ref{pro:splitProduct}.
Then,
\[
  \LL_{G_{1}} = \MM \LL_{\Ghat_{1}} \MM^{T} \quad \text{and} \quad 
  \LL_{G_{2}} = \MM \LL_{\Ghat_{2}} \MM^{T}.
\]
The proof now follows from   Fact \ref{fact:orderCAC}.
\end{proof}

For distinct vertices $i$ and $j$, we let $\edgu{i,j}$ denote the graph with an edge of weight $1$ between vertex $i$ and vertex $j$.
If $i = j$, we let $\edgu{i,j}$ be the empty graph.
With this notation, we can express the product demand graph 
  as
\[
  \sum_{i < j} d_{i} d_{j} \edgu{i,j}
=
  \frac{1}{2} \sum_{i,j \in V} d_{i} d_{j}\edgu{i,j}.
\]

This notation also allows us to precisely express our algorithm for sparsifying
  product demand graphs.

\begin{algbox}
$G'=\textsc{WeightedExpander}(\dd,\epsilon)$ 
\begin{enumerate}

\item Let $\nhat$ be the least integer greater than
  $2 n / \epsilon^{2}$ such that the algorithm described in Lemma \ref{lem:explicitExpanders}
  produces an $\epsilon$-approximation of $K_{\nhat}$.

\item Let $t = \frac{\sum_{k}d_{k}}{\nhat}$.

\item Create a new product demand graph $\Ghat$ with demand vector $\hat{\dd}$
 by 
splitting each vertex $i$ into a set of $\ceil{d_{i}/t}$ vertices, $S_i$:
\begin{enumerate}
\item $\floor{d_{i}/t}$ vertices with demand $t$.
\item one vertex with demand $d_{i} - t \floor{d_{i}/t}$.
\end{enumerate} 

\item Let $H$ be a set of $\nhat$ vertices in $\Ghat$ with demand $t$,
  and let $L$ contain the other vertices.  Set $k = \sizeof{L}$.



\item 
Partition $H$ arbitrarily into sets $V_{1}, \dots , V_{k}$, so that
  $\sizeof{V_{i}} \geq \floor{\nhat / k}$ for all $1 \leq i \leq k$.
  
\item 
Use the algorithm described in Lemma \ref{lem:explicitExpanders} to
  produce $\tilde{K}_{HH}$, an $\epsilon$-approximation of the complete graph on $H$.
Set 
\[
\Gtil = t^2 \tilde{K}_{HH} + \sum_{l \in L} 
  \frac{\sizeof{H}}{\sizeof{V_{l}}} \sum_{h \in V_{l}} 
     \dhat_{l} \dhat_{h} \edgu{l,h}.
\]

\item Let $G'$ be the graph obtained by collapsing together all vertices
  in each set $S_{i}$.
\end{enumerate}
\end{algbox}



This section and the next are devoted to the analysis of this algorithm.
Given Proposition~\ref{pro:collapseLoewner}, we just need to show that 
  $\Gtil$ is a good approximation to $\Ghat$.

\begin{proposition}\label{pro:numVertsAfterSplit}
The number of vertices in $\Ghat$ is at most $n + \nhat$.
\end{proposition}
\begin{proof}
The number of vertices in $\Ghat$ is
\[
  \sum_{i \in V} \ceil{d_{i} / t}
\leq
  n +   \sum_{i \in V} d_{i} / t
=
  n + \nhat .
\]
\end{proof}

So, $k \leq n$ and $\nhat \geq 2 k / \epsilon^{2}$.
That is, $\sizeof{H} \geq 2 \sizeof{L} / \epsilon^{2}$.
In the next section, we prove the lemmas that show that for these special product demand graphs  $\Ghat $ in which
  almost all weights are the maximum,
  our algorithm produces a graph $\Gtil$ that is a good approximation of $\Ghat$.

\begin{proof}(of Lemma~\ref{lem:weightedExpander})
The number of vertices in the graph $\Ghat$ will be between
  $n + 2 n / \epsilon^{2}$ and $n + 16 n / \epsilon^{2}$.
So, the algorithm described in Lemma \ref{lem:explicitExpanders}  will take
  $O (\log n)$ depth and $O (n / \epsilon^{4})$ work to produce an
  $\epsilon$ approximation of the complete graph on $\nhat$ vertices.
This dominates the computational cost of the algorithm.

Proposition
  \ref{pro:collapseLoewner} tells us that
  $G'$ approximates $G$ at least as well as $\Gtil$ approximates
  $\Ghat$.
To bound how well $\Gtil$ approximates $\Ghat$,
  we use two lemmas that are stated in the next section.
Lemma \ref{lemma:light_vertex_not_important} shows
  that
\[
  \Ghat_{HH} + \Ghat_{LH} \approx_{O(\epsilon^{2})} \Ghat .
\]
Lemma \ref{lem:replaceLH} shows that 
\[
 \Ghat_{HH} + \Ghat_{LH}
\approx_{4 \epsilon}
\Ghat_{HH} + \sum_{l \in L} 
  \frac{\sizeof{H}}{\sizeof{V_{l}}} \sum_{h \in V_{l}} 
     \dhat_{l} \dhat_{h} \edgu{l,h}.
\]
And, we already know that $t^2 \tilde{K}$ is an $\epsilon$-approximation of
  $\Ghat_{HH}$.
Fact \ref{frac:orderComposition} says that we can combine these three approximations to conclude that
  $\Gtil$ is an $O (\epsilon)$-approximation of $\Ghat$.

\end{proof}

\subsection{Product demand graphs with most weights maximal}
In this section, we consider product demand graphs in which almost all weights are the maximum.
For simplicity, we make a slight change of notation from the previous section.
We drop the hats, we let $n$ be the number of vertices in the product demand graph,
  and we order the demands so that
\[
  d_{1} \leq d_{2} \leq \dots \leq d_{k} \leq d_{k+1} = \dots = d_{n} = 1.
\]

We let $L = \setof{1, \dots , k}$ and $H = \setof{k+1, \dots , n}$
  be the set of low and high demand vertices, respectively.
Let $G$ be the product demand graph corresponding to $\dd$, and let
  $G_{LL}$, $G_{HH}$ and $G_{LH}$ be the subgraphs containing the
  low-low, high-high and low-high edges respectively.
We now show that little is lost by dropping the edges in $G_{LL}$
  when $k$ is small.

Our analysis will make frequent use of the
following Poincare inequality: 
\begin{lemma} \label{lemma:poincare}Let
$c \edgu{u,v}$ be an edge of weight $c$ and let $P$ be a path from
from $u$ to $v$ 
  consisting of edges of weights $c_{1},c_{2},\cdots,c_{k}$.
Then 
\[
c \edgu{u,v} \preceq c\left(\sum c_{i}^{-1}\right)P.
\]
\end{lemma}

As the weights of the edges we consider in this section are determined
  by the demands of their vertices,
  we introduce the notation
\[
  \edg{i,j} = d_{i} d_{j} \edgu{i,j}.
\]
With this notation, we can express the product demand graph 
  as
\[
  \sum_{i < j} \edg{i,j}
=
  \frac{1}{2} \sum_{i,j \in V} \edg{i,j}.
\]

\begin{lemma} \label{lemma:light_vertex_not_important}
If $\left|L\right|\leq\left|H\right|$,
then 
\[
G_{HH}+G_{LH}\approx_{3\frac{\left|L\right|}{\left|H\right|}} G.
\]
\end{lemma} 
\begin{proof}
The lower bound $G_{HH}+G_{LH}\preceq G_{HH}+G_{LH}+G_{LL}$
follows from $G_{LL}\succeq0$.

Using lemma \ref{lemma:poincare} and the assumptions $d_{l} \leq 1$
  for $l \in L$ and and $d_{h} = 1$ for $h\in H$, we derive for every $l_{1}, l_{2} \in L$,
\begin{align*}
\edg{l_{1}, l_{2}}
& =  \frac{1}{\left|H\right|^{2}}\sum_{h_{1},h_{2}\in H} \edg{l_{1}, l_{2}}\\
\intertext{\text{(by Lemma \ref{lemma:poincare})}}
& \preceq  \frac{1}{\left|H\right|^{2}}
  \sum_{h_{1},h_{2}\in H}d_{l_{1}}d_{l_{2}}
  \left(\frac{1}{d_{l_{1}}d_{h_{1}}}+\frac{1}{d_{h_{1}}d_{h_{2}}}+\frac{1}{d_{h_{2}}d_{l_{2}}}\right)
 \left(\edg{l_{1},h_{1}}+\edg{h_{1},h_{2}}+\edg{h_{2},l_{2}}\right)\\
 & \preceq 
 \frac{3}{\left|H\right|^{2}} \sum_{h_{1},h_{2}\in H}  \left(\edg{l_{1},h_{1}}+\edg{h_{1},h_{2}}+\edg{h_{2},l_{2}}\right)\\
 & = 
 \frac{3}{\left|H\right|}\sum_{h\in H} \left(\edg{l_{1},h}+\edg{l_{2},h}\right)+\frac{6}{\left|H\right|^{2}}G_{HH}.
\end{align*}
So,
\begin{align*}
G_{LL} & =
  \frac{1}{2} \sum_{l_{1}, l_{2} \in L} \edg{l_{1}, l_{2}}
\\
 & \preceq  \frac{1}{2}
  \sum_{l_{1},l_{2}}\left(\frac{3}{\left|H\right|}\sum_{h\in H}
  \left(\edg{l_{1},h}+\edg{l_{2},h}\right)+\frac{6}{\left|H\right|^{2}}G_{HH}\right)\\
 & =  \frac{3\left|L\right|}{\left|H\right|}G_{LH}+\frac{3\left|L\right|^{2}}{\left|H\right|^{2}}G_{HH}.
\end{align*}
The assumption $\sizeof{L} \leq \sizeof{H}$ then allows us to conclude
\[
G_{HH}+G_{LH}+G_{LL}\preceq\left(1+3\frac{\left|L\right|}{\left|H\right|}\right)\left(G_{HH}+G_{LH}\right).
\]
\end{proof} 

Using a similar technique, we will show that the edges between $L$ and $H$
  can be replaced by the union of a small number of stars.
In particular, we will partition the vertices of $H$ into $k$ sets,
  and for each of these sets we will create one star connecting
  the vertices in that set to a corresponding vertex in $L$.

We employ the following consequence of the Poincare inequality in Lemma \ref{lemma:poincare}.

\begin{lemma} \label{lemma:light_can_be_merge}
For any $\epsilon \leq 1$, $l \in L$ and $h_{1}, h_{2} \in H$,
\[
\epsilon \edg{h_{1},l}+ (1/2)\edg{h_{1},h_{2}}
\approx_{4 \sqrt{\epsilon}}
\epsilon \edg{h_{2},l}+ (1/2)\edg{h_{1},h_{2}}.
\]
\end{lemma} \begin{proof} 
By applying Lemma \ref{lemma:poincare} and
  recalling that $d_{h_{1}} = d_{h_{2}} = 1$ and $d_{l} \leq 1$, we compute
\begin{align*}
\edg{h_{1}, l}
 & \preceq   d_{h_{1}} d_{l} 
  \left(\frac{\sqrt{\epsilon}}{ d_{h_{1}} d_{h_{2}}}+\frac{1}{ d_{h_{2}} d_{l} }\right)
  \left(\frac{1}{\sqrt{\epsilon}}\edg{h_{1},h_{2}}+\edg{h_{2},l}\right)\\
 & \preceq  \frac{1+\sqrt{\epsilon}}{\sqrt{\epsilon}}\edg{h_{1},h_{2}}
    + (1+\sqrt{\epsilon})\edg{h_{2},l}
\\
 & \preceq  (1+\sqrt{\epsilon})\edg{h_{2},l}+\frac{2}{\sqrt{\epsilon}}\edg{h_{1},h_{2}}.
\end{align*}
Multiplying both sides by $\epsilon$ and adding $(1/2) \edg{h_{1}, h_{2}}$ then gives
\begin{align*}
\epsilon \edg{h_{1}, l} + (1/2) \edg{h_{1}, h_{2}}
& \pleq 
(1+\sqrt{\epsilon}) \epsilon \edg{h_{2}, l}
+
(2 \sqrt{\epsilon} + 1/2) \edg{h_{1},h_{2}}
\\
& \pleq
(1 + 4 \sqrt{\epsilon}) \left( \epsilon \edg{h_{2},l}+ (1/2)\edg{h_{1},h_{2}}  \right)
\\
& \pleq
e^{4 \sqrt{\epsilon }} \left( \epsilon \edg{h_{2},l}+ (1/2)\edg{h_{1},h_{2}}  \right).
\end{align*}
By symmetry, we also have
\[
\epsilon \edg{h_{2}, l} + (1/2) \edg{h_{1}, h_{2}}
\pleq 
e^{4 \sqrt{\epsilon }} \left( \epsilon \edg{h_{1},l}+ (1/2)\edg{h_{1},h_{2}}  \right).
\]
\end{proof}

\begin{lemma}\label{lem:replaceLH}
Recall that $L = \setof{1,\dots ,k}$ and let
  $V_{1}, \dots , V_{k}$ be a partition of $H = \setof{k+1, \dots , n}$
  so that $\sizeof{V_{l}} \geq s$ for all $l$.
Then,
\[
G_{HH} + G_{LH} 
\approx_{4 / \sqrt{s}}
G_{HH} + \sum_{l \in L} \frac{\sizeof{H}}{\sizeof{V_{l}}} \sum_{h \in V_{l}} \edg{l,h}.
\]
\end{lemma}
\begin{proof}
Observe that 
\[
G_{LH} = \sum_{l \in L} \sum_{h \in H} \edg{l,h}.
\]
For each $l \in L$, $h_{1} \in H$ and $h_{2} \in V_{l}$ we apply
  Lemma \ref{lemma:light_can_be_merge} to show that
\[
  \frac{1}{\sizeof{V_{l}}} \edg{l, h_{1}}
+ \frac{1}{2} \edg{h_{1}, h_{2}}
\approx_{4 / \sqrt{s}}
  \frac{1}{\sizeof{V_{l}}} \edg{l, h_{2}}
  + \frac{1}{2} \edg{h_{1}, h_{2}}.
\]
Summing this approximation over all $h_{2} \in V_{l}$
  gives
\[
 \edg{l, h_{1}}
+   \sum_{h_{2} \in V_{l}}
  \frac{1}{2} \edg{h_{1}, h_{2}}
\approx_{4 / \sqrt{s}}
  \sum_{h_{2} \in V_{l}}
\left(  \frac{1}{\sizeof{V_{l}}}
   \edg{l, h_{2}}
  + \frac{1}{2} \edg{h_{1}, h_{2}} \right)
.
\]
Summing the left-hand side of this this approximation
  over all $l \in L$ and $h_{1} \in H$ gives
\[
  \sum_{l \in L, h_{1} \in H}  \edg{l, h_{1}}
+ 
   \sum_{h_{2} \in V_{l}}
  \frac{1}{2} \edg{h_{1}, h_{2}}
=
  \sum_{l \in L, h_{1} \in H}  \edg{l, h_{1}}
+
  \frac{1}{2}
  \sum_{h_{1} \in H, l \in L}  
   \sum_{h_{2} \in V_{l}} \edg{h_{1}, h_{2}}
=
  G_{LH} + G_{HH}.
\]
On the other hand, the sum of the right-hand terms gives
\[
G_{HH} + 
  \sum_{l \in L, h_{1} \in H}
  \sum_{h_{2} \in V_{l}}
 \frac{1}{\sizeof{V_{l}}}
   \edg{l, h_{2}}
=
G_{HH} + 
  \sum_{l \in L}
  \sum_{h_{2} \in V_{l}}
 \frac{\sizeof{H}}{\sizeof{V_{l}}}
   \edg{l, h_{2}}.
\]
\end{proof}

\subsection{Weighted Bipartite Expanders}
\label{subsec:bipartite}

This construction extends analogously to bipartite product graphs.
The bipartite product demand graph of vectors $(\dd^{A},\dd^{B})$
is a complete bipartite graph whose weight between vertices $i\in A$
and $j\in B$ is given by $w_{ij}=d_{i}^{A}d_{j}^{B}$.
The main result that we will prove is:

\weightedBipartiteExp*

 Without loss of generality, we will assume
 $d_{1}^{A}\geq d_{2}^{A}\geq\cdots\geq d_{n^{A}}^{A}$
and $ d_{1}^{B}\geq d_{2}^{B}\geq\cdots\geq d_{n^{B}}^{B}$.
As the weights of the edges we consider in this section are determined
  by the demands of their vertices,
  we introduce the notation
\[
  \edg{i,j} = d^{A}_{i} d^{B}_{j} \edgu{i,j}.
\]

Our construction is based on a similar observation that if most vertices
on $A$ side have $d_{i}^{A}$ equaling to $d_{1}^{A}$ and most
vertices on $B$ side have $d_{i}^{B}$ equaling to $d_{1}^{B}$,
then the uniform demand graph on these vertices dominates the graph.

\begin{algbox}
$G'=\textsc{WeightedBipartiteExpander}(\dd^{A},\dd^{B},\epsilon)$ 
\begin{enumerate}

\item Let $n'=\max(n^{A},n^{B})$ and $\nhat$ be the least integer greater than
  $2 n' / \epsilon^{2}$ such that the algorithm described in
  Lemma \ref{lem:explicitExpanders} produces an $\epsilon$-approximation of $K_{\nhat ,\nhat}$.

\item Let $t^{A}=\frac{\sum_{k}d_{k}^{A}}{\nhat}$ and  $t^{B}=\frac{\sum_{k}d_{k}^{B}}{\nhat}$.

\item Create a new bipartite demand graph $\Ghat$ with demands
$\ddhat^{A}$ and $\ddhat^{B}$ follows:
\begin{enumerate}
  \item On the side $A$ of the graph, for each vertex $i$, create a subset $S_i$ consisting of $\ceil{d^{A}_{i}/t^A}$ vertices:
  \begin{enumerate}
  \item   $\floor{d^{A}_{i}/t^A}$ with demand $t^A$.
    \item  one vertex with demand $d^{A}_{i} - t^A \floor{d^{A}_{i}/t^A}$.
  \end{enumerate} 

\item Let $H^{A}$ contain $\hat{n}$ vertices of $A$ of with demand $t^{A}$, and let
   $L^{A}$ contain the rest.
Set $k^{A} = \sizeof{L^{A}}$.

\item Create the side $B$ of the graph with partition $H^{B}, L^{B}$
and demand vector $\ddhat^{B}$ similarly.
  \end{enumerate}

\item Partition $H^{A}$ into sets of size
$\sizeof{V^{A}_{i}} \geq \floor{\nhat / k^{A}}$, one corresponding
to each vertex $l \in L^{A}$.
Partition $V_{B}$ similarly.

\item 
Let $\tilde{K}_{H^{A} H^{B}} $ be a bipartite expander produced
by Lemma~\ref{lem:explicitExpanders} that $\epsilon$-approximates $K_{\hat{n}  \hat{n}}$, identified with the vertices $H^{A}$ and $H^{B}$.
  
 Set 
\[
\Gtil = t^{A} t^{B} \tilde{K} + \sum_{l \in L^{A}} 
  \frac{\sizeof{H^{B}}}{\sizeof{V^{B}_{l}}} \sum_{h \in V^{B}_{l}} 
     \dhat^{A}_{l} \dhat^{B}_{h} \edgu{l,h}
+ \sum_{l \in L^{B}} 
  \frac{\sizeof{H^{A}}}{\sizeof{V^{A}_{l}}} \sum_{h \in V^{A}_{l}} 
     \dhat^{B}_{l} \dhat^{A}_{h} \edgu{l,h}.
\]

\item Let $G'$ be the graph obtained by collapsing together all vertices 
  in each set $S^{A}_{i}$ and $S^{B}_{i}$.
\end{enumerate}
\end{algbox}

Similarly to the non-bipartite case, the Poincare inequality show that
the edges between low demand vertices can be completely omitted if
there are many high demand vertices which allows the demand routes
through high demand vertices.

\begin{lemma} \label{lemma:light_vertex_not_important2}Let $G$
be the bipartite product demand graph of the demand $(\dd_{i}^{A},\dd_{j}^{B})$.
Let $H^{A}$ a subset of vertices on $A$ side with demand higher
than the set of remaining vertices $L^{A}$ on $A$ side.
Define $H^{B},L^{B}$ similarly. Assume that $\sizeof{L^{A}}\leq\sizeof{H^{A}}$
and $\sizeof{L^{B}}\leq\sizeof{H^{B}}$, then 
\[
G_{H^{A}H^{B}}+G_{H^{A}L^{B}}+G_{L^{A}H^{B}}\approx_{3\max\left(\frac{\sizeof{L^{A}}}{\sizeof{H^{A}}},\frac{\sizeof{L^{B}}}{\sizeof{H^{B}}}\right)}G.
\]
\end{lemma} \begin{proof} 
The proof is analogous to Lemma~\ref{lemma:light_vertex_not_important},
but with the upper bound modified for bipartite graphs.

For every edge $l_A, l_B$, we embed it evenly into paths of
the form $l_A, h_B, h_A, l_B$ over all choices of $h_A$ and $h_B$.
The support of this embedding can be calculated using
Lemma~\ref{lemma:poincare}, and the overall accounting
follows in the same manner as Lemma~\ref{lemma:light_vertex_not_important}.
\end{proof}

It remains to show that the edges between low demand and high demand
vertices can be compressed into a few edges.
The proof here is also analogous to Lemma~\ref{lemma:light_can_be_merge}:
we use the Poincare inequality to show that all
demands can routes through high demand vertices.
The structure of the bipartite graph makes it helpful
to further abstract these inequalities via the following
Lemma for four edges.

\begin{lemma} \label{lem:light_can_be_merge_2}Let
$G$ be the bipartite product demand graph of the demand $(\dd_{i}^{A},\dd_{j}^{B})$.
Given $h_{A},l_{A}\in A$ and $h_{B,1},h_{B,2}\in B$. Assume that
$d_{h_{A}}^{A}=d_{h_{B,1}}^{B}=d_{h_{B,2}}^{B}\geq d_{l_{A}}^{A}$.
For any $\epsilon<1$ , we have 
\[
\epsilon\edg{l_{A},h_{B,1}}+\edg{h_{A},h_{B,2}}+\edg{h_{A},h_{B,1}}\approx_{3\sqrt{\epsilon}}\epsilon \edg{l_{A},h_{B,2}}+\edg{h_{A},h_{B,2}}+\edg{h_{A},h_{B,1}}.
\]
\end{lemma} \begin{proof} 
Using Lemma $\ref{lemma:poincare}$ and
$d_{h_{A}}^{A}=d_{h_{B,1}}^{B}=d_{h_{B,2}}^{B}\geq d_{l_{A}}^{A}$,
we have 
\begin{align*}
 &  \edg{l_{A},h_{B,1}} \\
& \preceq  d_{l_{A}}^{A} d_{h_{B,1}}^{B}\left(\frac{1}{ d_{l_{A}}^{A} d_{h_{B,2}}^{B}}+\frac{\sqrt{\epsilon}}{d_{h_{A}}^{A} d_{h_{B,2}}^{B}}+\frac{\sqrt{\epsilon}}{ d_{h_{A}}^{A} d_{h_{B,1}}^{B}}\right)\left(\edg{l_{A},h_{B,2}}+\frac{1}{\sqrt{\epsilon}}\edg{h_{A},h_{B,2}}+\frac{1}{\sqrt{\epsilon}}\edg{h_{A},h_{B,1}}\right)\\
 & \preceq  (1+2\sqrt{\epsilon})\edg{l_{A},h_{B,2}}+\frac{1+2\sqrt{\epsilon}}{\sqrt{\epsilon}}\edg{h_{A},h_{B,2}}+\frac{1+2\sqrt{\epsilon}}{\sqrt{\epsilon}}\edg{h_{A},h_{B,1}}.
\end{align*}
Therefore,
\begin{align*}
 &   \epsilon\edg{l_{A},h_{B,1}}+\edg{h_{A},h_{B,2}}+\edg{h_{A},h_{B,1}} \preceq  (1+3\sqrt{\epsilon})\left(\epsilon\edg{l_{A},h_{B,2}}+\edg{h_{A},h_{B,2}}+\edg{h_{A},h_{B,1}}\right).
\end{align*}
The other side is similar due to the symmetry.\end{proof} 

\begin{proof}(of Lemma~\ref{lem:weightedBipartiteExpander})
The proof is analogous to Lemma~\ref{lem:weightedExpander}.
After the splitting, the demands in $H^{A}$ are higher than the demands in $L^{A}$ and so is $H^{B}$ to $L^{B}$.
Therefore, Lemma \ref{lemma:light_vertex_not_important2} shows that
  that

\[
\Ghat_{H^{A}H^{B}}+\Ghat_{H^{A}L^{B}}+\Ghat_{L^{A}H^{B}} \approx_{3 \epsilon^{2}/2} \Ghat .
\]
By a proof analogous to Lemma \ref{lem:replaceLH}, one can use Lemma \ref{lem:light_can_be_merge_2} to show that 
\[
 \Ghat_{H^{A}H^{B}}+\Ghat_{H^{A}L^{B}}+\Ghat_{L^{A}H^{B}}
\approx_{O(\epsilon)}
\Ghat_{H^{A}H^{B}} + \frac{\sizeof{H^{B}}}{\sizeof{V^{B}_{l}}} \sum_{h \in V^{B}_{l}} 
     \dhat^{A}_{l} \dhat^{B}_{h} \edgu{l,h}
+ \sum_{l \in L^{B}} 
  \frac{\sizeof{H^{A}}}{\sizeof{V^{A}_{l}}} \sum_{h \in V^{A}_{l}} 
     \dhat^{B}_{l} \dhat^{A}_{h} \edgu{l,h}.
\]
And, we already know that $t^A t^B \tilde{K}$ is an $\epsilon$-approximation of
  $\Ghat_{H^{A}H^{B}}$.
Fact \ref{frac:orderComposition} says that we can combine these three approximations to conclude that
  $\Gtil$ is an $O (\epsilon)$-approximation of $\Ghat$.
%
 \end{proof}

\end{document}